\documentclass[10pt]{article} 
\usepackage{amsthm}
\theoremstyle{definition}
\newtheorem{definition}{Definition}[section]
\newtheorem{remark}[definition]{Remark}

\newtheorem{example}[definition]{Example}

\theoremstyle{plain}
\newtheorem{theorem}[definition]{Theorem}

\newtheorem{lemma}[definition]{Lemma}
\newtheorem{proposition}[definition]{Proposition}

\newtheorem{assumption}[definition]{Assumption}
\usepackage[english]{babel}
\usepackage{indentfirst}
\usepackage{fancyhdr}
\usepackage{graphicx}
\usepackage{newlfont}
\usepackage{amssymb} 
\usepackage{amsmath}
\usepackage{latexsym}
\usepackage{amsthm}
\usepackage{eucal}
\usepackage{eufrak}
\usepackage{bbold}
 \usepackage{tocbibind}
 \usepackage{newlfont}
 \usepackage{enumerate}
 \usepackage{listings}
 \usepackage{tikz}
 \usepackage{bbm}
\usepackage{color}
\usepackage[toc,page]{appendix}
\usepackage{hyperref}
\usepackage{leftidx}
\allowdisplaybreaks

\numberwithin{equation}{section}

\usepackage{hyperref}
\usepackage{leftidx}

\usepackage{color}

\newcommand{\black}[1]{\textcolor{black}{#1}}
\newcommand{\blue}[1]{\textcolor{black}{#1}}

\newcommand{\green}[1]{\textcolor{black}{#1}}

\usepackage[square,sort&compress,numbers]{natbib}
\usepackage{color} 

\numberwithin{equation}{section}

\definecolor{mygreen}{RGB}{28,172,0} 
\definecolor{mylilas}{RGB}{170,55,241}

\newcommand*\samethanks[1][\value{footnote}]{\footnotemark[#1]}

\newcommand{\RN}[1]{%
  \textup{\uppercase\expandafter{\romannumeral#1}}%
} 

\author{Francesca Biagini \thanks{Workgroup Financial Mathematics, Department of Mathematics, Ludwig-Maximilians-Universit\"{a}t M\"{u}nchen, Theresienstr. 39, 80333 Munich, Germany.} \and  Andrea Mazzon\samethanks[1] \and Thilo Meyer-Brandis\samethanks[1] \and  Katharina Oberpriller\thanks{Department of Mathematics, University of Freiburg, Ernst-Zermelo-Str. 1, 79104 Freiburg im Breisgau, Germany.}  }
 
\title{Liquidity based modeling of asset price bubbles via random matching}

\begin{document}

\maketitle

\begin{abstract}
In this paper we study the evolution of asset price bubbles driven by contagion effects spreading among investors via a random matching mechanism in a discrete-time version of the liquidity based model of \cite{JarrowProtter2012}. To this scope, we extend the Markov conditionally independent dynamic directed random matching of \cite{RandomMatchingDiscrete} to a stochastic setting to include stochastic exogenous factors in the model. We derive conditions guaranteeing that the financial market model is arbitrage-free and present some numerical simulation illustrating our approach.
\end{abstract}

\section{Introduction}
The formation of asset price bubbles has been object of many investigations in the economic and mathematical literature. Different causes have been indicated as triggering factors for bubble birth and evolution, such as a risk shifting problem in \cite{Allen}, the joint effect of the individual incentive to \emph{time the market} and the inability of arbitrageurs to coordinate their selling strategies in \cite{Abreu}, heterogenous beliefs between interacting traders as in \cite{Follmer}, \cite{HarrisonKreps}, \cite{Scheinkman}, \cite{Scheinkman2013}, \cite{Xiong}, \cite{Zhuk}, a disruption of the dynamic stability of the financial system in \cite{Douady2009}, \cite{Douady2011}, the diffusion of new investment decision rules from a few expert traders to a larger population of amateurs in \cite{Earl2007}, the tendency of investors to adopt the behavior of other agents in \cite{Kaizoji2000}, the presence of short-selling constraints in \cite{Miller} and of noise traders with erroneous stochastic beliefs in \cite{DeLong}. 

However, mathematical models for microfinancial interactions leading to the formation of asset price bubbles are still missing. This paper aims at filling this gap by studying a random matching mechanism  among investors which impacts the trading volume {of an asset and then its price via illiquidity effects}. To this purpose, we first introduce a discrete time version of the liquidity based model of \cite{JarrowProtter2012}, where the fundamental price of the asset is exogenously given, while the market price is influenced by the trading activities of investors {via an erosion of the limit order book}. {The birth of a bubble is then caused by a deviation of the market value from the fundamental one.} \\
Here, we model the signed volume of market orders by assuming that the investment attitudes of the traders on the market are influenced via a random matching mechanism. To this scope, we suppose that agents on the market can be of three types, i.e. optimistic, neutral and pessimistic regarding the future returns of the asset, and that they trade according to their type. This means that an {optimistic} agent places a buy market order {while a pessimistic one places a selling market order}. Neutral agents neither buy nor sell the asset. The evolution of the signed volume of market orders is {thus} determined by the fraction of optimistic, neutral and pessimistic agents, respectively. 

In order to model the evolution of these {quantities}, we extend the Markov conditionally independent dynamic directed random matching of \cite{RandomMatchingDiscrete} to a stochastic setting. More precisely, the model in \cite{RandomMatchingDiscrete} describes a mechanism how a continuum of agents search in a directed way for a suitable counterparty. The word ``directed'' refers to the fact that the search is not purely random, but the agents are motivated to meet another agent that provides them with some benefit. In particular, every agent is described by its type which may change at any time step, and can randomly mutate to another type {and} randomly match with another agent. {This meeting may induce a further type change}. Furthermore, agents can also enter some potentially enduring partnerships with random break-up times. 

These models have a broad application for example in the field of financial markets, monetary theory and labor economics. The first mathematical basis for this approach in a discrete time setting is provided in \cite{RandomMatchingDiscrete} and strongly relies on techniques of non-standard analysis, as a continuum of agents is considered. Given some deterministic functions describing the {probabilities associated to the random matching and random changes introduced above}, they prove existence of a dynamical system with independent agents types' and deterministic cross-sectional distribution of types. We now extend this model by allowing {the probabilities driving the system to also depend on an additional state of the world} to allow the random matching mechanism to be driven by some {stochastic} exogenous factors. Hereby, the technical difficulty is to find a suitable setting to extend the results in \cite{RandomMatchingDiscrete} in a consistent way. For this purpose we construct a probability space $\Omega$ as the product of the space $\hat \Omega$ of the random matching and the space $\tilde \Omega$ of the factors which may influence the transition probabilities, {and introduce a Markov kernel on $\Omega$}. After proving the existence of such a dynamical system with input processes, we study conditional type distributions. \\
We then apply these results to model investment attitudes {leading to bubble formation} in the discretized version of \cite{JarrowProtter2012}. More precisely, we assume that the signed volume of market orders is described by a random matching mechanism, where the agents can be of positive, negative or neutral type, as explained above. {The stochasticity of the transition probabilities is crucial here as it reflects the impact of heterogenous factors such as socio-economic indicators, external events, public news.} We are able to show that the market model is arbitrage-free by proving the existence of an equivalent martingale measure under suitable assumptions. Furthermore, we provide some examples for the input processes of the random matching mechanism in an arbitrage-free market model. We illustrate these results with numerical simulations. \\ 
The paper is organized as follows. In Section \ref{sec:FormationOfABubble} we introduce a discrete time version of the liquidity based model of asset prices in \cite{JarrowProtter2012}. In Section \ref{sec:main} we extend the directed random matching mechanism in \cite{RandomMatchingDiscrete} to a stochastic setting. We combine these two {constructions} in Section \ref{sec:tradingvolume}, {where we propose a model of} the signed volume of market orders influence by a random matching mechanism. In this setting we derive some conditions guaranteeing that the financial market model is arbitrage-free and we conclude with some numerical simulations.

\section{The formation of asset price bubbles} \label{sec:FormationOfABubble}
 
We consider a word-of-mouth mechanism spreading among investors who meet by random matching, giving rise to the formation of asset price bubbles. To this scope, we introduce a discretized version of the liquidity-based model in \cite{JarrowProtter2012}. 
\subsection{A liquidity-based model for asset price bubbles in discrete time}\label{sec:tradingimpact}
We here present a discrete time version of the continuous time model of \cite{JarrowProtter2012}, which explains the birth of bubbles as the deviation of the market price $S$ from the fundamental price $F$ caused by the impact of trading volume and illiquidity. \\
Let $T>0$ be a given trading horizon and consider a time discretization $0=:t_0<t_1 <...<t_{N-1}<t_N=T$ {of the interval $[0,T]$}. Also introduce a filtered probability space $(\Omega, \mathcal{F},(\mathcal{F}^{i})_{i=0,...,N},P)$, where we set $\mathcal{F}^{i}:=\mathcal{F}_{t_i}$ for $i=1,...,N$. In Section \ref{sec:main} we further specify a possible construction of this space in the context of random matching. The market model consists of the money market account $B \equiv 1$ and one liquid financial asset (stock), which is traded through limit and market orders.
\begin{remark}
In order to be consistent with the notation of the random matching mechanism, see Section \ref{sec:main}, we indicate the time $t_i$ with a superscript $i$ for filtrations or processes.
\end{remark}
The fundamental price of the asset is given by the stochastic process $\green{F=(F^{i})_{i=0,...,N}}$, where $F^i$ represents the value at time $t_i$ for $i=0,...,N$. Such a process is exogenously given. On the other hand, the market price of the asset is generated by the trading activity of the investors as we describe in the following.
 
 Coherently with the construction of \cite{JarrowProtter2012}, we assume that the average price to pay per share for a transaction of size $x$ via a market order \green{at time $t_i$} is given by
\begin{equation} \label{eq:DiscretizedSupplyCurve}
	S^{i}(x)=S^{i} +M^{i}x, \quad x \in \mathbb{R}^+, \ i=0,...,N,
\end{equation}
where $S=(S^{i})_{i=0,...,N}$ and $M=(M^{i})_{i=0,...,N}$ are non-negative, adapted processes on the space \linebreak
$(\Omega, \mathcal{F},\green{(\mathcal{F}^{i}})_{i=0,...,N},P)$, representing the quoted price and a measure of illiquidity, respectively. Fix a time $t_i$ for $i=1, \dots, N$. The limit order book at $t_i$ is described by the density function $\green{\rho^{i}(\cdot)}$, where $\green{\rho^{i}(z)}$ is the number of shares offered at price $z$ at time $t_i$. As in \cite{JarrowProtter2012}, the total amount paid by a trader who wants to buy $x$ shares at time $t_i$ is given by
\begin{equation}
\int_{S^{i}}^{z_x} z \rho^{i}(z) dz,
\end{equation}
where $z_x$ is the solution of
\begin{equation*}
	\int_{S^{i}}^{z_x} \rho^{i}(z) dz=x. 
\end{equation*}
Due to the linear structure in \eqref{eq:DiscretizedSupplyCurve} it follows that $\rho^{\green{i}}(z)=1/2 M^{\green{i}}$ and $z_x=S^{\green{i}}+2 M^{\green{i}}x$, {see \cite{JarrowProtter2012} for further details.} \\
Let $X=(\green{X^{i}})_{i=0,...,N}$ be an adapted stochastic process representing the signed volume of aggregate market orders (buy minus sell orders).
Next, we introduce a process $R=(\green{R^{i}})_{i=0,...,N}$ with values in $[0,1]$ to describe the short-term resiliency {of the limit order book}. In particular, {if $\Delta X$ buy market orders are executed at time $t_i$,}  $R^{\green{i}}$ represents the proportion of new {sell limit} orders {placed from $t_i$ to $t_{i+1}$, having therefore the effect to partly fill the temporary gap $[S^{\green{i}}, S^{\green{i}}+\Delta X]$} in the limit order book.
{If the gap caused by the new buy market orders is not fully filled before other market orders are executed, the market price of the asset deviates from the fundamental value, thus creating a bubble. However, it is observed that such a deviation decays in the long run, see \cite{JarrowProtter2012} for details. Such an effect is quantified by the speed of decay process $\kappa=(\green{\kappa^{i}})_{i=0,...N}$.}

The evolution of the market price process $S=(\green{S^{i}})_{i=1,...,N}$ is then given by 
\begin{align}
	S^{\green{i}}=S^{\green{i-1}} + F^{\green{i}}-F^{\green{i-1}}  - \kappa^{\green{i}}(S^{\green{i-1}}-F^{\green{i-1}})\Delta \green{t_i} + 2\Lambda^{\green{i}}M^{\green{i}}\Delta X^{\green{i}},\quad i = 1, \dots,N, \label{eq:DefMarketPriceProcess}
\end{align}
where  $\Lambda^{\green{i}}:=1-R^{\green{i}}$, $i=\black{0},...,N$. Moreover, $\Delta \green{t_i}:=t_{i}-t_{i-1}$, $\Delta X^{\green{i}}:= X^{\green{i}}-X^{\green{i-1}}$ for $i=1,...,N$. At initial time we have $X^{\black{0}}=0$ and $S^{\black{0}}=F^{\black{0}}$. In particular, \eqref{eq:DefMarketPriceProcess} is a discretized version of the SDE considered in \cite{JarrowProtter2012}.

{Following \cite{JarrowProtter2012}, we now provide the definition of an asset price bubble in this setting.}

\begin{definition}
	An \emph{asset price bubble} $\beta=(\beta^{\green{i}})_{i=0,...,N}$ is defined as  
	\begin{equation*}
		\beta^{\green{i}}:=S^{\green{i}}-F^{\green{i}}, \quad i = \black{0}, \dots,N.
	\end{equation*}
By \eqref{eq:DefMarketPriceProcess} we obtain that
\begin{align}
	\beta^{\green{i}}&=\beta^{\green{i-1}}-\kappa^{\green{i}}\beta^{\green{i-1}}\Delta \green{t_i} + 2\Lambda^{\green{i}}M^{\green{i}}\Delta X^{\green{i}},  \quad i = 1, \dots,N, \label{eq:BubbleEvolution}\\
	\beta_0 &= 0. \notag
\end{align}
\end{definition}
The \emph{birth} and the \emph{burst times of the bubble} are identified by the stopping times
\begin{equation} \label{eq:DefinitionBirth}
\tau_{{+}}:=\green{t_{\bar{l}}} \quad \text{ with } \quad \green{\bar{l}:=\inf \lbrace j=0,...,N: \beta^{\green{j}} > 0\rbrace {\wedge T}}
\end{equation}  
and
$$
\tau_0:=\green{t_{\bar{k}}} \quad \text{ with } \quad \green{\bar{k}:= \inf \lbrace j=0,...,N: j\geq \bar{l} \text{ in }  \eqref{eq:DefinitionBirth} \text{ such that } \beta^{j}=0 \rbrace {\wedge T}},
$$
respectively. We use here the convention $\inf \emptyset =+\infty$.
Note that $\tau_{{+}}$ is the first time when the three process $\Lambda$, $M$ and $X$ are different from zero, see \eqref{eq:BubbleEvolution}.

\begin{definition}
	The \emph{market wealth process} $W=(W^{i})_{i=0,...,N}$ is defined by
	\begin{equation*}
		W^{\green{i}}=D^{\green{i}}+S^{\green{i}} \textbf{1}_{\lbrace t_i < \tau \rbrace} + F^{\green{j}}  \textbf{1}_{\lbrace t_{\green{j}} = \tau\rbrace}, \quad i = 0, \dots,N,
	\end{equation*}
	and the \emph{fundamental wealth process} $W^F=(\green{W^{F,i}})_{i=0,...,N}$ by
	\begin{equation*}
		W^{F,\green{i}}=D^{\green{i}}+F^{\green{i}}, \quad i = 0, \dots,N.
	\end{equation*}
\end{definition}
Note that 
\begin{align*}
	W^{\green{i}}-W^{\green{F,i}}=S^{\green{i}}-F^{\green{i}}=\beta^{\green{i}}, \quad i = 0, \dots,N.
\end{align*}

Equation \eqref{eq:BubbleEvolution} shows that the main force driving the bubble evolution is the signed volume of market orders $X$. We now focus on modeling $X$ by assuming that the investment attitudes of the traders on the market are influenced via a random matching mechanism.
To this scope, we suppose that agents on the market can be of three types, i.e. optimistic, neutral and pessimistic regarding the future returns of the asset, and that they trade according to their type. This means that an {optimistic} agent places a buy market order {while a pessimistic one places a selling market order}. Neutral agents neither buy nor sell the asset. {Based on this characterization, from now on we refer to optimistic and pessimistic agents also as \emph{buyers} and \emph{sellers}, respectively. We admit that agents may influence each other if they meet, and that they may change their type at each time $t_i$, $i=1,...,N,$ via a random matching mechanism as we explain in Section \ref{sec:main}. The evolution of $X$ is determined by the processes \green{$p_i=(p^{j}_i)_{j=0,\dots,N}$}, $i=1,2,3$, standing for the fraction of optimistic, neutral and pessimistic agents, respectively. In particular, the value of $X$ at time $t_i$ is given by
 \begin{equation}\label{eq:Xp}
\green{ X^i=\Theta^i (p_1^{i}-p_3^{i}),} \quad i = 0, \dots, N,
\end{equation}
 where $\Theta=\green{(\Theta^{i})}_{i = 0, \dots, N}$ is an adapted stochastic process modelling the average size of buy market orders as in \cite{firstpaper}. We assume that at time $t_0=0$ it holds $\green{p_0^1=p_3^i}$, i.e. that the fraction of optimistic agents is equal to that of pessimistic ones, so that $\green{X^0=0}$. We now model the evolution of the fractions \green{$p_1$, $p_2$} and \green{$p_3$} by using a special case of the \emph{Markov {c}onditionally {i}ndependent dynamic directed random matching} which we introduce in the next section.

\section{Markov Conditionally Independent dynamic directed random matching}\label{sec:main}

Consider a probability space $({\Omega}, \cal{F}, P)$ representing all possible states of the world, on which we consider a large economy. The space of agents is given by an atomless probability space $(I, \cal{I}, \lambda)$. Furthermore, it is a common assumption that the agents also face some individual risks. The natural approach to take this into account is to consider a random variable $f$ on the product space $(I \times \Omega, \cal{I} \otimes \cal{F})$ to a Polish space $(Y, \cal{G})$ which is essentially pairwise independent, see Definition 2 in \cite{RandomMatchingDiscrete}.
\begin{definition}
Consider the random variable $f: (I \times \Omega, \cal{I}\otimes \cal{F}, \lambda \otimes P) \to (Y,\cal{G}),$ where $Y$ is a Polish space endowed with the Borelian $\sigma$-algebra $\cal{G}$. We set $f_i:=f(i, \cdot)$ for all $i \in I.$ We say that $f$ is \emph{essentially pairwise independent} if for $\lambda$-almost all $j \in I$, $f_j$ is independent of $f_i$ for $\lambda$-almost all $i \in I$.
\end{definition}
In Proposition 2.1 of \cite{sun2006exact} and Proposition 1.1 of \cite{sun1998almost} it is shown that an essentially pairwise independent random variable, which is also jointly measurable, is constant for $\lambda$-almost all $i \in I$. This is the so called ``sample measurability problem'', which has been studied in \cite{doob_1937}, \cite{Keisler_Sun}. To overcome this issue, the $\sigma$-algebra $\cal{I} \otimes \cal{F}$ needs to be enlarged to allow jointly measurable random variables to be essentially pairwise independent but not constant. This measurability problem can be solved by working with an extension of the product space $(I \times \Omega, \cal{I} \otimes \cal{F}, \lambda \otimes P)$ which still satisfies the Fubini property, see \cite{sun2006exact}.
We here recall the definition of a Fubini extension, see e.g. Definition 1 of \cite{RandomMatchingDiscrete}. 
\begin{definition}
A probability space $(I \times \Omega, \cal{W}, Q)$ is said to be a \emph{Fubini extension} of the product probability space $(I \times \Omega, \cal{I} \otimes \cal{F}, \lambda \otimes P)$ if for any real-valued $Q$-integrable random variable $f$ on $(I \times \Omega, \cal{W}, Q)$ we have that
\begin{enumerate}
\item  the functions $f_i(\cdot):=f(i,\cdot)$ and $f_{\omega}(\cdot):=f(\cdot,\omega)$ are integrable on $(\Omega, \cal{F}, P)$ for $\lambda$-almost all $i \in I$, and on 
$(I, \cal{I}, \lambda)$ for $P$-almost all $\omega \in \Omega$, respectively;
\item $\int_{\Omega} f_i dP$ and $\int_{I} f_{\omega} d\lambda$ are integrable on $(I, \cal{I}, \lambda)$ and on $(\Omega, \cal{F}, P)$, respectively, with 
$$
\int_{I \times \Omega} f dQ=\int_I\left(\int_{\Omega}f_idP\right)d\lambda=\int_{\Omega}\left(\int_{I}f_id\lambda\right)dP.
$$
\end{enumerate}
The Fubini extension is denoted by $(I \times \Omega, \cal{I} \boxtimes \cal{F}, \lambda \boxtimes P)$.
\end{definition}
Moreover, note that by definition it holds $\lambda \boxtimes P |_{\cal{I} \otimes \cal{F}}= \lambda \otimes P$. In Theorem 6.2 in \cite{sun1998ATheoryOfHyperfinite} and Proposition 5.6 in \cite{sun2006exact} it is shown that there exists a \emph{rich} Fubini extension, which allows the construction of processes with essentially pairwise independent and  jointly measurable random variables, which are not $\lambda$-almost surely constant. From now on, we always work with such a rich Fubini extension of the original product space. We now describe a matching mechanism among the agents, by following Definition 2 in \cite{duffie_qiao_sun_2017}. 
\begin{definition} \label{defi:Matching}
	\begin{enumerate}
		\item A \emph{full matching} $\phi:I \to I$ is a one-to-one mapping, such that for each $i \in I$, $\phi(i) \neq i$ and $\phi(\phi(i))=i$.
		\item A \emph{(partial) matching} $\psi$ is a matching from $I$ to $I$ such that for some subset $B$ of $I$, the restriction of $\psi$ to $B$ is a full matching on $B$, and $\psi(i)=i$ on $I \setminus B$. This means that agent $i$ is matched with agent $\psi(i)$ for $i \in B$, whereas any agent $i$ not in $B$ is unmatched, represented by setting $\psi(i)=i$.
		\item A \emph{random matching} on $\Omega$ is a mapping $\pi: I \times \Omega \to I$ such that $\pi_{\omega}:=\pi(\cdot, \omega)$ is a matching for each $\omega \in \Omega$.
	\end{enumerate}
\end{definition} 
In the sequel we use ``matching'' to denote a partial matching for the sake of simplicity. \\

\black{Next, we introduce the definition of a dynamical system with input processes.} A dynamical, directed random matching mechanism has been studied for the first time in \cite{RandomMatchingDiscrete}, where the  probabilities describing the random matching are deterministic functions which only depend on the current probability distribution on the space of extended type distributions. We generalize this approach by allowing the probabilities also to depend on the state of the world, which influences the random matching. Hereby, the technical difficulty is to find a suitable setting to extend the results in \cite{RandomMatchingDiscrete} in a consistent way, as the existence of the random matching system in \cite{RandomMatchingDiscrete} relies extensively on techniques of nonstandard analysis. This is necessary in order to construct a Fubini extension by working with Loeb spaces which satisfy a Fubini property. We provide a generalization of the setting in  \cite{RandomMatchingDiscrete} by assuming that $\Omega$ is a product space and by using stochastic kernels.\\
Let $(\tilde \Omega, \tilde{\mathcal{F}}, \tilde P)$ be a probability space and $(\hat \Omega, \hat{\mathcal{F}})$ another measurable space. We consider the product space  
\begin{equation} \label{eq:ProductSpace}
	(\Omega,  \mathcal{F}):={(\tilde \Omega \times \hat \Omega , \tilde{\mathcal{F}} \otimes \hat{\mathcal{F}})},
\end{equation} 
and recall the definition of a Markov kernel from Definition 8.25 in \cite{Klenke_2020} for the reader's convenience.
\begin{definition} \label{def:MarkovKernel}
	Let $(\tilde{\Omega}, \tilde{\mathcal{F}}), (\hat{\Omega}, \hat{\mathcal{F}})$ be measurable spaces. A map $\kappa:\tilde{\Omega} \times \hat{\mathcal{F}} \to [0,\infty]$ is called a \emph{Markov kernel} or \emph{stochastic kernel} from $\tilde{\Omega}$ to $\hat{\Omega}$ if:
	\begin{enumerate}
		\item $\kappa(\cdot, \hat{A})$ is $\tilde{\mathcal{F}}$-measurable for any $\hat{A} \in \hat{\mathcal{F}}$;  
		\item $\kappa(\tilde{\omega}, \cdot)$ is a probability measure on $(\hat{\Omega}, \hat{\mathcal{F}})$ for any $\tilde{\omega} \in \tilde{\Omega}$. 
	\end{enumerate}
\end{definition} 
Let $\hat{P}$ be a Markov kernel (or stochastic kernel) from $\tilde \Omega$ to $\hat \Omega$. Given $\tilde \omega \in \tilde \Omega,$ we set $\hat P^{\tilde \omega}:=\hat P(\tilde \omega)$ with a slight notational abuse. We then introduce a probability measure $P$ on $(\Omega,  \mathcal{F})$ as the semidirect product of $\tilde P$ and $\hat P$, that is,
\begin{equation}\label{eq:ProbhatPSetting}
P({\tilde A \times \hat{A}}) := (\tilde P \ltimes \hat P)({\tilde A \times \hat{A}}) =  \int_{\tilde A}\hat P^{\tilde \omega}(\hat A)d\tilde P(\tilde \omega)
\end{equation} 
\black{for $\tilde{A} \in \tilde{\mathcal{F}}$, $\hat{A} \in \hat{\mathcal{F}}$.}\\
Let $(I \times \Omega, \cal{I} \boxtimes \cal{F},\lambda \boxtimes P)$ be {a rich} Fubini extension of $(I \times \Omega, \cal{I} \otimes \cal{F}, \lambda \otimes P)$. We classify all agents in ${I}$ according to their type belonging to the finite space $S=\lbrace 1,2,...,K \rbrace$. We say that an agent has type $J$ if he is not matched. We denote by $\hat{S}:=S \times (S \cup \lbrace J \rbrace)$ the extended type space. If an agent has the extended type $(k,l),$ this means that he is of type $k \in S$ and is currently matched to another agent of type $l \in S$. If an agent of type $k$ is not matched at the moment, the agent's extended type is $(k,J)$. We consider probability distributions on $\hat{S}.$ In particular, we introduce the space $\hat{\Delta}$ of extended type distributions, which is the set of probability distributions ${p}$ on $\hat{S}$ satisfying ${p}(k,l)={p}(l,k)$ for any $k$ and $l$ in $S$. We endow $\hat{\Delta}$ with the topology $\cal{T}^{\Delta}$ induced by the topology on the space of matrices with $|S|$ rows and $|S| + 1$ columns. Moreover, let $\hat{p}=(\hat{p}^n)_{n \geq 1}$ be a stochastic process on $(\Omega, \mathcal{F},P)$ with values in $\hat{\Delta}$, {representing the  evolution of the underlying extended type distribution}. We assume that $\hat{p}^0$ is deterministic.\\
In this setting we describe how agents may change their type by random matching with other agents. Consider time periods $(n)_{n \geq 1}$.
Each time period $n$ can be divided into three steps: \emph{mutation}, \emph{random matching} and \emph{match-induced type changing with break-up}. We assume that the probabilities of these three mechanisms depend on the state of the world $\omega \in \Omega$, i.e. these steps are determined by input processes $(\eta^n, \theta^n, \xi^n, \sigma^n, \varsigma^n)_{n \geq 1}$ on $(\Omega, \mathcal{F},P)$ as we describe next. Here, $(\eta^n, \theta^n, \xi^n, \sigma^n, \varsigma^n)$ are matrix valued processes, with $(\eta^n, \theta^n, \xi^n, \sigma^n, \varsigma^n)=(\eta^n_{kl},\theta^n_{kl}, \xi^n_{kl}, \sigma^n_{kl}[r,s], \varsigma^n_{kl}[r])_{k,l,r,s \in S \times S \times S \times S }$ for $n \geq 1$.  \\
The first step of each time period $n$ is the random mutation step, i.e. an agent of type $k \in S$ becomes an agent of type $l$ with a given mutation probability $\eta^{n}_{kl}$, where
\begin{align}
	\eta^n_{kl}: ({\Omega},\mathcal{F},P) \to ([0,1],\mathcal{B}([0,1])) \notag \\
	\eta_{kl}^n(\omega):=\eta_{kl}(\tilde{\omega},n, \hat{p}^n(\tilde{\omega}, \hat{\omega})) \label{eq:MutationProcessEta1}
\end{align}
with $\eta_{kl}:\tilde{\Omega} \times \mathbb{N} \times \hat{\Delta} \to [0,1]$. Precisely, if $\hat{p}^n(\tilde{\omega},\hat{\omega})$ is the underlying extended type distribution at time $n \in \mathbb{N}$ under the scenario $(\tilde{\omega}, \hat \omega) \in \Omega$, then $\eta_{kl}(\tilde{\omega},n,\hat{p}^n(\tilde{\omega}, \hat{\omega}))$ represents the probability that an agent of type $k$ becomes an agent of type $l$ at time $n$ given $\tilde{\omega}$. Here we assume that for every $n \geq 1$, $k,l \in S,$ the function $\eta^n_{kl}$ is $(\mathcal{F},\mathcal{B}([0,1]))$-measurable, and that for each $k \in S$ and $\omega \in \Omega$ it holds
\begin{equation} \label{eq:EtaSumUpTo1}
	\sum_{l \in S}\eta_{kl}^n(\omega)=1. 
\end{equation}
In the second step any currently unmatched agent can be matched. For each $(k,l) \in S \times S$ we define
\begin{align}
	\theta^{n}_{kl}: ({\Omega},\mathcal{F},P) \to ([0,1],\mathcal{B}([0,1])) \label{eq:MatchingProcessTheta1}\\
	\theta_{kl}^n(\omega):=\theta_{kl}(\tilde{\omega},n, \hat{p}^n(\tilde{\omega}, \hat{\omega})) \label{eq:MatchingProcessTheta2}
\end{align}
with $\theta_{kl}:\tilde{\Omega} \times \mathbb{N} \times \hat{\Delta} \to [0,1]$. If $\hat{p}^n(\tilde{\omega}, \hat{\omega})$ is the underlying extended type distribution at time $n \in \mathbb{N}$ under the scenario $(\tilde{\omega}, \hat \omega) \in \Omega$, then $\theta_{kl}(\tilde{\omega},n,\hat{p}^n(\tilde{\omega}, \hat{\omega}))$ is the probability that an unmatched agent of type $k$ is matched to an agent of type $l$ given the scenario $\tilde{\omega}$ at time $n$. 
Here, we assume that for every $n \geq 1$, $k,l \in S$ the function $\theta^n_{kl}$ is $(\mathcal{F},\mathcal{B}([0,1]))$-measurable, and that for all $k,l \in S, \tilde \omega \in \tilde{\Omega},$ and $n \in \mathbb{N}$ the function $\hat{p}_{kJ}\theta_{kl}(\tilde{\omega}, n,\hat{p})$ is continuous in $\hat{p} \in \hat{\Delta}$ {with respect to the topology $\cal{T}^{\Delta}$}. Moreover, for any $k,l \in S, \tilde\omega \in \tilde{\Omega}, n \in \mathbb{N}$ and $\hat{p} \in \hat{\Delta}$ we suppose the following to hold 
\begin{equation} \label{eq:ThetaCondition}
	\hat{p}_{kJ}\theta_{kl}(\tilde{\omega},n,\hat{p})=\hat{p}_{lJ}\theta_{lk}(\tilde{\omega},n,\hat{p}) \quad \text{ and } \quad \sum_{r \in S} \theta_{kr}(\tilde{\omega},n,\hat{p}) \leq 1.
\end{equation}
Moreover, we define $b^n_k({\omega}):=1-\sum_{l \in S} \theta_{kl}^n(\omega)$.\\ In the third step, a currently matched pair of agents of respective types $k$ and $l$, including those who have been matched at the second step, can break up. To describe this behavior we consider the process
\begin{align}
	\xi^{n}_{kl}: ({\Omega},\mathcal{F},P) \to ([0,1],\mathcal{B}([0,1])) \notag \\
	\xi_{kl}^n(\omega):=\xi_{kl}(\tilde{\omega},n, \hat{p}^n(\tilde{\omega}, \hat{\omega})) \label{eq:BreakingUpProcessXi}
\end{align}
with $\xi_{kl}:\tilde{\Omega} \times \mathbb{N} \times \hat{\Delta} \to [0,1]$.  If $\hat{p}^n(\tilde{\omega}, \hat{\omega})$ is the underlying extended type distribution at time $n \in \mathbb{N}$ under the scenario $(\tilde{\omega}, \hat \omega) \in \Omega$, then $\xi_{kl}(\tilde{\omega},n,\hat{p}^n(\tilde{\omega}, \hat{\omega}))$ represents the probability that a matched pair of types $k$ and $l$ breaks up under $\tilde{\omega}$. We assume that for every $n \geq 1$, $k,l \in S,$ it holds  $\xi_{kl}^n=\xi_{lk}^n$ and that the function $\xi^n_{kl}$ is $(\mathcal{F},\mathcal{B}([0,1]))$-measurable. \\
For each $(k,l, r, s) \in S \times S \times S \times S$ and $n \geq 1$, we now introduce a process $\sigma_{kl}^n[r,s]$ given by
\begin{align}
	\sigma^{n}_{kl}[r,s]: ({\Omega},\mathcal{F},P) \to ([0,1],\mathcal{B}([0,1])) \notag \\
	\sigma_{kl}^n[r,s](\omega):=\sigma_{kl}[r,s](\tilde{\omega},n, \hat{p}^n(\tilde{\omega}, \hat{\omega})).\label{eq:sigma}
\end{align}
In particular, if $\hat{p}^n(\tilde{\omega}, \hat{\omega})$ is the underlying extended type distribution at time $n \in \mathbb{N}$ under the scenario $(\tilde{\omega}, \hat \omega) \in \Omega$, then $\sigma_{kl}[r,s](\tilde{\omega},n,\hat{p}^n(\tilde{\omega}, \hat{\omega}))$ is the probability that a matched pair of agents of respective types $k$ and $l$, which stays in their relationship, becomes a pair of agents of type $r$ and $s$, at time $n$ given $\tilde \omega$. Here, we assume that for every $n \geq 1$, $k,l,r,s \in S$ the function $\sigma^n_{kl}[r,s]$ is $(\mathcal{F},\mathcal{B}([0,1]))$-measurable. We also assume that for each $\omega \in {\Omega}$ it holds
\begin{equation}\label{eq:conditionsigma}
	\sum_{r,s \in S} \sigma_{kl}^n[r,s](\omega)=1 \quad \text{ and } \quad \sigma_{kl}^n[r,s](\omega)=\sigma_{lk}^n[s,r](\omega)
\end{equation}
for any $k,l,r,s \in S$. \\
For each $(k,l, r) \in S \times S \times S$, if a matched pair of agents of respective types $k$ and $l$ breaks up, the agent of type $k$ can become an agent of type $r$ with probability $\varsigma_{kl}^n[r]$ given by the process
\begin{align}
	\varsigma^{n}_{kl}[r]: ({\Omega},\mathcal{F},P) \to ([0,1],\mathcal{B}([0,1])) \nonumber \\
	\varsigma_{kl}^n[r](\omega):=\varsigma_{kl}[r](\tilde{\omega},n, \hat{p}^n(\tilde{\omega}, \hat{\omega})) \label{eq:MatchingAfterBreakUp1}
\end{align}
with $\varsigma_{kl}[r]:\tilde{\Omega} \times \mathbb{N} \times \hat{\Delta} \to [0,1]$.  Precisely, if $\hat{p}^n(\tilde{\omega}, \hat{\omega})$ is the underlying extended type distribution at time $n \in \mathbb{N}$ under the scenario $(\tilde{\omega}, \hat \omega) \in \Omega$, then $\varsigma_{lk}^n[r](\tilde{\omega},n,\hat{p}^n(\tilde{\omega}, \hat{\omega}))$ represents the probability that given a pair of agents of type $k$ and $l$ that break up at time $n$, the agent of type $k$ becomes an agent of type $r$ given $\tilde{\omega}$.  Here, we assume that for every $n \geq 1$, $k,l,r \in S$ the function $\varsigma^n_{kl}[r]$ is $(\mathcal{F},\mathcal{B}([0,1]))$-measurable, and that for each $\omega \in \Omega$
\begin{equation} \label{eq:VarSigmaSumUpTto1}
	\sum_{r \in S } \varsigma_{kl}^n[r](\omega)=1.
\end{equation} 
We assume that $(\eta^0, \theta^0, \xi^0, \sigma^0, \varsigma^0)$ are deterministic, i.e. $\eta(0,\cdot), \theta(0,\cdot), \xi(0,\cdot), \sigma(0,\cdot), \varsigma(0,\cdot): \hat{\Delta} \to [0,1]$.
\begin{remark}
	Note that for fixed $\tilde{\omega} \in \tilde{\Omega}$ the processes $(\eta, \theta, \xi,\sigma, \varsigma)$ are deterministic and depend only on the current time $n$ and the extended type distribution $\hat{p}^n$ at time $n$. This means that for fixed $\tilde{\omega} \in \tilde{\Omega}$ the setting boils down to the framework of \cite{RandomMatchingDiscrete}.
\end{remark}
Next, we use the processes $(\eta, \theta, \xi,\sigma, \varsigma)$ to give a definition of a dynamical system $\mathbb{D}$ with random probabilities, as an extension of a dynamical system in \cite{RandomMatchingDiscrete}.
\begin{definition} \label{defi:DynamicalSystemDiscrete}
A \emph{dynamical system} $\mathbb{D}$ defined on $(I \times \Omega, \cal{I} \boxtimes \mathcal{F},\lambda \boxtimes P )$ is a triple $\Pi=(\alpha, \pi,g)=(\alpha^n, \pi^n, g^n)_{n \in \mathbb{N} \backslash \lbrace 0 \rbrace}$ such that for each integer period $n\geq 1$ we have
\begin{enumerate}
	\item $\alpha^n: I \times {\Omega} \to S$ is the $\cal{I} \boxtimes {\mathcal{F}}$-measurable agent-type function. The corresponding end-of-period type of agent $i$ under the realization ${ \omega} \in {\Omega}$ is given by $\alpha^n(i,{\omega}) \in S$.
	\item A random matching $\pi^n: I \times {\Omega} \to I$, describing the end-of-period agent $\pi^n(i)$ to whom agent $i$ is currently matched, if agent $i$ is currently matched. If agent $i$ is not matched, then $\pi^{n}(i)=i$. The associated $\cal{I} \boxtimes {\mathcal{F}}$-measurable partner-type function $g^{n}:I \times {\Omega} \to S \cup \lbrace J \rbrace$ is given by
		\[g^n(i,{\omega})=\begin{cases}
	\alpha^n(\pi^n(i,{\omega}),{\omega}) & \text{ if } \pi^n(i,{\omega}) \neq i\\
		J& \text{ if } \pi^n(i,{\omega}) = i,
	\end{cases}\]
	providing the type of the agent to whom agent $i$ is matched, if agent $i$ is matched, or $J$ if agent $i$ is not matched. 
\end{enumerate}
Let the initial condition $\Pi^0=(\alpha^0,\pi^0,g^0)$ of $\mathbb{D}$ be given, \black{i.e. functions $\alpha^0: I \to S$, $\pi^0:I \to I$ and $g^0: I \to S \cup \lbrace J \rbrace$.} We now construct a dynamical system $\mathbb{D}$ defined on $(I \times \Omega, \cal{I} \boxtimes \mathcal{F},\lambda \boxtimes P )$ with {input processes} $(\eta^n,\theta^n,\xi^n, \sigma^n, \varsigma^n)_{n\geq 1}$. We assume that $\Pi^{n-1}=(\alpha^{n-1},\pi^{n-1},g^{n-1})$ is given for some $n \geq 1$, and define $\Pi^{n}=(\alpha^n,\pi^n, g^n)$ by characterizing the three sub-steps of random change of types of agents, random matchings, break-ups and possible type changes after matchings and break-ups as follows. 

\textbf{Mutation:} For $n \geq 1$ consider an $\cal{I} \boxtimes {\mathcal{F}}$-measurable post-mutation function $$\bar{\alpha}^n: I \times \Omega \to S.$$ In particular, $\bar{\alpha}_i^n({\omega}):=\bar{\alpha}^n(i,\omega)$ is the type of agent $i$ after the random mutation under the scenario ${\omega} \in {\Omega}$. The type of the agent to whom {an agent} is matched is identified by a $\cal{I} \boxtimes {\mathcal{F}}$-measurable function $$\bar{g}^n:I \times \Omega \to S \ {\cup \  \lbrace J \rbrace},$$ given by 
$$
\bar{g}^n(i,{\omega})=\bar{\alpha}^n(\pi^{n-1}(i,{\omega}),{\omega})
$$ 
for any ${\omega} \in {\Omega}$. {In particular, $\bar{g}_i^n({\omega}):=\bar{g}^n(i,\omega)$ is the type of the agent to whom {agent} \black{$i$} is matched  under the scenario ${\omega} \in {\Omega}$.} Given $\hat{p}^{n-1}$ and $\tilde \omega \in \tilde \Omega$, for any $k_1,k_2,l_1$ and $l_2$ in $S$, for any $r \in S \cup \lbrace J \rbrace$, for $\lambda$-almost every agent $i$, we set
\small{
\begin{align}
&\hat P^{\tilde \omega}\left(\bar{\alpha}_i^n{(\tilde{\omega},\cdot)}=k_2, \ \bar{g}_i^n{(\tilde{\omega},\cdot)}=l_2 \vert \alpha_{i}^{n-1}{(\tilde{\omega},\cdot)}=k_1, \ g_i^{n-1}{(\tilde{\omega},\cdot)}=l_1, \ \hat{p}^{n-1} {(\tilde{\omega},\cdot)}\right)\blue{(\hat{\omega})}\nonumber \\&\quad=\eta_{k_1,k_2}\left( \tilde \omega,n,\hat{p}^{n-1}(\tilde{\omega},{\hat{\omega}})\right)\eta_{l_1,l_2}\left( \tilde \omega,n,\hat{p}^{n-1}(\tilde{\omega},{\hat{\omega}})\right), \label{eq:IndiMutation1}
\end{align} 
\begin{align}
&\hat P^{\tilde \omega}\left(\bar{\alpha}_i^n{(\tilde{\omega},\cdot)}=k_2, \ \bar{g}_i^n{(\tilde{\omega},\cdot)}=r \vert \alpha_{i}^{n-1}{(\tilde{\omega},\cdot)}=k_1, \ g_i^{n-1}{(\tilde{\omega},\cdot)}=J, \ \hat{p}^{n-1}{(\tilde{\omega},\cdot)}\right)\blue{(\hat{\omega})} \nonumber \\
&\quad =\eta_{k_1,k_2}\left(\tilde{\omega},n,\hat{p}^{n-1}(\tilde{\omega},{\hat{\omega}})\right)\delta_J(r), \label{eq:IndiMutation2}
\end{align}}
We then \black{define} 
$$
\bar{\beta}^n{(\omega)}=(\bar{\alpha}^n{(\omega)}, \bar{g}^n{(\omega)}), \quad n \geq 1.
$$
The post-mutation extended type distribution realized in the state of the world ${\omega \in \Omega}$ is {denoted by $\check{p}(\omega)=(\check{p}^n(\omega)[k,l])_{k \in S,  l \in S \cup \black{ \lbrace J \rbrace}}$, where 
	\begin{equation}\label{eq:PostMutationExtendedType}
	\check{p}^n(\omega)[k,l]:=\lambda(\lbrace i \in I: \bar{\alpha}^n(i,\omega)=k, \  \bar{g}^n(i,\omega)=l\rbrace).
	\end{equation}
}
\textbf{Matching:} {We introduce a random matching} $\bar{\pi}^n: I \times { \Omega} \to I$ {and the associated post-matching partner-type function $\bar{\bar{g}}^n$ given by
\[\bar{\bar{g}}^n(i,{ \omega})=\begin{cases}
\bar{\alpha}^n(\bar{\pi}^n(i,{ \omega}),{ \omega}) & \text{ if } \bar{\pi}^n(i,{ \omega}) \neq i\\
J& \text{ if }  \bar{\pi}^n(i,{ \omega})=i,
\end{cases}\]
satisfying} the following properties:
\begin{enumerate}
\item $\bar{\bar{g}}^n$ is $\cal{I} \boxtimes {\mathcal{F}}$-measurable.
	\item {For any $\tilde \omega \in \tilde \Omega$, any $k,l \in S$ and any $r \in S \cup \lbrace J \rbrace$, it holds
\begin{equation} \notag
	\hat P^{\tilde \omega}(\bar{\bar{g}}^n{(\tilde{\omega},\cdot)}=r \vert \bar{\alpha}^n_i{(\tilde{\omega},\cdot)}=k, \ \bar{g}_i^n{(\tilde{\omega},\cdot)}=l)(\hat{\omega})=\delta_l(r).
\end{equation}
This means that
	\begin{equation} \notag
		\bar{\pi}^{n}_{{ \omega}}(i)=\pi_{{ \omega}}^{n-1}(i) \quad \text{ for any } i \in \lbrace i: \pi^{n-1}(i,{ \omega}) \neq i \rbrace.
	\end{equation}}
\item Given $\tilde{\omega}\in \tilde{\Omega}$ and the post-mutation extended type distribution $\check{p}^n$ in \eqref{eq:PostMutationExtendedType}, an unmatched agent of type $k$ is matched to a unmatched agent of type $l$ with conditional probability {$\theta_{kl}(\tilde{\omega},n,\check{p}^n)$}, that is for $\lambda$-almost every agent $i$ and $\hat P^{\tilde{\omega}}$-almost every $\hat{\omega}$, we define
\begin{equation} \label{eq:MatchingCondProb1}
	\hat P^{\tilde \omega}(\bar{\bar{g}}^n{(\tilde{\omega},\cdot)}=l \vert \bar{\alpha}^n_i{(\tilde{\omega},\cdot)}=k, \ \bar{g}_i^n{(\tilde{\omega},\cdot)}=J, \ \check{p}^n{(\tilde{\omega},\cdot)})(\hat{\omega})=\theta_{kl}^n(\tilde \omega, \check{p}^{n}(\tilde{\omega},\hat{{\omega}})).
\end{equation}
This also implies that
\begin{equation} \label{eq:MatchingCondProb2}
	\hat P^{\tilde \omega}(\bar{\bar{g}}^n{(\tilde{\omega},\cdot)}=J \vert \bar{\alpha}^n_i{(\tilde{\omega},\cdot)}=k, \ \bar{g}_i^n{(\tilde{\omega},\cdot)}=J, \ \check{p}^n{(\tilde{\omega},\cdot)})(\hat{\omega})=1-\sum_{l \in S} \theta_{kl}^n(\tilde{\omega},\check{p}^n(\tilde{\omega},{\hat{\omega}}))=b^k(\tilde{\omega},\check{p}^n(\tilde{\omega},{\hat{\omega}})).
\end{equation}
\end{enumerate}
The extended type of agent $i$ after the random matching step is 
$$\bar{\bar{\beta}}^n_i{(\omega)}=(\bar{\alpha}_i^n{(\omega)},\bar{\bar{g}}_i^n{(\omega)}), \quad n \geq 1.$$ 
{We denote the post-matching extended type distribution realized in ${ \omega} \in {\Omega}$ by $\check{\check{p}}^n(\omega)=(\check{\check{p}}^n(\omega)[k,l])_{k \in S,  l \in S \cup \black{\lbrace J \rbrace}}$, where 
	\begin{equation}\label{eq:PostMatchingExtendedType}
	\check{\check{p}}^n(\omega)[k,l]:=\lambda(\lbrace i \in I: \bar{\bar{\alpha}}^n(i,\omega)=k, \ \bar{g}^n(i,\omega)=l\rbrace).
	\end{equation}}

\textbf{Type changes of matched agents with break-up:} We now define a random matching $\pi^n$ {by}
\begin{equation}
\pi^n(i)=\begin{cases} \label{eq:BreakUpPi}
	\bar{\pi}^n(i) & \text{ if } \bar\pi^n(i) \neq i\\ 
		i& \text{ if } \bar\pi^n(i) = i.
	\end{cases}
\end{equation}
We then introduce an $(\cal{I} \boxtimes {\mathcal{F}})$-measurable agent-type function $\alpha^n$ and an $(\cal{I} \boxtimes {\mathcal{F}})$-measurable partner\black{-type} function $g^n$ {with} 
$$
g^n(i,{\omega})=\alpha^n(\pi^n(i,{\omega}), {\omega}), \quad n \geq 1,
$$
for all $(i,{\omega}) \in I \times {\Omega}$. Given $\tilde{\omega} \in \tilde{\Omega}$, $\check{\check{p}}^n \in \hat{\Delta}$, for any $k_1,k_2,l_1,l_2 \in S$ and $r \in S \cup \lbrace J \rbrace$, for $\lambda$-almost every agent $i$, and for $\hat{P}^{\tilde{\omega}}$-almost every $\hat \omega$, we set 
\begin{align}
	\hat P^{\tilde \omega}\left(\alpha_i^n{(\tilde{\omega},\cdot)}=l_1, \  g_i^n(\tilde{\omega},\cdot)=r \vert \bar{\alpha}_i^n{(\tilde{\omega},\cdot)}=k_1, \ \bar{\bar{g}}^n_i{(\tilde{\omega},\cdot)}=J\right)\blue{(\hat{\omega})}=\delta_{k_1}(l_1) \delta_J(r), \label{eq:BreakUpCondProb0.1}
\end{align}
\begin{align}
	&\hat P^{\tilde \omega}\left(\alpha_i^n{(\tilde{\omega},\cdot)}=l_1, \ g_i^n{(\tilde{\omega},\cdot)}=l_2 \vert \bar{\alpha}_i^n{(\tilde{\omega},\cdot)}=k_1, \ \bar{\bar{g}}^n_i{(\tilde{\omega},\cdot)}=k_2, \ \check{\check{p}}^n{(\tilde{\omega},\cdot)} \right)\blue{(\hat{\omega})} \nonumber \\ &\quad =\left(1-\xi_{k_1k_2}( \tilde \omega, n, \check{\check{p}}^n(\tilde{\omega}, {\hat{\omega}}))\right) \sigma_{k_1 k_2}[l_1,l_2](\tilde \omega, n, \check{\check{p}}^n(\tilde{\omega},{\hat{\omega}})), \label{eq:BreakUpCondProb1} 
\end{align}
\begin{align}
	&\hat P^{\tilde \omega}\left(\alpha_i^n{(\tilde{\omega},\cdot)}=l_1, \  g_i^n{(\tilde{\omega},\cdot)}=J \vert \bar{\alpha}_i^n{(\tilde{\omega},\cdot)}=k_1, \ \bar{\bar{g}}^n_i{(\tilde{\omega},\cdot)}=k_2, \ \check{\check{p}}^n {(\tilde{\omega},\cdot)} \right)\blue{(\hat{\omega})} \nonumber \\& \quad=\xi_{k_1k_2}(\tilde \omega, n,\check{\check{p}}^n(\tilde{\omega},{\hat{\omega}})) \varsigma_{k_1 k_2}^n[l_1](\tilde \omega, n, \check{\check{p}}^n(\tilde{\omega},{\hat{\omega}})). \label{eq:BreakUpCondProb2}
	\end{align}
	The extended-type function at the end of the period is $$\beta^n{({\omega})}=(\alpha^n{({\omega})},g^n{({\omega})}), \quad n \geq 1.$$ 
\end{definition}	
{We denote the extended type distribution at the end of period $n$ realized in ${ \omega} \in {\Omega}$ by ${\hat{p}}^n(\omega)=({\hat{p}}^n(\omega)[k,l])_{k \in S,  l \in S \cup J}$, where 
	\begin{equation}\label{eq:ExtendedTypeEndOfPeriod}
	\hat{{p}}^n(\omega)[k,l]:=\lambda(\lbrace i \in I: {\alpha}^n(i,\omega)=k, \ {g}^n(i,\omega)=l\rbrace).
	\end{equation}}

\begin{remark} \label{remark:Interpretatation}
We provide the intuition behind Definition \ref{def:indsys}. If we set
$$
P(\hat A) := P( \tilde \Omega \times \hat{A})
$$
for any  $\hat A \in \hat{\mathcal{F}}$, then by \eqref{eq:ProbhatPSetting} we have that
\begin{align}
P\left(\hat{A} \big|\bar{\mathcal{F}}\right)(\tilde{\omega}_1) := P\left(\hat A  \times \tilde \Omega \big| \bar{\mathcal{F}}\right)(\tilde{\omega}_1) &= \int_{\tilde \Omega} \hat P^{\tilde \omega}(\hat A) \tilde{P}\left(d\tilde \omega| \bar{\mathcal{F}}\right)(\tilde{\omega}_1) = \hat P^{\tilde \omega_1}(\hat A)\notag
\end{align}
 for every $ \tilde{\omega}_1 \in \tilde{\Omega}$.  Hence, for fixed $\tilde\omega \in \tilde\Omega$, the probabilities {$\hat P^{\tilde  \omega}(\cdot)$ which appear in Definition \ref{defi:DynamicalSystemDiscrete}} might be regarded as conditional probabilities on the product space $\Omega = \hat \Omega \times \tilde \Omega$ endowed with the $\sigma$-algebra $\bar{\mathcal{F}}=\lbrace \emptyset, \Omega \rbrace \otimes \tilde{\mathcal{F}}.$ 
\end{remark}

In the following definition we describe a dynamical system $\mathbb{D}$ which satisfies additional conditional independence assumptions.

\begin{definition}\label{def:indsys}
	A dynamical system $\mathbb{D}$ as in Definition \ref{defi:DynamicalSystemDiscrete} is \emph{Markov conditionally independent (MCI) given $\tilde{\omega} \in \tilde{\Omega}$} if for $\lambda$-almost every $i$ and $j$, for $\hat P^{\tilde{\omega}}$-almost every $\hat \omega \in \hat \Omega$, for every period $n\geq 1$, and for all $k_1,k_2 \in S$, $l_1,l_2 \in S \cup \lbrace J \rbrace$, the following properties hold:
	\begin{enumerate}
		\item Initial dependence: $\beta_i^0$ and $\beta_j^0$ are independent. 
		\item Markov and independent mutation, conditional to $\tilde \omega $:
		\begin{align}
			&\hat P^{\tilde \omega}\left(\bar{\beta}_i^n{(\tilde{\omega},\cdot)}=(k_1,l_1), \ \bar{\beta}_j^n{(\tilde{\omega},\cdot)}=(k_2,l_2) \Big\vert (\beta_{i}^t{(\tilde{\omega},\cdot)})_{t=0}^{n-1}, \ (\beta_j^t{(\tilde{\omega},\cdot)})_{t=0}^{n-1}\right){(\hat{\omega})}  \nonumber \\
			&\black{=} \hat P^{\tilde \omega}\left(\bar{\beta}_i^n{(\tilde{\omega},\cdot)}=(k_1,l_1) \Big\vert \beta_{i}^{n-1}{(\tilde{\omega},\cdot)}\right){(\hat{\omega})}\hat P^{\tilde \omega}\left(\bar{\beta}_j^n{(\tilde{\omega},\cdot)}=(k_2,l_2) \Big\vert \beta_{j}^{n-1}{(\tilde{\omega},\cdot)}\right){(\hat{\omega})}. \label{eq:MCIMutation}
		\end{align}
		\item Markov and independent random matching, conditional to $\tilde \omega $:
		\begin{align}
			&\hat P^{\tilde \omega}\left(\bar{\bar{\beta}}_i^n{(\tilde{\omega},\cdot)}=(k_1,l_1), \ \bar{\bar{\beta}}_j^n{(\tilde{\omega},\cdot)}=(k_2,l_2) \Big\vert \bar{\beta}_{i}^n{(\tilde{\omega},\cdot)}, \ \bar{\beta}_{j}^n{(\tilde{\omega},\cdot)}, \ (\beta_{i}^t{(\tilde{\omega},\cdot)})_{t=0}^{n-1}, \ (\beta_j^t{(\tilde{\omega},\cdot)})_{t=0}^{n-1}\right){(\hat{\omega})} \nonumber \\
			&\black{=} \hat P^{\tilde \omega}\left(\bar{\bar{\beta}}_i^n{(\tilde{\omega},\cdot)}=(k_1,l_1) \Big\vert \bar{\beta}_{i}^{n}{(\tilde{\omega},\cdot)}\right){(\hat{\omega})}\hat P^{\tilde \omega}\left(\bar{\bar{\beta}}_j^n{(\tilde{\omega},\cdot)}=(k_2,l_2) \Big\vert \bar{\beta}_{j}^{n}{(\tilde{\omega},\cdot)}\right){(\hat{\omega})}.\label{eq:MCIMatching}
		\end{align}
		\item Markov and independent matched-agent type changes with break-ups,  conditional to $\tilde \omega $:
			\begin{align}
			&\hat P^{\tilde \omega}\left(\beta_i^n{(\tilde{\omega},\cdot)}=(k_1,l_1), \beta_j^n{(\tilde{\omega},\cdot)}=(k_2,l_2) \Big\vert \bar{\bar{\beta}}_{i}^n{(\tilde{\omega},\cdot)}, \bar{\bar{\beta}}_{j}^n{(\tilde{\omega},\cdot)}, (\beta_{i}^t{(\tilde{\omega},\cdot)})_{t=0}^{n-1}, (\beta_j^t{(\tilde{\omega},\cdot)})_{t=0}^{n-1}\right){(\hat{\omega})} \nonumber\\
			&\black{=} \hat P^{\tilde \omega}\left(\beta_i^n{(\tilde{\omega},\cdot)}=(k_1,l_1) \Big\vert \bar{\bar{\beta}}_{i}^{n}{(\tilde{\omega},\cdot)}\right)(\hat{\omega})\hat P^{\tilde \omega}\left(\beta_j^n{(\tilde{\omega},\cdot)}=(k_2,l_2) \Big\vert \bar{\bar{\beta}}_{j}^{n}{(\tilde{\omega},\cdot)}\right){(\hat{\omega})}. \label{eq:MCIBreakUp}
		\end{align}
	\end{enumerate}
\end{definition}

We now prove the existence of a MCI random matching by using the same arguments as in \cite{RandomMatchingDiscrete}. The proof relies on the product structure of the space $\Omega$ in \eqref{eq:ProductSpace} and the Markov kernel $P$ in \eqref{eq:ProbhatPSetting}, as well as on concepts from nonstandard analysis. Note here that an object with an upper left star means the transfer of a standard object to the nonstandard universe. For a detailed overview of the necessary tools of nonstandard analysis, we refer to Appendix D.2. in \cite{RandomMatchingDiscrete}. \\
From now on, we work under the following assumption.
\begin{assumption} \label{assum:ExistenceHyperfinite}
	Let $(\tilde{\Omega}, \tilde{\mathcal{F}}, \tilde{P})$ be the probability space \black{in \eqref{eq:ProductSpace}}. We assume that there exists its corresponding hyperfinite internal probability space, which we denote from now on also by $(\tilde{\Omega}, \tilde{\mathcal{F}}, \tilde{P})$ by a slight notational abuse. 
\end{assumption}
\begin{remark}
	It is possible to construct such a space whose corresponding Loeb space can be transferred to a classical standard probability space, as it can be seen for example in \cite{RandomMatchingDiscrete}, where $(\tilde{\Omega}, \tilde{\mathcal{F}}, \tilde{P})$ is the space of trajectories of a multi-dimensional Markov process. 
\end{remark}
  
In the \black{outlined} setting we now prove the existence of a {rich} Fubini extension $(I \times \Omega, \cal{I} \boxtimes \mathcal{F}, \lambda \boxtimes P)$, on which a dynamical system $\mathbb{D}$ described in Definition \ref{defi:DynamicalSystemDiscrete} for such input probabilities is defined. 

\begin{remark}
Before we state the main results, we give an intuition how the product structure and the Markov kernel allows us to use the same arguments as in \cite{RandomMatchingDiscrete}. If we fix the state of the world $\tilde{\omega},$ the input functions $(\eta, \theta, \xi, \sigma, \varsigma)$ are deterministic and we are in the setting of \cite{RandomMatchingDiscrete}. Thus, by the results in \cite{RandomMatchingDiscrete} the existence of a space $\hat{\Omega}$ and $\hat{P}$ follows directly. Motivated by this, we fix $\tilde{\omega}\in \tilde{\Omega}$ and construct for each $\tilde{\omega} \in \tilde{\Omega}$ a measure $\hat{P}^{\tilde{\omega}}$. By using the definition of the Markov kernel, we can then define a measure $P$ on the product space.
\end{remark}

In a first step, we focus on the random matching step and show the existence of a suitable hyperfinite probability space and partial matching. This is a generalization of Lemma 7 in \cite{RandomMatchingDiscrete}.
\begin{proposition} \label{OnlyMatchingDiscrete}
Let $(I, \cal{I}_0, \lambda_0)$ be a hyperfinite counting probability space with Loeb space $(I, \cal{I}, \lambda)$. Then, there exists a hyperfinite internal set $\Omega$ with internal power set $\mathcal{F}_0$ such that for any initial internal type function $\alpha^0:I \to S,$ any initial internal partial matching $\pi^0:I \to I$ with
	\[g^0(i)=\begin{cases}
\alpha^0(\pi^0(i)) & \text{ if } \pi^0(i) \neq i\\
J& \text{ if } \pi^0(i) = i,
\end{cases}\]
 and any matching probability function $\theta_{kl}: \tilde{\Omega} \times \mathbb{N} \times \hat{\Delta} \to [0,1]$ and $\hat{p} \in \leftidx{^*}{\hat{\Delta}}$, there exist an internal random matching $\pi$ from $I \times \Omega$ to $I$ and an internal probability measure $P_0$ with the following properties.
\begin{enumerate}
	\item It holds
	$$\Omega:=\tilde{\Omega} \times \hat{\Omega} \quad \text{ and } \quad \mathcal{F}_0:= \tilde{\mathcal{F}} \times \mathcal{\hat{F}}_0,$$ where $\hat{\Omega}$ is a hyperfinite internal set, $\mathcal{\hat{F}}_0$ its internal power set and $(\tilde{\Omega}, \tilde{\mathcal{F}})$ is the \black{hyperfinite internal probability} space \black{which exists by Assumption \ref{assum:ExistenceHyperfinite}.}
	\item We have $P_0:=\tilde{P} \ltimes \hat P_0^{\hat{p}}$, where $\hat P_0^{\hat{p}}$ is a Markov kernel from $\tilde \Omega$ to $\hat \Omega$. From now on, we denote $\hat P_0^{\hat{p},\tilde{\omega}}:=\hat P_0^{\hat{p}}(\tilde \omega)$. 
	\item The internal random matching $\pi: I \times \Omega \to I$ is defined as 
	$$
	\pi(i, (\tilde{\omega}, \hat{\omega})):=\hat{\pi}(i,\hat{\omega}),
	$$ 
	where $\hat{\pi}: I \times \hat{\Omega} \to I $ is an internal random matching. We use the notation 
	\begin{equation} \label{eq:NotationPiFixed}
	\pi_{\hat{\omega}}(i):=\hat{\pi}(i,\hat{\omega})=\pi(i, (\tilde{\omega}, \hat{\omega}))
	\end{equation}
	 for any $\omega =(\tilde{\omega},\hat{\omega}) \in \Omega$. 
	\item Let $H=\lbrace i: \pi^0(i) \neq i \rbrace$. Then 
	$$
	{\hat P_0}^{\hat{p},\tilde{\omega}}\left(\lbrace {\hat{\omega}} \in {\hat{\Omega}}: \pi_{\hat{\omega}}(i)=\pi^0(i) \text{ for any } i \in H \rbrace\right)=1
	$$
	 for any $\hat{p} \in \leftidx{^*}{\hat{\Delta}}$ and $\tilde{\omega}\in\tilde{\Omega}$. 
	\item  The internal mapping from $I \times {\Omega}$ to $S \cup \lbrace J \rbrace$ is defined by the immersion 
	$$
	g(i,(\tilde{\omega}, \hat{\omega})):=\hat{g}(i,\hat{\omega}),
	$$
	 where $\hat{g}$ is the internal mapping from $I \times \blue{\hat{\Omega}}$ to $S \cup \lbrace J \rbrace$, given by
		\[\hat{g}(i,{\hat{\omega}})=\begin{cases}
	\alpha^0(\hat{\pi}(i, {\hat{\omega}})) & \text{ if } \hat{\pi}(i,{\hat{\omega}}) \neq i\\
		J& \text{ if } \hat{\pi}(i,{\hat{\omega}}) = i,
\end{cases}\]
for any $(i,{\hat{\omega}}) \in I \times {\hat{\Omega}}$. Then for any $k,l \in S$ and fixed $\hat{p} \in \leftidx{^*}{\hat{\Delta}}$ and $\tilde{\omega} \in \tilde{\Omega}$ we have $$ \hat{P}_0^{\hat{p},\tilde{\omega}}(\hat{g}_i=l) \simeq \theta_{kl}(\tilde{\omega}, 0,\hat{p})$$ for $\lambda$-almost every agent $i \in I$ satisfying $\alpha^0(i)=k$ and $\pi^0(i)=i$. 
\item For any $\hat{p} \in \leftidx{^*}{\hat{\Delta}}$ and $\tilde{\omega}\in\tilde{\Omega}$, denote the corresponding Loeb probability spaces of the internal probability space{s} $({\hat{\Omega}}, \mathcal{\hat{F}}_0,\hat{P}_0^{\hat{p},\tilde{\omega}})$ and $(I \times {\hat{\Omega}}, \cal{I}_0 \otimes \mathcal{\hat{F}}_0, \lambda_0 \otimes {\hat{P}}_0^{\hat{p},\tilde{\omega}})$ by $({\hat{\Omega}}, \mathcal{\hat{F}}, \blue{\hat{P}}^{\hat{p},\tilde{\omega}})$ and $(I \times \blue{\hat{\Omega}}, \cal{I} \boxtimes \mathcal{\hat{F}}, \lambda \boxtimes \blue{\hat{P}}^{\hat{p},\tilde{\omega}})$, respectively. Moreover, {denote} the corresponding Loeb probability spaces of the internal probability space{s} $(\Omega, \mathcal{F}_0,P_0)$ and $(I \times \Omega, \cal{I}_0 \otimes \mathcal{F}, \lambda_0 \otimes P_0)$ by $(\Omega, \mathcal{F},P)$ and $(I \times \Omega, \cal{I} \boxtimes \mathcal{F}, \lambda \boxtimes P)$, {respectively}. The mapping $\hat{g}$ is an essentially pairwise independent random variable from $(I \times \blue{\hat{\Omega}}, \cal{I} \boxtimes \mathcal{F}, \lambda \boxtimes \blue{\hat{P}}^{\hat{p},\tilde{\omega}})$ to $S \cup \lbrace J \rbrace$ for any $\hat{p} \in \leftidx{^*}{\hat{\Delta}}$ and $\tilde{\omega}\in\tilde{\Omega}$.
\end{enumerate}
\end{proposition}

\begin{proof}
See Appendix \ref{sec:AppendixProof}.
\end{proof}

We are now ready to give the following theorem, which is the main result of the section.
\begin{theorem} \label{theorem:ConstructionSpaceExtensionDiscrete}
Let Assumption \ref{assum:ExistenceHyperfinite} hold and $(\eta_{kl},\theta_{kl}, \xi_{kl}, \sigma_{kl}[r,s], \varsigma_{kl}[r])_{k,l,r,s \in S \times S \times S \times S }$ be \black{some} input functions. Then for any extended type distribution $\ddot{p} \in \hat{\Delta}$ and any deterministic initial condition $\Pi^0=(\alpha^0, \pi^0)$ there exists a {rich} Fubini extension $(I \times \Omega, \cal{I} \boxtimes \mathcal{F}, \lambda \boxtimes P)$ on which a discrete dynamical system $\mathbb{D}=\left(\Pi^{n}\right)_{n=\black{1}}^{\infty}$ as in Definition \ref{defi:DynamicalSystemDiscrete} can be constructed with {discrete time input processes $(\eta^n,\theta^n,\xi^n, \sigma^n, \varsigma^n)_{n\geq 1}$ coming from $(\eta_{kl},\theta_{kl}, \xi_{kl}, \sigma_{kl}[r,s], \varsigma_{kl}[r])_{k,l,r,s \in S \times S \times S \times S }$}. In particular, 
$$  \Omega = \tilde{\Omega} \times  \hat{\Omega}, \quad \mathcal{F} = \tilde{\mathcal{F}} \otimes  \hat{\mathcal{F}},\quad  P = \tilde P \ltimes \hat P,$$
where $(\hat \Omega, \hat{\mathcal{F}})$ is a measurable space and $\hat P$ a Markov kernel from $\tilde \Omega$ to $\hat \Omega$. The dynamical system $\mathbb{D}$ is also MCI according to Definition \ref{def:indsys} and with initial cross-sectional extended type distribution {$\hat{p}^0$ equal to $\ddot{p}^0$ with probability one.} 
\end{theorem}

\begin{proof}
See Section 2 in \cite{biagini_mazzon_oberpriller_supplement}.
\end{proof}

\black{We now} state some properties of the dynamical system $\mathbb{D}$ with input processes, which is a generalization of the results in Appendix C in \cite{RandomMatchingDiscrete}. In particular, given $\tilde{\omega} \in \tilde{\Omega}$ the following result allows to recursively calculate $\hat{p}^n, \check{p}^n, \check{\check{p}}^n$ for every $n \geq 1$, which will be useful for applications. 

 For each time $n \geq 1$ we define a map $\Gamma^n:{\tilde{\Omega}} \times \hat{\Delta} \to \hat{\Delta} $ as follows 
\begin{align}
	\Gamma^n_{kl}({\tilde{\omega}},\hat{p}) &=\sum_{k_1, l_1 \in S} \left(1-\xi_{k_1 l_1}\left({\tilde{\omega}},n,\tilde{\tilde{p}}^n\right)\right) \sigma_{k_1 l_1}[k,l]\left({\tilde{\omega}},n,\tilde{\tilde{p}}^n\right) \tilde{p}^n_{k_1l_1} \nonumber  \\
	&\quad+\sum_{k_1, l_1 \in S} \left(1-\xi_{k_1 l_1}\left({\tilde{\omega}},n,\tilde{\tilde{p}}^n\right)\right) \sigma_{k_1 l_1}[k,l]\left({\tilde{\omega}},n,\tilde{\tilde{p}}^n\right) {\theta}_{k_1l_1}\left({\tilde{\omega}},n,\tilde{p}^n\right)\tilde{p}^n_{k_1J}, \label{eq:DefinitionGamma}
\end{align}
and 
{
\begin{align}
\Gamma^n_{kJ}({\tilde{\omega}},\hat{p})&=  b_{k}\left(\tilde{\omega},n,\tilde{p}^n\right)\tilde{p}_{kJ}^n + \sum_{k_1,l_1 \in S} \xi_{k_1l_1}\left(\tilde{\omega},n,\tilde{\tilde{p}}^n\right) \varsigma_{k_1l_1}[k]\left(\tilde{\omega}, \tilde{\tilde{p}}^n\right)\tilde{p}^n_{k_1l_1} \nonumber \\
&\quad +\sum_{k_1, l_1 \in S} \xi_{k_1 l_1  }\left(\tilde{\omega},n, \tilde{\tilde{p}}^n\right) \varsigma_{k_1l_1}[k]\left(\tilde{\omega},n,\tilde{\tilde{p}}^n\right) \theta_{k_1 l_1}\left(\tilde{\omega},n,\tilde{p}^n\right)\tilde{p}_{k_1J}^n \label{eq:x}
\end{align}}
with
\begin{align*}
	\tilde{p}_{kl}^n&= \sum_{k_1,l_1 \in S } \eta_{k_1 k}\left({\tilde{\omega}},n,\hat{p}\right) \eta_{l_1 l}\left({\tilde{\omega}},n,\hat{p}\right)\hat{p}_{k_1l_1} \\
	\tilde{p}_{kJ}^n&=\sum_{l \in S} \hat{p}_{lJ} \eta_{lk}\left({\tilde{\omega}},n,\hat{p}\right),
\end{align*}
and 
\begin{align*}
	\tilde{\tilde{p}}_{kl}^n&=\tilde{p}^n_{kl}+\theta_{kl}\left({\tilde{\omega}},n,\tilde{p}^n\right)\tilde{p}^n_{kJ} \\
	\tilde{\tilde{p}}_{kJ}^n&=b_k\left({\tilde{\omega}},n,\tilde{p}^n\right) \tilde{p}^n_{kJ}.
\end{align*}
\begin{theorem} \label{theorem:PropertiesDiscrete}
	Assume that the discrete dynamical system $\mathbb{D}$ introduced in Definition \ref{defi:DynamicalSystemDiscrete} is Markov conditionally independent given $\tilde{\omega} \in \tilde{\Omega}$ according to Definition \ref{def:indsys}. Given $\tilde{\omega} \in \tilde{\Omega}$, the following holds:
\begin{enumerate}
	\item For each $n \geq 1$, $\mathbb{E}^{\hat{P}^{\tilde{\omega}}}[\hat{p}^n]= \Gamma^n ({\tilde{\omega}},\mathbb{E}^{\hat{P}^{\tilde{\omega}}}[\hat{p}^{n-1}])$.
	\item For each $n \geq 1, $
		\begin{align*}
		\mathbb{E}^{\hat{P}^{\tilde{\omega}}}[\check{p}^n_{kl} ]=\sum_{k_1, l_1 \in S}  \eta_{k_1,k}\left({\tilde{\omega}},n,\mathbb{E}^{\hat{P}^{\tilde{\omega}}}[\hat{p}^{n-1}]\right)\eta_{l_1,l}\left({\tilde{\omega}},n,\mathbb{E}^{\hat{P}^{\tilde{\omega}}}[\hat{p}^{n-1}]\right) \mathbb{E}^{\hat{P}^{\tilde{\omega}}}[\hat{p}^{n-1}_{k_1,l_1}]
	\end{align*}
	and
	\begin{align} \nonumber
		\mathbb{E}^{\hat{P}^{\tilde{\omega}}}[\check{p}_{kJ}^n]=\sum_{k_1 \in S} \eta_{k_1,k}\left({\tilde{\omega}},n,\mathbb{E}^{\hat{P}^{\tilde{\omega}}}[\hat{p}^{n-1}]\right)  \mathbb{E}^{\hat{P}^{\tilde{\omega}}}[\hat{p}^{n-1}_{k_1,J}].
	\end{align}
	\item For each $n \geq 1$, 	
	\begin{align} \nonumber
		\mathbb{E}^{\hat{P}^{\tilde{\omega}}}[\check{\check{p}}^n_{kl}] =\mathbb{E}^{\hat{P}^{\tilde{\omega}}}[\check{p}^n_{kl}]+ {\theta}_{kl}\left({\tilde{\omega}},n,\mathbb{E}^{\hat{P}^{\tilde{\omega}}}[\check{p}^n ]\right)\mathbb{E}^{\hat{P}^{\tilde{\omega}}}[\check{p}^n_{kJ}]
	\end{align}
	and
	\begin{align} \nonumber
	\mathbb{E}^{\hat{P}^{\tilde{\omega}}}[\check{\check{p}}^n_{kJ} ] 
		={b}_{k}\left({\tilde{\omega}},n,\mathbb{E}^{\hat{P}^{\tilde{\omega}}}[\check{p}^n ]\right) \mathbb{E}^{\hat{P}^{\tilde{\omega}}}[\check{p}^n_{kJ}].
	\end{align}
\item For $\lambda$-almost every agent $i$, the extended-type process $\lbrace \beta_i^n \rbrace_{n=0}^{\infty}$ is a Markov chain in $\hat{S}$ on $(I \times \hat{\Omega}, \cal{I} \boxtimes \hat{\mathcal{F}}, \lambda \boxtimes \hat{P}^{\tilde{\omega}}),$ whose transition matrix $z^n$ at time $n-1$ is given by
\begin{align}
	z^n_{(k'J)(kl)}(\tilde \omega)&= \sum_{k_1, l_1,k' \in S}\left(1-\xi_{k_1 l_1}({\tilde{\omega}},n,\tilde{\tilde{p}}^{\tilde{\omega},n})\right) \sigma_{k_1 l_1}[k,l]\left({\tilde{\omega}},n,\tilde{\tilde{p}}^{\tilde{\omega},n}\right) {\theta}_{k_1l_1}\left({\tilde{\omega}},n,\tilde{p}^{\tilde{\omega},n}\right) \nonumber \\
	 &\quad \cdot {\eta}_{k' k_1}\left({\tilde{\omega}},n,\mathbb{E}^{\hat{P}^{\tilde{\omega}}}[\hat{p}^{n-1}]\right) \label{eq:TransitionMatrix1}\\ 
	z^n_{(k'l')(kl)}(\tilde \omega)&=\sum_{k_1, l_1,k',l' \in S} \left(1-\xi_{k_1 l_1}\left({\tilde{\omega}},n,\tilde{\tilde{p}}^{\tilde{\omega},n}\right)\right) \sigma_{k_1 l_1}[k,l]\left({\tilde{\omega}},n,\tilde{\tilde{p}}^{\tilde{\omega},n}\right)\eta_{k' k_1}\left({\tilde{\omega}},n,\mathbb{E}^{\hat{P}^{\tilde{\omega}}}[\hat{p}^{n-1}]\right) \nonumber \\
	&\quad \cdot \eta_{l' l_1}\left({\tilde{\omega}},n,\mathbb{E}^{\hat{P}^{\tilde{\omega}}}[\hat{p}^{n-1}]\right)  \label{eq:TransitionMatrix2} \\
	z^n_{(k'l')(kJ)}(\tilde \omega)&=\sum_{k_1,l_1 \in S} \xi_{k_1 l_1}\left({\tilde{\omega}},n,\tilde{\tilde{p}}^{\tilde{\omega},n}\right) \varsigma_{k_1 l_1}[k]\left({\tilde{\omega}},n,\tilde{\tilde{p}}^{\tilde{\omega},n}\right) \nonumber  \\
	& \quad \quad \cdot \eta_{k'k_1}\left({\tilde{\omega}},n,\mathbb{E}^{\hat{P}^{\tilde{\omega}}}[\hat{p}^{n-1}]\right)\eta_{l'l_1}\left({\tilde{\omega}},n,\mathbb{E}^{\hat{P}^{\tilde{\omega}}}[\hat{p}^{n-1}]\right)  \label{eq:TransitionMatrix3}\\
	z_{(k'J)(kJ)}^n(\tilde \omega)&={b}_{k}\left({\tilde{\omega}},n,\tilde{p}^{\tilde{\omega},n}\right) {\eta}_{k'k}\left({\tilde{\omega}},n,\mathbb{E}^{\hat{P}^{\tilde{\omega}}}[\hat{p}^{n-1}]\right)\nonumber\\
	 & \quad +\sum_{k_1,l_1\in S} \xi_{k_1 l_1}\left({\tilde{\omega}},n,\tilde{\tilde{p}}^{\tilde{\omega},n}\right) \varsigma_{k_1 l_1}[k]\left({\tilde{\omega}},n,\tilde{\tilde{p}}^{\tilde{\omega},n}\right)\eta_{k_1l_1}\left({\tilde{\omega}},n,\tilde{p}^{\tilde{\omega},n}\right)\nonumber \\
	 & \quad \quad  \cdot \eta_{k'k_1}\left({\tilde{\omega}},n,\mathbb{E}^{\hat{P}^{\tilde{\omega}}}[\hat{p}^{n-1}]\right).  \label{eq:TransitionMatrix4}
\end{align}
\item For $\lambda$-almost every $i$ and every $\lambda$-almost every $j$, the Markov chains $\lbrace \beta_i^n \rbrace_{n=0}^{\infty}$ and $\lbrace \beta_j^n \rbrace_{n=0}^{\infty}$ are independent on $( \hat{\Omega},  \hat{\mathcal{F}}, \hat{P}^{\tilde{\omega}})$.
\item For $\hat{P}^{\tilde{\omega}}$-almost every $\hat{\omega} \in \hat{\Omega}$, the cross sectional extended type process $\lbrace \beta^n_{\hat{\omega}} \rbrace_{n=0}^{\infty}$ is a Markov chain on $(I, \cal{I}, \lambda)$ with transition matrix $z^n$ at time $n-1$, which is defined in \eqref{eq:TransitionMatrix1}- \eqref{eq:TransitionMatrix4}.
\item We have $\hat{P}^{\tilde{\omega}}$-a.s. that
\begin{align*}
	\mathbb{E}^{\hat{P}^{\tilde{\omega}}}[\check{p}^n_{kl}]=\check{p}^n_{kl}, \quad
	\mathbb{E}^{\hat{P}^{\tilde{\omega}}}[\check{\check{p}}^n_{kl}]=\check{\check{p}}^n_{kl} \quad \text{ and } \quad
	\mathbb{E}^{\hat{P}^{\tilde{\omega}}}[\hat{p}^n_{kl}]&=\hat{p}^n_{kl}.
\end{align*}
\end{enumerate}
\end{theorem}
\begin{proof}
See Section 3 in \cite{biagini_mazzon_oberpriller_supplement}. 
\end{proof}
\begin{remark}
By using Remark \ref{remark:Interpretatation}, the results in Theorem \ref{theorem:PropertiesDiscrete} also hold on the product space $\Omega = \hat{\Omega} \times \tilde{\Omega}$ by conditioning on ${\bar{\mathcal{F}}}$ with respect to the probability measure $P$.
\end{remark}

\section{Application: A dynamic directed random matching model for the evolution of asset price bubbles} \label{sec:tradingvolume}
We now use the random matching mechanism described in Section \ref{sec:main} to describe interactions among investors in the setting of Section \ref{sec:FormationOfABubble}.  
We consider an atomless probability space $(I,\cal{I},\lambda)$ representing the space of investors. We introduce the space of investors' types $S=\{1,2,3\}$, where investors of type $1$, $2$, $3$ are respectively optimistic, neutral and pessimistic, and the space $\hat{\Delta} =(\green{p_{ij}})_{i=1,2,3, j=1,2,3,J}$ of \black{processes with values in the space of matrices with $3$ rows and $4$ columns\footnote{\black{This space can be endowed with a topology on the space of matrices.}}} representing the extended type distributions. We denote with $\green{p_{ij}^{{n}}}$ the fraction of investors of type $i=1,2,3$ {at time $\green{t_n}$} matched with a partner of type $j=1,2,3$, and with \green{$p_{iJ}^{{n}}$} the fraction of unmatched agents of type $i$ \green{at time $t_n$}. 
For any $k=0,\dots,N$, we then have 
\begin{equation}\label{eq:pi3}
\green{p^k_i = \sum_{j=1}^3 p^{k}_{ij}+p_{iJ}^k.}
\end{equation}
Let $(\tilde{\Omega}, \tilde{\mathcal{F}}, \tilde{\mathbb{F}}=(\tilde{\mathcal{F}}^{\green{i}})_{i=0,...,N}, \tilde{P})$ be a filtered probability space from Section \ref{sec:tradingimpact} on which the stochastic processes $(F^{\green{i}})_{i=0,...N}$, $(M^{\green{i}})_{i=0,...N}$, $(\Lambda^{\green{i}})_{i=0,...N}$, $(\kappa^{\green{i}})_{i=0,...N}, (\Theta^{\green{i}})_{i = 0, \dots, N}$ are  defined. Then, Theorem \ref{theorem:ConstructionSpaceExtensionDiscrete} provides the existence of a rich Fubini extension $(I \times \Omega, \cal{I} \boxtimes \mathcal{F}, \lambda \boxtimes P)$ for a discrete dynamical system $\mathbb{D}=\left(\Pi^{n}\right)_{n=\black{1}}^{\infty}$ as in Definition \ref{defi:DynamicalSystemDiscrete} with 
$$  \Omega = \tilde{\Omega} \times  \hat{\Omega}, \quad \mathcal{F} = \tilde{\mathcal{F}} \otimes  \hat{\mathcal{F}},\quad  P = \tilde P \ltimes \hat P,$$
where $(\hat \Omega, \hat{\mathcal{F}})$ is a measurable space on which the stochastic process \green{$\hat{p}^{i}=(\hat p^k_{i})_{k = 0, \dots, N}, i=1,2,3,$} can be constructed as in the proof of Theorem \ref{theorem:ConstructionSpaceExtensionDiscrete}, and $\hat P$ a Markov kernel from $\tilde \Omega$ to $\hat \Omega$ according to Definition \ref{def:MarkovKernel}.

By Theorem \ref{theorem:PropertiesDiscrete}, the dynamics of the fraction of optimistic, neutral and pessimistic agents are identified by the conditional transition probabilities $(\eta_{kl},\theta_{kl}, \xi_{kl}, \sigma_{kl}[r,s], \varsigma_{kl}[r])_{k,l,r,s \in S \times S \times S \times S }$. Next, we specify the form of these functions in this setting.

We do not make any particular assumptions on $(\theta_{kl}, \xi_{kl}, \varsigma_{kl}[r])_{k,l,r,s \in S \times S \times S}$, which identify the conditional probabilities of matching, break-up and post break-up change of type. We just emphasize that their evolution depends on some state of $\tilde \Omega$. We focus instead on modeling the functions $\eta_{k_1,k_2}$ and $\sigma_{k_1 k_2}[l_1,l_2]$, $k_1, k_2, l_1, l_2 =1,2,3,$ appearing in \eqref{eq:IndiMutation1} and \eqref{eq:BreakUpCondProb1}, respectively. Such quantities represent the probabilities of type change, independent of matching \black{or} after a new match, conditional to the state of $\tilde \Omega$ and to the current values of the fractions of investors' types. {They} depend on the state of $\tilde \Omega$, on the values of the current fractions of optimistic, neutral and pessimistic investors, and on the type of the agent to whom they are matched in the case of post-match type change, as follows:
\begin{enumerate}[(i)]
\item A match with an optimistic agent may induce an upgrade in type (i.e. \black{from} $3$ to $1$ or $2$; or from $2$ to $1$). 
\item A match with a pessimistic agent may induce a downgrade of type (i.e. \black{from} $1$ to $2$ or $3$; or from $2$ to $3$). 
\item The propensity of an agent to become \emph{more optimistic}, i.e. to change type from $3$ to $1, 2$ or from $2$ to $1$, is increasing with respect to $p^1-p^3$. On the contrary, the propensity of an agent to become \emph{more pessimistic} is increasing with respect to $p^3-p^1$. This holds for both the type changes, conditional or unconditional to the match.
\end{enumerate}

Before giving an example of such conditional probabilities, we comment on the above model.

\begin{remark}\label{remark:changeviews}
The main features of the model described above are the following:

\begin{enumerate}[(i)]
\item All the processes determining the matching, breaking-up and type change probabilities depend on some underlying stochastic processes on the probability space $(\tilde \Omega, \tilde{\mathcal{F}}, \tilde  P) $. These factors may represent socio-economic indicators as well as the fundamental price of the assets or the impact of public news. In this way, the birth of the bubble may be determined by changes in the driving factors, leading to an higher probability of matching with optimistic traders or of changing type from pessimistic to optimistic and vice versa.

\item Investors can change their views with a probability depending on {the number of} buyers and sellers in the market, {see \eqref{eq:Xp}.} {This assumption reflects that}, with an increasing price, investors {may be} prone to think that the price will grow further at least on the short or medium term. For example in the model of \cite{scheinkman2003overconfidence}  agents are willing to pay prices exceeding their own valuation of the fundamental value of a bubbly asset because of speculation opportunities in the near future.

\item Traders can also influence each other when they meet, with a probability again depending on the difference between the number of buyers and sellers. This reflects some phenomena which are typically observed during \black{both} the blowing-up {and the bursting} \black{phase.} {W}hen the bubble grows, \black{heard behavior may be induced by} a tendency to mimic the actions of acquaintances making gains and by word-of-mouth spread of information regarding the fast increase of the stock price. \black{This will eventually} fuel further up the prices, see \cite{lux1995herd} for a description and formalization of this phenomenon {and  \cite{Bayer2014}, where such mechanisms are documented in the US 2007 Housing bubble.} Similar attitudes, but in a different direction, characterize the investors' behavior starting the burst and {speeding up} the decrease of the price after the burst.

\end{enumerate}
\end{remark}

We conclude the section by giving a concrete example of a possible choice of  $\eta_{k_1,k_2}$ and $\sigma_{k_1 k_2}[l_1,l_2]$, $k_1, k_2, l_1, l_2 =1,2,3$. 

\begin{example} \label{Example1} 
Fix $\tilde \omega \in \tilde \Omega$ and $j \in \{1, \dots, N \}$, representing time. We set 
$$\eta^{\tilde\omega, j}_{k_1,k_2}:=\eta_{k_1 k_2}(\tilde{\omega},j, \cdot)$$ 
and $$\sigma^{\tilde\omega, j}_{k_1 k_2}[l_1,l_2]:=\sigma_{k_1 k_2}[l_1,l_2](\tilde{\omega},j, \cdot).$$
We model such conditional probabilities in the following way.
\begin{enumerate}[(i)]
\item {Consider} an increasing function $f^j:\mathbb{R}_+ \to [0,1/2]$ such that $f^{j}(0)=0$ and set
\begin{equation}\label{fplusfminus}
f^{j}_+=f^{ j}\left(({p_1^{j}-p_3^{j}})^+\right), \qquad f^{j}_-=f^{j}\left(({p_3^{j}-p_1^{j}})^+\right).
\end{equation} 
Given the random variables $F_{ijk}: (\tilde \Omega, \black{\tilde{\mathcal{F}}}) \to \left([0,1/2], \mathcal{B}([0,1/2])\right)$, $i,j,k=1,2,3,$ the type change probabilities after a match are then defined as follows:
\begin{itemize}
\item After a match of two agents of the same type who stay in a relationship we have $$\sigma^{\tilde\omega, j}_{kk}(r,s) =\delta_r(k)\delta_s(k),$$ i.e. they both maintain their types.
\item After a match of two agents who are respectively optimistic and neutral and stay in a relationship, we {assume} 
$$\sigma^{\tilde\omega, j}_{12}(k,\ell)=0 \text{ if } k=3  \text{ or }\ell=3 \quad \text{and} \quad \sigma_{12}^{\tilde\omega, j}(2,1)=0,$$
 i.e. none of the agents changes his type to pessimistic and they do not switch types.
Moreover, we set
$$ 
\sigma^{\tilde\omega, j}_{12}(1,1)=F_{121}(\tilde \omega)+f^{j}_+, \quad \sigma_{12}^{\tilde\omega, j}(2,2)=F_{122}(\tilde \omega)+f^{j}_-, \quad \sigma_{12}^{\tilde\omega, j}(1,2)=1-\sigma^{\tilde\omega, j}_{12}(1,1)-\sigma_{12}^{\tilde\omega, j}(2,2),
$$
{i.e.} the probability that the neutral agent changes to optimistic due to the match is given by a random term depending only on $\tilde \Omega$ plus a term which is strictly positive only if the number of buyers is higher than the number of sellers, and increasing with respect to their difference. The opposite holds for a possible change of the optimistic agent to neutral.
\item After a match of two agents who are respectively neutral and pessimistic and stay in a relationship, we assume that
$$\sigma^{\tilde\omega, j}_{23}(k,\ell)=0 \text{ if } k=1 \text{ or } \ell=1 \quad \text{and}\quad \sigma_{23}^{\tilde\omega, j}(3,2)=0,$$ 
i.e. we exclude the possibility {that} one of them may become optimistic. In analogy to the previous case, we fix
$$ 
\sigma^{\tilde\omega,j}_{23}(2,2)=F_{232}(\tilde \omega)+f^{j}_+, \quad \sigma^{\tilde\omega, j}_{23}(3,3)=F_{233}(\tilde \omega)+f^{j}_-, \quad \sigma^{\tilde\omega, j}_{23}(2,3)=1-\sigma^{\tilde\omega,j}_{23}(2,2)-\sigma^{\tilde\omega, j}_{23}(3,3).
$$
\item {For a} match of two agents who are respectively neutral and pessimistic and stay in a relationship, we {put}
$$\sigma^{\tilde\omega,j}_{13}(3,1)=\sigma^{\tilde\omega, j}_{13}(3,2)=\sigma_{13}^{\tilde\omega,j}(2,2)=0,$$
 {i.e.} the agents do not switch their types and do not become neutral. We also set
\begin{align}
& \sigma_{13}^{\tilde \omega, j}(1,1)=F_{131}(\tilde \omega)+(f^{j}_+)^2, \quad \sigma^{\tilde\omega, j}_{13}(1,2)=F_{132}(\tilde \omega)+f^{j}_+(1-f^{j}_+),  \notag \\ & 
 \sigma_{13}^{\tilde \omega, j}(3,3)=F_{133}(\tilde \omega)+(f^{j}_-)^2,\quad \sigma^{\tilde\omega, j}_{13}(2,3)=F_{132}(\tilde \omega)+f^{j}_-(1-f^{j}_-),\notag \\ &
  \sigma^{\tilde\omega, j}_{13}(1,3)=1-\sigma_{13}^{\tilde \omega, j}(1,1)- \sigma^{\tilde\omega, j}_{13}(1,2)- \sigma_{13}^{\tilde \omega, j}(3,3)-\sigma^{\tilde\omega, j}_{13}(2,3)
\end{align}
consistently with the construction above.
\item {Condition \eqref{eq:conditionsigma} holds, i.e.} $\sigma_{k\ell}^{\tilde \omega, j}(r,s)=\sigma_{\ell k}^{\tilde \omega, j}(s,r)$.\end{itemize}
\item {Consider} an increasing function $g^{j}:\mathbb{R}_+ \to [0,1/2]$ such that $g^{ j}(0)=0$, and introduce $g^{ j}_+$ and $g^{ j}_-$ analogously to \eqref{fplusfminus}.
We define 
\begin{equation}\notag
B^{j}= \left( \begin{array}{ccc}
1 - g^{ j}_- &g^{ j}_-(1-g^{ j}_-) & (g^{ j}_-)^2  \\
g^{ j}_+ & 1 - g^{ j}_+ - g^{ j}_- &g^{ j}_-   \\
(g^{ j}_+)^2 & g^{ j}_+(1-g^{ j}_+) & 1 - g^{ j}_+  \\
 \end{array} \right),
\end{equation}
and then the matrix of the probabilities $\eta^{\tilde\omega, j}_{k, l}=[B^{j}]_{k l}+[C(\omega)]_{k l}$, $k, l  =1,2,3$, where 
$$
C_{ij}(\omega): (\tilde \Omega, \tilde{\mathcal{F}}) \to \black{\left([0,1/2], \mathcal{B}([0,1/2])\right)}
$$
are random variables, $i,j=1,2,3$. 
\end{enumerate}
\end{example}

\begin{example}
We focus again on the post-matching change of type described as in Example \ref{Example1}: agents may change their views after a match because of the information from other traders. In particular, during the blowing-up phase of the bubble, upgrade of types may be induced by matches with optimistic agents (i.e. buyers). In the current example we include the idea that the tendency of an agent to switch to more optimistic forecasts after a match also depends on the number of optimistic agents she has already met (and vice versa).
We include these considerations as follows. 
We let again $N$ be the finite number of time periods. The type\footnote{This is a small notational difference with respect to the setting of Section \ref{sec:main}, where $0$ is not included in the indices set. In this way the representation in \eqref{eq:typeswithnumbertypes} is easier.} of an agent is identified by a number $k \in \tilde{S}:= \{0, 1, \dots, 3 (N+1)^3 -1 \}$ defined as
\begin{equation}\label{eq:typeswithnumbertypes}
k = n_o + n_n (N+1) + n_p (N+1)^2 + (v-1) (N+1)^3,
\end{equation} 
where $n_o$, $n_n$, $n_p$ are the number of optimistic, neutral and pessimistic investors that the agent has met, respectively, and $v=1,2,3$ indicates an optimistic, neutral and pessimistic view, respectively\footnote{Clearly, not all $k$ in \eqref{eq:typeswithnumbertypes} represent feasible types: for example, at the $n$-th step, only types represented by $k = n_o + n_n (N+1) + n_p (N+1)^2 + v (N+1)^3$ with $n_o+n_n+n_p \le n$ are possible. When a type is infeasible, the fraction of agents of that type is zero.}. This means that, after a match, an agent of type $k$ in \eqref{eq:typeswithnumbertypes} changes to type
\begin{equation}\notag
\bar k = \bar{n}_o + \bar{n}_n (N+1) + \bar{n}_p (N+1)^2 + (\bar{v}-1) (N+1)^3,
\end{equation} 
where $\bar{n}_i= n_i+1$ for $i = o,n,p$ and $\bar{n}_j= n_j$ for $j \ne i$. Moreover, $\bar{v} \ne v$ if and only if the investor changes forecasts. 

Note that the views of an agent of type $k$ are immediately identified as optimistic if $k < (N+1)^3$, neutral if $k \in [(N+1)^3, 2(N+1)^3)$ and pessimistic if $ k \ge 2(N+1)^3$.
We can also recover the fractions of optimistic, neutral and pessimistic investors at time $t_k$, denoted above by $\green{p^k_i}$, $i=1,2,3,$ respectively, by
$$
\green{p_i^k = \sum_{l=i (N+1)^3}^{(i+1)(N+1)^3-1} \tilde{p}^k_l, \quad i=1,2,3, \quad k=0,...,N,}
$$
where $\green{\tilde{p}^k_l}$ is the fraction of investors of type $l \in \tilde{S}$ at time $t_k$. We then extend Example \ref{Example1} by defining $\sigma^{\tilde\omega, j}_{k_1 k_2}[l_1,l_2]$, $k_1, k_2, l_1, l_2 \in \tilde S, \ l=1,...N$, by also including the number of optimistic, neutral and pessimistic agents already met: pessimistic agents of type $l =  n_o + n_n (N+1) + n_p (N+1)^2 + 2 (N+1)^3$ which encounter an optimistic agent may switch to neutral or optimistic with a probability increasing with respect to $n_o$. This is in line with Remark \ref{remark:changeviews} (iii). 

As an example, for two agents of types
$$
k_1= n_o + n_n (N+1) + n_p (N+1)^2  + 2 (N+1)^3 \quad \text{ and } \quad k_2= m_o + m_n (N+1) + m_p (N+1)^2,
$$
and defining
$$
l_1= (n_o + 1) + n_n (N+1) + n_p (N+1)^2 \quad \text{ and } \quad l_2= m_o + m_n (N+1) + (m_p+1) (N+1)^2,
$$
we can set
 $$
 \sigma^{\tilde\omega, j}_{k_1 k_2}[l_1,l_2] = \frac{n_o}{n_o+n_n+n_p}F_{131}(\tilde \omega)+(f^{j}_+)^2,
 $$
where $F_{131}(\cdot)$ and $f^{j}_+$ are defined in Example \ref{Example1}. Here,  $\sigma^{\tilde\omega, j}_{k_1 k_2}[l_1,l_2](\tilde{\omega})$ is the probability that a pessimistic agent (note that $k_1 > 2(N+1)^3$) that has already met $n_o$ optimistic, $n_n$ neutral and $n_p$ pessimistic investors and meets a further optimistic investor (note that $ k_2 < (N+1)^3$), becomes optimistic (note that $l_1 < (N+1)^3$). This probability is increasing with respect to $n_o$, so that the more optimistic investors the agent has already met, the more she is prone to switch to optimistic views.
\end{example}

\subsection{Absence of arbitrage}
In \black{this section}, we {provide a setting under which} the financial market model considered in Section \ref{sec:tradingvolume} is arbitrage-free, {i.e{.} it admits} an equivalent martingale measure $Q$ for $S$, see Theorem 1.7 in \cite{follmer_stochastic_finance}. For simplicity, we assume that $\green{\kappa^{k}=0}$ for all $k=\black{0},...,N$. Then, by \eqref{eq:DefMarketPriceProcess} and \eqref{eq:Xp} it follows that 
\begin{align}
	S^{k}&= S^{{k-1}}+F^{k}-F^{{k-1}}+2 \Lambda^{k}M^{k} (X^{k}-X^{{k-1}}) \nonumber \\
	&=S^{{k-1}}+F^{k}-F^{{k-1}}+ 2 \Lambda^{k}M^{k} \left[ \Theta^{k}(p^{k}_1-p^{k}_3)- \Theta^{{k-1}}(p^{{k-1}}_1- p^{{k-1}}_3) \right] \label{eq:RewrittenS}
\end{align}
for any $k=1,...,N$.
For $P=\tilde{P} \ltimes \hat{P} $ given as in \eqref{eq:ProbhatPSetting}, we define a measure $Q$ of the form
\begin{equation} \label{eq:DefinitionQ}
	Q:=\tilde{Q} \ltimes \hat{P},
\end{equation}
where $\tilde{Q}$ is a probability measure on $\tilde{\Omega}$. 

\begin{assumption} \label{eq:IndependencAssumption}
Let \green{$F=(F^{k})_{k=\black{0},...,N}, \Lambda=(\Lambda^{k})_{k=\black{0},...,N}, M=(M^{k})_{k=\black{0},...,N}, \Theta=(\Theta^{k})_{k=\black{0},...,N}$ be stochastic processes defined on $(\tilde{\Omega},\tilde{\mathcal{F}},Q)$} with $Q$ in \eqref{eq:DefinitionQ} such that
\begin{enumerate}
	\item The stochastic processes $\Lambda M:=(\green{\Lambda^{k} M^{k}})_{k=1,...,N}$ and $X=(X^{\green{k}})_{k=1,...N}$ are conditionally independent under $Q$, i.e. $\green{\Lambda^{k}M^{k}}$ and $\green{X^{k}}$ are conditionally independent given $\green{\mathcal{F}^{{k-1}}}$ for every $k=1,...,N$.
	\item For every $i=k,...,N$, $\Theta^{\green{k}}$ and $\green{p^{k}_1-p^{k}_3}$ are conditionally independent under $Q$ given \green{$\mathcal{F}^{{k-1}}$.}
	\item The stochastic processes $F, \Theta$ are $(\tilde{Q},\tilde{\mathbb{F}})$-martingales.
\end{enumerate}
\end{assumption}

\begin{proposition} \label{prop:MarketArbitrageFree}
	Let Assumption \ref{eq:IndependencAssumption} hold. If $Q=\tilde{Q}\ltimes \hat{P}$ is a martingale measure {for} \green{$p_1-p_3:={(p^{{k}}_1-p^{{k}}_3)_{k=0,...{,}N}}$} and $\tilde{Q}$ is equivalent to $\tilde{P}$, then \black{$Q$ is an equivalent martingale measure for the discounted asset price $S$. Hence,} the market model introduced in Section \ref{sec:tradingvolume} is arbitrage-free.
\end{proposition}
\begin{proof}
Let $Q=\tilde{Q}\ltimes \hat{P}$ be a martingale measure for \green{$p_1-p_3$} with respect to $\mathbb{F}$. Under Assumption \ref{eq:IndependencAssumption} and by using \eqref{eq:RewrittenS} we get 
\begin{align}
&\mathbb{E}^Q\left[ S^{k}-S^{{k-1}} \vert \mathcal{F}^{{k-1}} \right] \nonumber \\
&=\mathbb{E}^Q\left[ F^{k}-F^{{k-1}} \vert \mathcal{F}^{{k-1}} \right] +2\mathbb{E}^Q\left[ \Lambda^{k}M^{k} \vert \mathcal{F}^{k-1} \right] \mathbb{E}^Q\left[\Theta^{k}(p_1^k-p_3^k)- \Theta^{k-1}(p^{k-1}_1- p^{{k-1}}_3) \vert \mathcal{F}^{{k-1}} \right] \nonumber \\
&=\mathbb{E}^{\tilde{Q}}\left[ F^k-F^{k-1} \vert \tilde{\mathcal{F}}^{k-1} \right] +2\mathbb{E}^{\tilde{Q}}\left[ \Lambda^{k}M^{k} \vert \tilde{\mathcal{F}}^{k-1} \right] \left( \mathbb{E}^Q\left[\Theta^{k}(p^{k}_1-p^k_3) \vert \mathcal{F}^{{k-1}} \right]- \Theta^{{k-1}}(p^{{k-1}}_1- p^3_{{k-1}})\right) \nonumber \\
&= 2\mathbb{E}^{\tilde{Q}}\left[ \Lambda^{k}M^{k} \vert \tilde{\mathcal{F}}^{{k-1}} \right] \bigg( \mathbb{E}^{\tilde{Q}}\left[\Theta^{k} \vert \tilde{\mathcal{F}}^{{k-1}} \right] \mathbb{E}^Q\left[p^{k}_1-p^{k}_3 \vert \mathcal{F}^{{k-1}} \right]- \Theta^{{k-1}}(p^{{k-1}}_1- p^{{k-1}}_3)\bigg) \nonumber \\
&= 2\mathbb{E}^{\tilde{Q}}\left[ \Lambda^{k}M^{k} \vert \tilde{\mathcal{F}}^{{k-1}} \right] \Theta^{{k-1}}\left(p^{{k-1}}_1- p^{{k-1}}_3-( p^{{k-1}}_1- p^{{k-1}}_3)\right)\nonumber  \\ 
&=0. \label{eq:CalculationsMartingaleProperty}
\end{align}
Moreover, by the definition of $Q$ it is clear that $Q$ is equivalent to $P$ if $\tilde{Q}$ is equivalent to $\tilde{P}$.
\end{proof}
In the following, we analyze under which conditions \green{$p_1-p_3$} is a martingale under $Q$ with respect to $\mathbb{F}$. 
For the sake of simplicity, we assume that the agents break up immediately, by setting $\xi_{kl} = 1$ for any $k, l = 1,2,3$. Thus, we only consider the probability $\varsigma_{kl}^n[r]$ {that when two agents of types $k,l$ break up immediately after matching, the first agent becomes of type $r$. \black{Similarly to} Example \ref{Example1} we choose suitable functions \black{$\tilde\eta_{ij},\theta_{ij},\tilde\varsigma_{ilj}: \tilde{\Omega} \times \mathbb{N} \to [0,1]$, for any $i,j,l=1,2,3$, $i \ne j$, and define}
\begin{align}
	\eta_{ij}(\tilde{\omega},k,\green{p^{{k-1}}})&:=\black{\tilde\eta_{ij}(\tilde{\omega},k)+}f_{ij}\left(\green{p^{{k-1}}_{1J}-p^{{k-1}}_{3J}}\right)\black{,} \label{eq:ExampleFunction1} \\
	\theta_{i\black{l}}(\tilde{\omega},k,\green{\tilde{p}^{{k}}})&:= \theta(\tilde{\omega},k) \tilde{p}^{k}_{\black{l}J}\black{,} \label{eq:ExampleFunction2}\\
	\varsigma_{il}[j](\tilde{\omega}, k,\green{\tilde{p}^{{k}}})&:=	 \black{\delta_{\lbrace l\rbrace}(j)\left(\black{\tilde\varsigma_{ilj}(\tilde{\omega},k)+}g_{ilj}\left(\green{{\tilde{p}}^{{k}}_{1J}-{\tilde{p}}^{{k}}_{3J}}\right)\right)},\label{eq:ExampleFunction3}
\end{align}
for $k=1,\dots,N$, where $f_{ij}, g_{ijl}: \mathbb{R} \to [0,1]$,  in such a way that  \eqref{eq:ThetaCondition} holds and that $\eta_{ij}, \varsigma_{il}[j] \in [0,1]$. Here \green{${p}^{{k-1}}$} is the distribution of agents types after the $(k-1)$-th time step and \green{$\tilde{p}^{{k}}$} is the distribution of types after the random change \black{at} the $k$-th time step. \black{Finally, in view of  \eqref{eq:EtaSumUpTo1} and \eqref{eq:VarSigmaSumUpTto1}, we define
\begin{equation}\label{eq:etaii}
\eta_{ii}(\tilde{\omega},k,\green{p^{{k-1}}}) = 1 - \sum_{j=1, j \ne i}^3 \eta_{ij}(\tilde{\omega},k,\green{p^{{k-1}}}) 
\end{equation}
and 
\begin{equation}\label{eq:varsigmaili}
\varsigma_{il}[i](\tilde{\omega}, k,\green{\tilde{p}^{k}})= 1 - \sum_{j=1, j \ne i}^3 \varsigma_{il}[j](\tilde{\omega}, k,\green{\tilde{p}^{k}}). 
\end{equation}
\black{Note that \eqref{eq:ExampleFunction3} and \eqref{eq:varsigmaili} imply that
\begin{equation}\label{eq:allindicessame}
\varsigma_{ii}[i](\tilde{\omega}, k,\green{\tilde{p}^{{k}}})=1 - \sum_{j=1, j \ne i}^3 \varsigma_{ii}[j](\tilde{\omega}, k,\green{\tilde{p}^{{k}}}) = 1,
\end{equation}
that is, when two agents of the same type meet, they keep their type.}
}
\begin{remark}
\black{In the construction \eqref{eq:ExampleFunction1}-\eqref{eq:ExampleFunction3}, type changes are governed by a stochastic driver defined on the space $\tilde \Omega$ plus a term which only depends on the former distribution of types.} Also note that $\varsigma_{ij}[l]$ depends in general on the distribution $\tilde{\tilde{p}}$ immediately after the break up, which coincides with the distribution $\tilde{p}$ before the matching, as agents break up immediately.
\end{remark}
\begin{lemma}
	Consider a random matching mechanism with immediate break-up, where the functions $\eta_{ij},\theta_{ij},\varsigma_{ij}[l]$ for $i,j,l=1,2,3\black{,}$ are defined in \eqref{eq:ExampleFunction1}-\eqref{eq:ExampleFunction3}. Let $Q$ be given as in \eqref{eq:DefinitionQ}. Then it holds
	\begin{align} \label{eq:SimplificationWithImmediateBreakUp}
	\mathbb{E}^Q\left[p^{k}_1- p^{k}_3 \vert \mathcal{F}^{{k-1}}\right]&=\mathbb{E}^{\tilde{Q}}\left[ \Gamma_{1J}^{k}(\cdot, p^{{k-1}})-\Gamma_{3J}^{k}(\cdot, p^{{k-1}})\vert \tilde{\mathcal{F}}^{{k-1}}\right],
\end{align}
with
\black{\small{\begin{align}
&\Gamma^{k}_{iJ}(\tilde{\omega}, p^{{k-1}})\nonumber \\
&=\left( 1- \theta(\tilde{\omega},k) \right) \black{F_i^{\tilde \omega}(p^{{k-1}})} + \theta(\tilde{\omega},k) \black{F_{i}^{\tilde \omega}(p^{{k-1}})}\black{F_{i}^{\tilde \omega}(p^{{k-1}})} \nonumber \\
& \quad +   \sum_{k_1 = 1, k_1 \ne i}^3  \left(\tilde\varsigma_{k_1 i i}(\tilde{\omega},k)+g_{k_1 i i}\left( \black{F_1^{\tilde \omega}(p^{{k-1}})}-\black{F_{3}^{\tilde \omega}(p^{k-1})}\right) + 1 - \tilde\varsigma_{i k_1 k_1}(\tilde{\omega},k)-g_{i k_1 k_1} \left( \black{F_{\tilde \omega}^1(p^{{k-1}})}-\black{F_{3}^{\tilde \omega}(p^{{k-1}})}\right)\right) \nonumber \\
&\quad \quad \cdot \black{F_{i}^{\tilde \omega}(p^{{k-1}})}\black{F_{k_1}^{\tilde \omega}(p^{{k-1}})}\theta(\tilde{\omega},k),\label{eq:GammaExample}
\end{align}}}
 for $i=1,2,3$ and $k \geq 2$, {where
 \begin{align} \label{eq:DefinitionF^1}
 	\black{F_i^{\tilde \omega}(p^{{k-1}})}&:=  \sum_{l=1, \black{l \ne i}}^3 p^{{k-1}}_{lJ} \left(\black{\tilde\eta_{li}(\tilde{\omega},k)+}f_{li}\left(p^{{k-1}}_{1J}-p^{{k-1}}_{3J}\right)\right) \notag \\
	& \quad \black{+p^{{k-1}}_{iJ} \left(1-\sum_{l=1, \black{l \ne i}}^3\left(\black{\tilde\eta_{il}(\tilde{\omega},k)+}f_{il}\left(p^{{k-1}}_{1J}-p^{{k-1}}_{3J}\right)\right)\right)}.
 \end{align}}
\end{lemma}
\begin{proof}
From the definition of $Q$ in \eqref{eq:DefinitionQ} we get
\begin{align}
\mathbb{E}^Q[{p}^{k}]&=\int_{\tilde{\Omega}} \mathbb{E}^{\hat{P}^{\tilde{\omega}}}[p^{k}]d \tilde{Q}(\tilde{\omega})	\nonumber\\
&=\int_{\tilde{\Omega}} \Gamma^{k}(\tilde{\omega}, \mathbb{E}^{\hat{P}^{\tilde{\omega}}}[p^{{k-1}}]) d \tilde{Q}(\tilde{\omega}) \nonumber \\
&=\int_{\tilde{\Omega}} \Gamma^{k}(\tilde{\omega}, p^{{k-1}}) d \tilde{Q}(\tilde{\omega}) \nonumber \\
&=\mathbb{E}^{\tilde{Q}}\left[ \Gamma^{k}(\cdot, p^{{k-1}})\right] \label{eq:Representation}
\end{align}
by \black{P}oint 1. and 7. in Theorem \black{\ref{theorem:PropertiesDiscrete}}.
By \eqref{eq:pi3} \black{we have}
\black{
\begin{align}
	\mathbb{E}^Q\left[ p^{k}_1- p^{k}_3 \vert \mathcal{F}^{{k-1}} \right]&= \mathbb{E}^Q \left[ \sum_{j=1}^3 p^{k}_{1j}+ p^{k}_{1J} - \sum_{j=1}^3 p^{k}_{3j}-p^{k}_{3J} \vert \mathcal{F}^{{k-1}}\right]\nonumber \\
	&= \mathbb{E}^Q \left[  p^{k}_{1J} - p^{k}_{3J}\vert \mathcal{F}^{{k-1}} \right],\nonumber
\end{align}
as the agents immediately break up.}
\black{Then by} \eqref{eq:Representation} it follows
\begin{align}
	\mathbb{E}^Q\left[p^{k}_1- p^{k}_3 \vert \mathcal{F}^{{k-1}}\right]&=\mathbb{E}^{\tilde{Q}}\left[ \Gamma_{1J}^{k}(\cdot, p^{{k-1}})- \Gamma_{3J}^{k}(\cdot, p^{{k-1}})\vert \tilde{\mathcal{F}}^{{k-1}}\right], \label{eq:ConditionalExpectationLemma}
\end{align}
where we use that \black{$\Gamma_{ij}^{k}(\tilde{\omega},\hat{p})=0$ for all $i,j=1,2,3, \tilde{\omega} \in \tilde{\Omega}, \hat{p} \in \hat{\Delta}$ by \eqref{eq:DefinitionGamma}, as}{ $\xi_{k_1 l_1} = 1$ for any $k_1, l_1 =1,2,3$}. {In particular,  by \eqref{eq:x}, \eqref{eq:ExampleFunction1}-\eqref{eq:ExampleFunction3} and \eqref{eq:varsigmaili}-\eqref{eq:allindicessame}, in this setting $\Gamma_{t_k}^{iJ}$ is given by}
\black{\begin{align}
&\Gamma^{k}_{iJ}(\tilde{\omega}, p^{{k-1}})\nonumber \\
&=\left( 1- \theta(\tilde{\omega},k) \right) \tilde{p}_{iJ}^{{k}} + \theta(\tilde{\omega},k) \tilde{p}_{iJ}^{{k}}\tilde{p}_{iJ}^{{k}} \nonumber \\
& \quad +   \sum_{k_1 = 1, k_1 \ne i}^3  \left(\tilde\varsigma_{k_1 i i}(\tilde{\omega},k)+g_{k_1 i i}\left( \tilde{p}^{\black{k}}_{\black{1J}}-\tilde{p}^{\black{k}}_{\black{3J}}\right) + 1 - \tilde\varsigma_{i k_1 k_1}(\tilde{\omega},k)-g_{i k_1 k_1} \left( \tilde{p}_{1J}^{{k}}-\tilde{p}_{3J}^{{k}}\right)\right) \nonumber \\
&\quad \quad \cdot \tilde{p}_{iJ}^{{k}}\tilde{p}_{k_1J}^{{k}}\theta(\tilde{\omega},k),\label{eq:ApplicationGamma1}
\end{align}}
with \black{${\tilde{p}^{{k}}_{iJ}= \black{F_i^{\tilde \omega}(p^{k-1}})},$ where $\black{F_i{\tilde \omega}(p_{t_{k-1}})}$ is defined in  \eqref{eq:DefinitionF^1}.}
The result follows by putting together \eqref{eq:ConditionalExpectationLemma}, {\eqref{eq:ApplicationGamma1} and} \black{\eqref{eq:DefinitionF^1}}.
\end{proof}
{We now give an example where the functions $\black{(\tilde{\eta}_{ij})_{i,j=1,2,3}},\theta,\black{(\tilde{\varsigma}_{ijk})_{i,j,k=1,2,3}}$ \black{can be chosen to} guarantee the existence of a probability measure $Q$ of the form \eqref{eq:DefinitionQ} which is a martingale measure for $p^1-p^3$}. Proposition \ref{prop:MarketArbitrageFree} {then implies that such a \black{measure} $Q$ is an equivalent martingale measure for the market price $S$ of the asset, and consequently that the market model is arbitrage{-}free}.
\begin{example} \label{example:MarketArbitrageFree}
Let \begin{align}
\tilde{\Omega}:=  \prod_{k=1}^N \left(\tilde{\Omega}_k \times \tilde{\tilde{\Omega}}_k\right),
 \end{align}
 with $\tilde{\Omega}_k:=\lbrace \tilde{\omega}_1^k, \tilde{\omega}_2^k\rbrace$ for $k=1,...,N$ and \black{$\tilde{\tilde{\Omega}}_k:=\prod_{r=1}^3 \tilde{\tilde{\Omega}}_{k,r}$} where \black{$\tilde{\tilde{\Omega}}_{k,r}=\{\tilde{\tilde{\omega}}^k_{i_r}, \dots,  \tilde{\tilde{\omega}}^k_{i_{r+1}}\}$ with $i_1=1, i_2 = l_1, i_3=l_2, i_4 = l$ for finite numbers $1<l_1<l_2<l$, endowed with $\sigma$-algebras $\tilde{\mathcal{F}}^{\black{k}}$ and $\tilde{\tilde{\mathcal{F}}}^{\black{k}},$ respectively. Without mentioning any further, we assume that each $\sigma$-algebra is generated by the subsets of the corresponding space.  Moreover, we denote by $\omega, \tilde{\omega}^k$ elements of $\Omega, \tilde{\Omega}^k$ for $k=1,...,N$, respectively.}  
 Let the processes $\Theta$ and $F$ be defined on \black{$\prod_{k=1}^{N} \tilde{\tilde{\Omega}}_{k,1}$ and $\prod_{k=1}^{N} \tilde{\tilde{\Omega}}_{k,2}$}, respectively, and $\Lambda$ and $M$ be defined on \black{$\prod_{k=1}^N \tilde{\tilde{\Omega}}_{k,3}$}. \black{Introduce $P:=\tilde{P} \ltimes \hat{P}$  by choosing 
 \begin{align} \label{eq:ProducProbabilitiesPQ}
 	\tilde{P}:=\bigotimes_{k=1}^N \left(\tilde{P}_k \otimes \tilde{\tilde{P}}_k\right)   
 \end{align}
 with $\tilde{P}_k$ probability measure on $\tilde{\Omega}_k$ and  $\tilde{\tilde{P}}_k = \prod_{r=1}^3 \tilde{\tilde{P}}_{k,r}$ where $\tilde{\tilde{P}}_{k,1}$, $\tilde{\tilde{P}}_{k,2}$ and $\tilde{\tilde{P}}_{k,3}$ are probability measures on $\tilde{\tilde{\Omega}}_{k,1}$, $\tilde{\tilde{\Omega}}_{k,2}$ and $\tilde{\tilde{\Omega}}_{k,3}$, respectively, for any $k=1, \dots, N$, such that $\Theta$ and $F$ are martingales under $\prod_{k=1}^N\tilde{\tilde{P}}_{k,1}$ and $\prod_{k=1}^N\tilde{\tilde{P}}_{k,2}$, respectively}. This implies that $\Theta$ and $F$ are also martingales under $P$. \black{Moreover, we assume that}
 \begin{equation}\label{eq:qandpequiv}
 \tilde{P}_k(\tilde{\omega}_i^k)>0 \quad \text{ for any $i=1,2$, {$k=1,\dots,N$}}
 \end{equation}
  and $\tilde{\tilde{P}}_k(\tilde{\tilde{\omega}}_i^k)>0$ for any $i=1,...,l$,  {$k=1,\dots,N$}.

 We define 
\begin{align} \label{eq:DefinitionIntensitiesExample}
 	\theta(\tilde{\omega},k):=\theta_1(k) \textbf{1}_{\lbrace \tilde{\omega}^k=\tilde{\omega}^k_1\rbrace} + \theta_2(k) \textbf{1}_{\lbrace \tilde{\omega}^k=\tilde{\omega}^k_2\rbrace}
\end{align}
{and
\begin{equation} \label{eq:pushtowards3or1}
\tilde{\eta}_{ij}(\tilde \omega,k) = \begin{cases}
 \tilde{\eta}_{1}^{13}(k)\textbf{1}_{\lbrace \tilde{\omega}^k=\tilde{\omega}^k_1\rbrace} +  \tilde{\eta}_{2}^{13}(k)\textbf{1}_{\lbrace \tilde{\omega}^k=\tilde{\omega}^k_2\rbrace} \qquad &\text{ for $i=1$, $j=3$}, \\
 \tilde{\eta}_{1}^{31}(k)\textbf{1}_{\lbrace \tilde{\omega}^k=\tilde{\omega}^k_1\rbrace} +  \tilde{\eta}_{2}^{31}(k)\textbf{1}_{\lbrace \tilde{\omega}^k=\tilde{\omega}^k_2\rbrace} \qquad &\text{ for $i=3$, $j=1$}, \\
  \tilde{\eta}_{1}^{21}(k)\textbf{1}_{\lbrace \tilde{\omega}^k=\tilde{\omega}^k_1\rbrace} +  \tilde{\eta}_{2}^{21}(k)\textbf{1}_{\lbrace \tilde{\omega}^k=\tilde{\omega}^k_2\rbrace} \qquad &\text{ for $i=2$, $j=1$}, \\
 \tilde{\eta}_{1}^{23}(k)\textbf{1}_{\lbrace \tilde{\omega}^k=\tilde{\omega}^k_1\rbrace} +  \tilde{\eta}_{2}^{23}(k)\textbf{1}_{\lbrace \tilde{\omega}^k=\tilde{\omega}^k_2\rbrace} \qquad &\text{ for $i=2$, $j=3$}, \\
 0 \qquad &\text{ for all other indices},
\end{cases}
\end{equation}
for some functions $	\tilde{\eta}_{1}^{13}, \tilde{\eta}_{2}^{13},  \tilde{\eta}_{1}^{31}, \tilde{\eta}_{2}^{31}, \tilde{\eta}_{1}^{21}, \tilde{\eta}_{2}^{21}, \tilde{\eta}_{1}^{23}, \tilde{\eta}_{2}^{23} : \{1,\dots,N\} \to [0,1/2].$ The functions \eqref{eq:ExampleFunction3} are defined by
\begin{equation} \label{eq:pushtowards3or1aftermatch}
\tilde{\varsigma}_{ijl}(\tilde \omega,k) = \begin{cases}
 \tilde{\varsigma}_{1}^{13}(k)\textbf{1}_{\lbrace \tilde{\omega}^k=\tilde{\omega}^k_1\rbrace} +  \tilde{\varsigma}_{2}^{13}(k)\textbf{1}_{\lbrace \tilde{\omega}^k=\tilde{\omega}^k_2\rbrace} \qquad &\text{ for $i=1$, $j=3$, $l=3$}, \\
 \tilde{\varsigma}_{1}^{31}(k)\textbf{1}_{\lbrace \tilde{\omega}^k=\tilde{\omega}^k_1\rbrace} +  \tilde{\varsigma}_{2}^{31}(k)\textbf{1}_{\lbrace \tilde{\omega}^k=\tilde{\omega}^k_2\rbrace} \qquad &\text{ for $i=3$, $j=1$, $l=1$}, \\
 0 \qquad &\text{ for all other indices},
\end{cases}
\end{equation}
for some functions 
\begin{equation}\label{eq:functionslessonehalf}
\tilde{\varsigma}_{1}^{31}, \tilde{\varsigma}_{2}^{31}, \tilde{\varsigma}_{1}^{13}, \tilde{\varsigma}_{2}^{13} : \{1,\dots,N\} \to [0,1/2].
\end{equation}}

We assume that the functions $f_{ij}: \mathbb{R} \to [0,1]$, $i,j=1,2,3$, \black{$i \ne j$}, which appear in \eqref{eq:ExampleFunction1} satisfy
\begin{align}
	 f_{2,1}(x),f_{3,1}(x),\black{f_{3,2}}(x) \in (0,1/2], \quad f_{1,2}(x)=f_{2,3}(x)=f_{1,3}(x)=0,\label{eq:f11} \\
	 \black{f_{3,1}(x)+\black{f_{3,2}}(x) \le 1/2}, \label{eq:sumnomorethanonehalf}
\end{align}
for any $x>0$. 

Moreover, we assume that $g_{ijj}: \mathbb{R} \to [0,1]$ , $i,j=1,2,3$, \black{$i \ne j$}, appearing in \eqref{eq:ExampleFunction3} is such that
\black{\begin{align}
	g_{2 1 1}(x), g_{3 1 1}(x), g_{3 2 2}(x) & \in (0,1/2],  \label{eq:ConditionStrictlyPositive}\\
	g_{1 3 3}(x), g_{1 2 2}(x), g_{2 3 3}(x) & = 0 \label{eq:ConditionG1}
\end{align}}
for any $x>0$.

Finally, we assume that \eqref{eq:f11} and \eqref{eq:ConditionStrictlyPositive}-\eqref{eq:ConditionG1} hold switching the indices $1$ and $3$ when $x<0$.
\end{example}

\begin{remark}
In equations \eqref{eq:DefinitionIntensitiesExample}-\eqref{eq:functionslessonehalf} the probabilities governing type changes which are defined on $(\tilde \Omega, \tilde{\mathcal{F}})$ are identified only by two possible states at every time. Moreover, we only allow for random type changes before the matching to pessimistic and optimistic type. For post matchings type change, we assume that neutral investors cannot change their type. Such assumptions allow for less lengthy computations in the following. An extension to a more general case, see for example Section \ref{sec:simulations}, can be easily provided. In particular, when the probabilities are identified by an higher number of states, we have more degree of freedom for defining the measure $\tilde Q$. 

{Assumptions \eqref{eq:f11} and \eqref{eq:ConditionStrictlyPositive}-\eqref{eq:ConditionG1} are in line with the model in Section \ref{sec:tradingvolume}. In particular, under \eqref{eq:f11}, if there are more optimistic than pessimistic investors, agents can switch to more pessimistic views \black{only due to some exogenous, stochastic effects modeled by $\tilde{\eta}$. Furthermore, there is instead a strictly positive term, depending on the difference between optimistic and pessimistic investors, which increases the probability that pessimistic or neutral traders} switch to more optimistic forecasts. The reverse happens if there are more pessimistic than optimistic traders. \\
Similar considerations hold for the upgrade/downgrade probabilities in \eqref{eq:ExampleFunction3}.
Under the hypothesis \green{$p_1>p_3$}, there is always a strictly positive \black{term which increases the probability} of upgrade of type, see \eqref{eq:pushtowards3or1aftermatch}\black{,} whereas the probability downgrade of type \black{after matching is only given by a stochastic, exogenous effect represented by $\tilde{\varsigma}$}, see \eqref{eq:pushtowards3or1aftermatch}. Again, the reverse holds when \green{$p_1<p_3$}.}

Further bounds on conditions on the involved functions are necessary to guarantee that the type change probabilities remain in $(0,1)$.
\end{remark}

{The following lemma is a straight forward consequence of the construction of $\tilde \Omega$ and $\tilde P$, together with the definition of the processes $\Theta$, $F$, $\Lambda$ and $M$ in Example \ref{example:MarketArbitrageFree} and with \eqref{eq:qandpequiv}.}
{\begin{lemma}\label{lem:assumptionQ}
In the setting of Example \ref{example:MarketArbitrageFree}, \black{introduce the probability measure $Q:=\tilde{Q} \ltimes \hat{P}$ with 
 \begin{align} \label{eq:ProducProbabilitiesQ}
 	\tilde{Q}:=\bigotimes_{k=1}^N \left(\tilde{Q}_k \otimes \tilde{\tilde{P}}_k\right),   
 \end{align}
 where $\tilde{\tilde{P}}_k$ is defined in \eqref{eq:ProducProbabilitiesPQ} and $ \tilde{Q}_k(\tilde{\omega}_i^k)>0 \text{ for any $i=1,2$, {$k=1,\dots,N$}}$. Then $Q$ satisfies Assumption \ref{prop:MarketArbitrageFree} and is equivalent to $P$.}
\end{lemma}}
{We then get the following result.}
\begin{proposition}\label{prop:existencemeasure}
	In the setting of Example \ref{example:MarketArbitrageFree}, there exist functions \black{$\tilde{\eta}_{i}^{31},  \tilde{\eta}_{i}^{13}, \tilde{\eta}_{i}^{21}, \tilde{\eta}_{i}^{23},\theta_i,\tilde\varsigma^{31}_i,\tilde\varsigma^{13}_i$, $i=1,2,$ appearing in {\eqref{eq:DefinitionIntensitiesExample}-\eqref{eq:pushtowards3or1aftermatch}}} such that the market is arbitrage-free.
	\end{proposition}
\begin{proof}
	Let $\tilde{Q}$ be of the form \black{\eqref{eq:ProducProbabilitiesQ}}. \black{Moreover, define
\begin{align}
q(k):={\tilde Q}\left(\prod_{l=1}^{k-1}\tilde{\Omega}_{l}\times \tilde{\omega}^k_1 \times \prod_{l=k+1}^{N}\tilde{\Omega}_{l} \times \prod_{l=1}^{N}\tilde{\tilde{\Omega}}_{l} \right) =\tilde{Q}_k(\tilde{\omega}_1^k). \label{eq:DefinitionEquivalentMeasure}
\end{align}
for any $k=1,\dots,N$.}
	By Proposition \ref{prop:MarketArbitrageFree} and Lemma \ref{lem:assumptionQ}, we need to find $q(k)$ for $k \in \lbrace 1,...,N \rbrace$ such that
	\begin{align}
	& p_1^{{k-1}}-p_3^{{k-1}} \nonumber \\
	&=\mathbb{E}^Q\left[p^{k}_1- p^{k}_3 \vert \mathcal{F}^{{k-1}}\right]\nonumber \\
	&=\mathbb{E}^{\tilde{Q}}\left[ \Gamma_{1J}^{k}(\cdot, p^{{k-1}})-\Gamma_{3J}^{k}(\cdot, p^{{k-1}})\vert \tilde{\mathcal{F}}^{{k-1}}\right] \label{eq:ExplanationProposition1} \\
	&=\mathbb{E}^{\tilde{Q}_k}\left[ \Gamma_{1J}^{k}(\cdot, p^{{k-1}})-\Gamma_{3J}^{k}(\cdot, p^{{k-1}})\vert \tilde{\mathcal{F}}^{{k-1}}\right] \notag\\
 &=q(k) a_1 + (1-q(k)) a_2, \label{eq:ConditionConditionalExpectation}
\end{align}
where
\black{\small{\begin{align}
&a_i=\left( 1- \theta_i(k) \right)  \left[\black{F^1_{i}(p^{{k-1}})}-\black{F^{3}_{i}(p^{{k-1}})}\right] \notag \\
&\quad + \theta_i(k) \left(\black{F^{1}_{i}(p^{{k-1}})}\black{F^{1}_{i}(p^{{k-1}})} -\black{F^{3}_{i}(p^{{k-1}})}\black{F^{3}_{i}(p^{{k-1}})}\right) \nonumber \\
& \quad + \theta_i(k) \Big( \left[g_{2 1 1}\left( \black{F^1_i(p^{{k-1}})}-\black{F^{3}_i(p^{{k-1}})}\right) + \left(1 - g_{1 2 2}\left( \black{F^1_i(p^{{k-1}})} - \black{F^1_i(p^{{k-1}})}\right)\right) \right]\black{F^{1}_{i}(p^{{k-1}})}\black{F^{2}_{i}(p^{{k-1}})} \notag \\
& \qquad- \left[g_{2 3 3}\left( \black{F^1_i(p^{{k-1}})}-\black{F^{3}_i(p^{{k-1}})}\right) + \left(1 - g_{3 2 2}\left( \black{F^1_i(p^{{k-1}})} - \black{F^1_i(p^{{k-1}})}\right)\right) \right]\black{F^{3}_{i}(p^{{k-1}})}\black{F^{2}_{i}(p^{{k-1}})} \nonumber \\
& \quad \quad+ \left[\tilde{\varsigma}_{i}^{31}(k)+g_{3 1 1}\left( \black{F^1_{i}(p^{{k-1}})}-\black{F^{3}_{i}(p^{{k-1}})}\right) - \tilde{\varsigma}_{i}^{13}(k)-g_{1 1 3}\left( \black{F^1_{i}(p^{{k-1}})} -\black{F^3_{i}(p^{{k-1}})} \right)\right]\black{F^{1}_{i}(p^{{k-1}})}\black{F^{3}_{i}(p^{{k-1}})}\Big),\label{eq:ai}
\end{align}}
with
\begin{align}
	\black{F^1_{i}(p^{{k-1}})} &= p^{{k-1}}_{2J}\left(\black{\tilde\eta_{i}^{21}(k)+}f_{21}\left(p^{{k-1}}_{1J}-p^{{k-1}}_{3J}\right)\right)+ p^{{k-1}}_{3J}\left(\black{\tilde\eta_{i}^{31}(k)+}f_{31}\left(p^{{k-1}}_{1J}-p^{{k-1}}_{3J}\right)\right)\nonumber \\
	&\quad \black{+p^{{k-1}}_{1J} \left(1-f_{12}\left(p^{{k-1}}_{1J}-p^{{k-1}}_{3J}\right)-\tilde\eta_{i}^{13}(k)-f_{13}\left(p^{{k-1}}_{1J}-p^{{k-1}}_{3J}\right)\right)}\label{eq:F1i}
\end{align}
and
\begin{align}
	\black{F^3_{i}(p^{{k-1}})} &= p^{{k-1}}_{2J}\left(\black{\tilde\eta_{i}^{23}(k)+}f_{23}\left(p^{{k-1}}_{1J}-p^{{k-1}}_{3J}\right)\right)+ p^{{k-1}}_{1J}\left(\black{\tilde\eta_{i}^{13}(k)+}f_{13}\left(p^{{k-1}}_{1J}-p^{{k-1}}_{3J}\right)\right)\nonumber \\
	&\quad \black{+p^{{k-1}}_{3J} \left(1-f_{32}\left(p^{k-1}_{1J}-p^{{k-1}}_{3J}\right)-\tilde\eta_{i}^{31}(k)-f_{31}\left(p^{{k-1}}_{1J}-p^{{k-1}}_{3J}\right)\right)}\label{eq:F3i}
\end{align}
for $i=1,2$. \\ Note that \eqref{eq:ExplanationProposition1} comes from \eqref{eq:SimplificationWithImmediateBreakUp}, whereas \eqref{eq:ConditionConditionalExpectation} follows by  \eqref{eq:GammaExample} and  \eqref{eq:DefinitionIntensitiesExample}-\eqref{eq:pushtowards3or1aftermatch}}.

By \eqref{eq:ConditionConditionalExpectation} it follows that \black{$q(k)$ has to satisfy}
\begin{align}
	q(k)=\frac{p_{1}^{{k-1}}-p_{3}^{{k-1}}-a_2}{a_1-a_2}= \frac{a_2+p^{{k-1}}_3-p^{{k-1}}_1}{a_2-a_1},
\end{align}
and that 
\begin{align} \label{eq:FirstCondition}
	a_1< p^{{k-1}}_1-p^{{k-1}}_3  \quad \text{ and }\quad  a_2>p^{{k-1}}_1-p^{{k-1}}_3  
\end{align}
or
\begin{align} \label{eq:SecondCondition}
	a_1> p^{{k-1}}_1-p^{{k-1}}_3 \quad \text{ and }\quad a_2 < p^{{k-1}}_1-p^{{k-1}}_3 
\end{align}
to guarantee that $q(k) \in (0,1)$.
Without loss of generality we assume from now on that $p^{{k-1}}_1-p^{{k-1}}_3 >0$, as identical considerations hold with opposite sign if $p^{{k-1}}_1-p^{{k-1}}_3<0$. The goal is to find conditions on \black{$\tilde{\eta}_{i}^{31}(k),  \tilde{\eta}_{i}^{13}(k), \tilde{\eta}_{i}^{21}(k), \tilde{\eta}_{i}^{23}(k) ,\theta_i(k),\tilde\varsigma^{31}_i(k),\tilde\varsigma^{13}_i(k)$, $i=1,2$}, such that \eqref{eq:FirstCondition} holds. In order to do that, we focus on $a_2$. 

{We first derive} conditions which guarantee that $a_2>p^{{k-1}}_1-p^{{k-1}}_3$. \black{Choose $\tilde\eta^{31}_2(k)$ and $\tilde\eta^{13}_2(k)$ such that
\begin{equation}\label{eq:ratioeta}
\frac{\tilde\eta^{13}_2(k)}{\tilde\eta^{31}_2(k)}<\frac{p^{{k-1}}_{3}}{p^{{k-1}}_{1}},
\end{equation}
and (to simplify the following computations) $\tilde\eta^{21}_2(k)=\tilde\eta^{23}_2(k)=0$. Also fix $\tilde\varsigma^{31}_2(k) > \tilde\varsigma^{13}_2(k)$. Then by \eqref{eq:F1i}-\eqref{eq:F3i}, by assumption \eqref{eq:f11} and since $p^{{k-1}}_1-p^{{k-1}}_3 >0$, we have that
\small{\begin{align}\notag
F^1_{2}(p^{{k-1}})-F^3_{2}(p^{{k-1}})&=p^{{k-1}}_{2J}f_{21}\left(p^{{k-1}}_{1J}-p^{{k-1}}_{3J}\right)+ p^{{k-1}}_{3J}\left(\tilde\eta_{2}^{31}(k)+f_{31}\left(p^{{k-1}}_{1J}-p^{{k-1}}_{3J}\right)\right) +p^{{k-1}}_{1J} \left(1-\tilde\eta_{2}^{13}(k)\right) \notag \\
&\quad- p^{{k-1}}_{1J}\tilde\eta_{2}^{13}(k)-p^{{k-1}}_{3J} \left(1-f_{32}\left(p^{{k-1}}_{1J}-p^{{k-1}}_{3J}\right)-\tilde\eta_{2}^{31}(k)-f_{31}\left(p^{{k-1}}_{1J}-p^{{k-1}}_{3J}\right)\right)\label{eq:F1minusF2i} \\
& \ge p^{{k-1}}_{1J}\left(1-2\tilde\eta_{2}^{13}(k)\right)-p^{{k-1}}_{3J}\left(1-2\tilde\eta_{2}^{31}(k)\right) > p^{{k-1}}_{1J}-p^{{k-1}}_{3J},\label{eq:F1biggerF3} 
\end{align}}
where the last inequality follows from \eqref{eq:ratioeta}.\\
Similarly, by \eqref{eq:ai}, \eqref{eq:F1biggerF3}, by assumptions \eqref{eq:ConditionStrictlyPositive}-\eqref{eq:ConditionG1} and again since $p^{{k-1}}_1-p^{{k-1}}_3 >0$, we get
\begin{align}
a_2\ge&\left( 1- \theta_2(k) \right) \left [\black{F^1_{2}(p^{{k-1}})}-\black{F^{3}_{2}(p^{{k-1}})}\right] \notag \\
& \quad + \theta_2(k) \Big( F^{2}_{2}(p^{{k-1}})\left(F^{1}_{2}(p^{{k-1}})-F^{3}_{2}(p^{{k-1}}) \right)+ \left(\tilde{\varsigma}_{2}^{31}(k) - \tilde{\varsigma}_{2}^{13}(k)\right)\black{F^{1}_{2}(p^{{k-1}})}\black{F^{3}_{2}(p^{{k-1}})}\Big),\notag \\
>&\left( 1- \theta_2(k) \right)  [\black{F^1_{2}(p^{{k-1}})}-\black{F^{3}_{2}(p^{{k-1}})}] \notag.
\end{align}
}
 In order to guarantee that $a_2 > p^{{k-1}}_1- p^{{k-1}}_3,$ we then choose $\theta_2(k)$ such that
\begin{equation*}
	1- \theta_2(k)  > \frac{p^{{k-1}}_1-p^{{k-1}}_3}{\black{F^1_{2}(p^{{k-1}})}- \black{F^{3}_{2}(p^{{k-1}})}},
\end{equation*}
By \eqref{eq:F1biggerF3} this is possible if $\theta_2$ is small enough.\\
Next, we derive conditions which guarantee that $a_2< p^{{k-1}}-1-p^{{k-1}}_3$. \black{Choose 
\begin{equation}\label{eq:choicesfortildeeta}
\tilde\eta_{2}^{31}(k)=\tilde\eta_{2}^{21}(k)=0 \text{ and }\tilde\eta_{2}^{13}(k)=\tilde\eta_{2}^{23}(k)=1/2.
\end{equation} 
From \eqref{eq:F1i}-\eqref{eq:F3i} we derive
\begin{align}
	\black{F^1_{2}(p^{{k-1}})}-\black{F^{3}_{2}(p^{{k-1}})} &= p^{{k-1}}_{2J}f_{21}\left(p^{{k-1}}_{1J}-p^{{k-1}}_{3J}\right)+p^{{k-1}}_{3J}f_{31}\left(p^{{k-1}}_{1J}-p^{{k-1}}_{3J}\right)+ p^{{k-1}}_{1J}\left(1-\black{\frac{1}{2}}\right)\notag \\
	&\quad -\frac{1}{2}p^{{k-1}}_{1J}-\frac{1}{2}p^{{k-1}}_{2J}-p^{{k-1}}_{3J}\left(1- f_{31}\left(p^{{k-1}}_{1J}-p^{{k-1}}-{1J}\right)- f_{32}\left(p^{{k-1}}_{1J}-p^{{k-1}}_{1J}\right)\right)\notag \\
	&\le 0,  \label{eq:differenceFnegative}
\end{align} 
where the first equality and \eqref{eq:differenceFnegative} follow from \eqref{eq:choicesfortildeeta} and \eqref{eq:f11}-\eqref{eq:sumnomorethanonehalf}, respectively. Looking now at \eqref{eq:ai} with $i=2$, we can then choose $ \theta_2(k)$ small enough so that  $a_2 < p^{k-1}_1-p^{{k-1}}_3$.} This concludes the proof.
\end{proof}

\black{\begin{remark}
The parameters $\tilde{\eta}_{2}^{31}(k),  \tilde{\eta}_{2}^{13}(k), \tilde{\eta}_{2}^{21}(k), \tilde{\eta}_{2}^{23}(k) ,\theta_2(k),\tilde\varsigma^{31}_2(k),\tilde\varsigma^{13}_2(k)$ given in the proof of Proposition \ref{prop:existencemeasure} guarantee the existence of an equivalent martingale measure in the setting of Example \ref{example:MarketArbitrageFree}. Other choices are of course possible under which an equivalent martingale measure still exists: for example, one can show that \eqref{eq:choicesfortildeeta} can be relaxed. 
\end{remark}}

\subsection{Numerical simulations}\label{sec:simulations}
In this section we provide some numerical simulations of the model defined in Equations \eqref{eq:BubbleEvolution} and \eqref{eq:Xp}, where $p_1$ and $p_3$ are governed by the dynamical system introduced in Section \ref{sec:main} with matching and type change probabilities having dynamics given by \eqref{eq:ExampleFunction1}-\eqref{eq:ExampleFunction3}. In particular, we consider the following setting. We discretize the time interval $[0,T]$ in $N$ subintervals. The processes $\Lambda=(\Lambda^n)_{i=0,\dots,N}$ and $M=(M^n)_{n=0,\dots,N}$ appearing in \eqref{eq:DefMarketPriceProcess} are binomial models, defined on the probability space $(\tilde \Omega, \tilde{\mathcal{F}}, \tilde P)$, which approximate two geometric Brownian motions with drift equal to zero and volatilities $\sigma_{\Lambda}>0$ and $\sigma_M>0$. That is, $\Lambda^n:=Y^n \Lambda^{n-1}$, $n=1,\dots,N$ and constant $\Lambda^0>0$, where $Y^n$ is a random variable defined on $(\tilde \Omega, \tilde{\mathcal{F}})$
	such that 
	\begin{equation*}
	\tilde{P}(Y^n=u)=p, \quad 	\tilde{P}(Y^n=1/u)=1-p,
	\end{equation*}
where 
$$
u := e^{\sigma_{\Lambda} T/N}, \qquad p = \frac{1-d}{u-d},
$$
and $M$ is defined analogously. The functions $f_{ij}$ and $g_{ijj}$,  $i,j=1,2,3,$ and \black{$i \ne j$}, which appear in \eqref{eq:ExampleFunction1} and \eqref{eq:ExampleFunction3}, respectively, are defined by 
\begin{equation}\label{eq:functionssimulations}
	f_{ij}(x):=g_{ijj}(x):=\begin{cases}
	\frac{1}{3} (x^+)^{0.4}  \quad &\text{if $i=2,j=1$ or $i=3,j=2$ }  \\
	\frac{1}{3} (-x^-)^{0.4}  \quad &\text{if $i=1,j=2$ or $i=2,j=3$ }  \\
	\left(\frac{1}{3} (x^+)^{0.4}\right)^2  \quad &\text{if $i=3,j=1$}  \\
	\left(\frac{1}{3} (-x^-)^{0.4}\right)^2  \quad &\text{if $i=1,j=3$ }  
	\end{cases}
\end{equation}
for any $x\in \mathbb{R}$. Note that the choice of these functions are coherent with our framework, where investors might have a type upgrade or downgrade when $p_1-p_3>0$ or $p_1-p_3<0$, respectively. Moreover, since  $|p_1-p_3|\le 1$, a direct switch from optimistic to pessimistic views or vice versa is more unlikely to happen. Also note that, since the processes $\tilde\eta_{ij}$ and $\tilde\varsigma_{ilj}$, $i,j,l=1,2,3$, are bounded by $1/4$, \eqref{eq:functionssimulations} guarantees that the values in both \eqref{eq:ExampleFunction1} and \eqref{eq:ExampleFunction3} are bounded by $1/2$. In this way, the sum for fixed $i$ does not exceed $1$.
Furthermore, we have:
\begin{enumerate}
\item The process $\Theta=(\Theta^n)_{n=0,\dots,N}$ in \eqref{eq:Xp} is defined by $\Theta^n:= 2/\pi \arctan(Z_{\Theta}^n)$, where $Z_{\Theta}=(Z_{\Theta}^n)_{n=0,\dots,N}$ is a binomial model on $(\tilde \Omega, \tilde{\mathcal{F}}, \tilde P)$ approximating a geometric Brownian motion with drift equal to zero and volatility $\sigma_{\Theta}>0$. Note that this choice guarantees that the dynamics of $\Theta$ stay in $(0,1)$.
\item The process $\theta=(\theta^n)_{n=0,\dots,N}$ in \eqref{eq:ExampleFunction2} is defined by $\theta^n:= 2/\pi \arctan(Z_{\theta}^n)$, where $Z_{\theta}=(Z_{\theta}^n)_{n=0,\dots,N}$ is a binomial model on $(\tilde \Omega, \tilde{\mathcal{F}}, \tilde P)$ approximating a geometric Brownian motion with drift equal to zero and volatility $\sigma_{\theta}>0$.
\item For $i,j=1,2,3,$ the processes $\tilde\eta_{ij}=(\tilde\eta_{ij}^n)_{n=0,\dots,N}$ in \eqref{eq:ExampleFunction1} are defined by $\tilde\eta_{ij}^n:= 2/\pi \arctan(Z_{\eta,i,j}^n)$, where $Z_{\eta,i,j}=(Z_{\eta,i,j}^n)_{n=0,\dots,N}$ is a binomial model on $(\tilde \Omega, \tilde{\mathcal{F}}, \tilde P)$ approximating a geometric Brownian motion with drift equal to zero and volatility $\sigma_{\eta}>0$.
\item For $i,j=1,2,3,$ the processes $\tilde\varsigma_{ij}=(\tilde\varsigma_{ij}^n)_{n=0,\dots,N}$ in \eqref{eq:ExampleFunction3} are defined by $\tilde\varsigma_{ij}^n:= 2/\pi \arctan(Z_{\varsigma,i,j}^n)$, where $Z_{\varsigma,i,j}=(Z_{\varsigma,i,j}^n)_{n=0,\dots,N}$ is a binomial model on $(\tilde \Omega, \tilde{\mathcal{F}}, \tilde P)$ approximating a geometric Brownian motion with drift equal to zero and volatility $\sigma_{\varsigma}>0$.
 \end{enumerate}
The processes $Z_{\Theta}, Z_{\theta}, Z_{\eta,i,j}, Z_{\varsigma,i,j}$, $i,j=1,2,3$, introduced above, are all independent of each other. We choose parameters $N=100$, $T=1$, $p^0_1=p_2^0=p_3^0=1/3$, $\Lambda^0=M^0=1$, $\sigma_{\Lambda}=\sigma_M=0.3$, $\Theta^0=5$, $\sigma_\Theta=0.2$, $\eta^0_{ij}=\varsigma^0_{ijl}=0.2$, $i,j,l=1,2,3$, $\sigma_{\eta}=\sigma_{\varsigma}=0.4$, $\theta=0.5$, $\sigma_{\theta}=0.2$. We also set the reversion process in  \eqref{eq:BubbleEvolution} to be constant, i.e. specifically, $\kappa^n=0.01$ for all $n=1,\dots,N$. In Figure 1 we show some trajectories of  the process $(p^n_1-p_3^n)_{n = 0, \dots,N}$ on the right and the bubble process $\beta$ on the left. We decided to divide the trajectories in three couples of panels for the reader's convenience.
We note the following:
\begin{itemize}
\item In our model, bubbles can be negative: this is a difference with respect to the classical martingale theory of bubbles introduced by Cox and Hobson \cite{CoxHobson} and Loewenstein and Willard \cite{LoewensteinWillard} and mainly developed by Jarrow, Protter et al. \cite{JarrowProtter2009}, \cite{JarrowProtter2011},  \cite{JarrowProtter2007},   \cite{JarrowProtter2010}, \cite{JarrowKchiaProtter} and \cite{Biagini}, where the fundamental value is always smaller than the market value. However, this is not in contrast with absence of arbitrages and is in agreement with real markets, where asset prices might sometimes be underestimated, see for example \cite{filardo2011impact} and \cite{siegel2003asset}.
\item Bubbles increase and also burst with different speeds. In particular, our model allows both for hard landing (i.e. steep and fast decrease) and soft landing (i.e. soft and slow decrease) after the burst of the bubble, see \cite{becker2002soft} and \cite{wan2018prevention} for an analysis on hard and soft landing. 
\item Due to the presence of the functions $f_{ij}$ and $g_{ijj}$,  $i,j,l=1,2,3$, in our model investors may change type independently or after a match according to the current phase of the market, that is, if prices are rising, they have an higher probability to switch to more optimistic views, and vice versa. This drives both the ascending and the descending phase of the bubble. However, there are some exogenous factors, which we model on the space $(\tilde \Omega, \tilde{\mathcal{F}}, \tilde{P})$, that may influence the bubble as well, and may anticipate or delay the burst. In particular, they can make the increase of the bubble before the burst (respectively the decrease after the burst) more or less steep. 
\end{itemize}

\begin{figure}\label{fig:trajectories}
\centering
\begin{minipage}{.5\textwidth}
  \centering
  \includegraphics[scale=0.3]{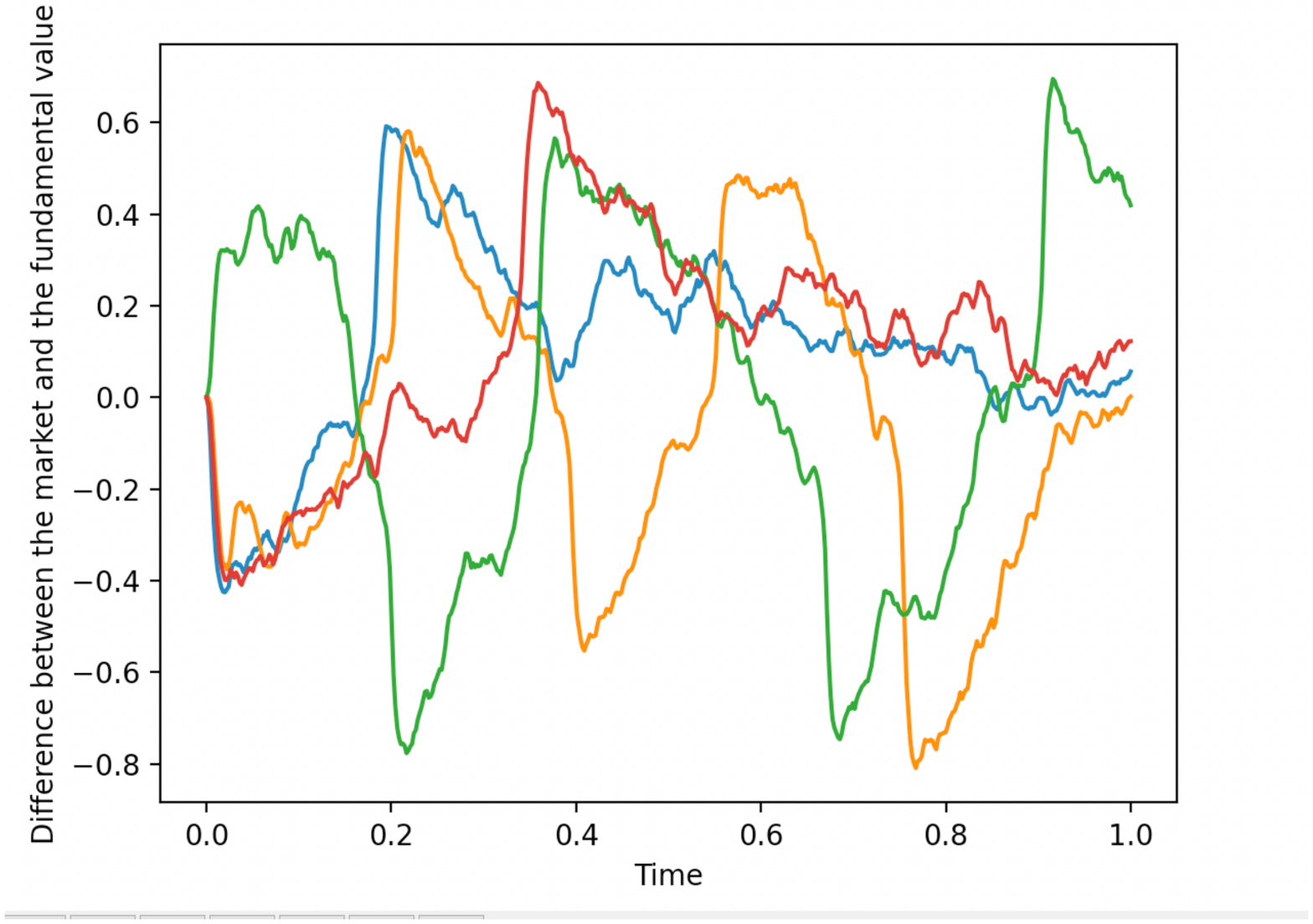}
\end{minipage}%
\begin{minipage}{.5\textwidth}
  \centering
  \includegraphics[scale=0.3]{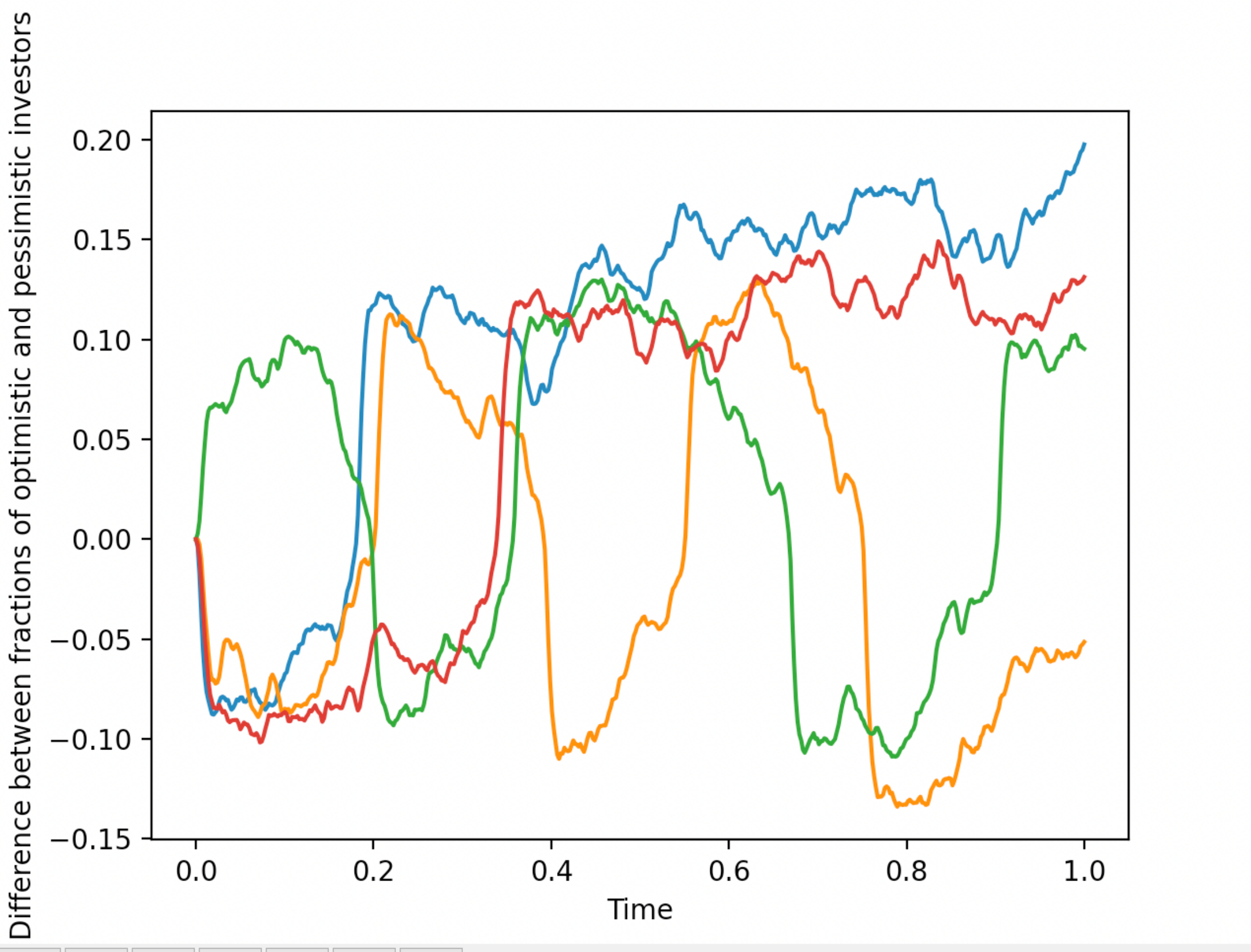}
\end{minipage}
\centering
\begin{minipage}{.5\textwidth}
  \centering
  \includegraphics[scale=0.3]{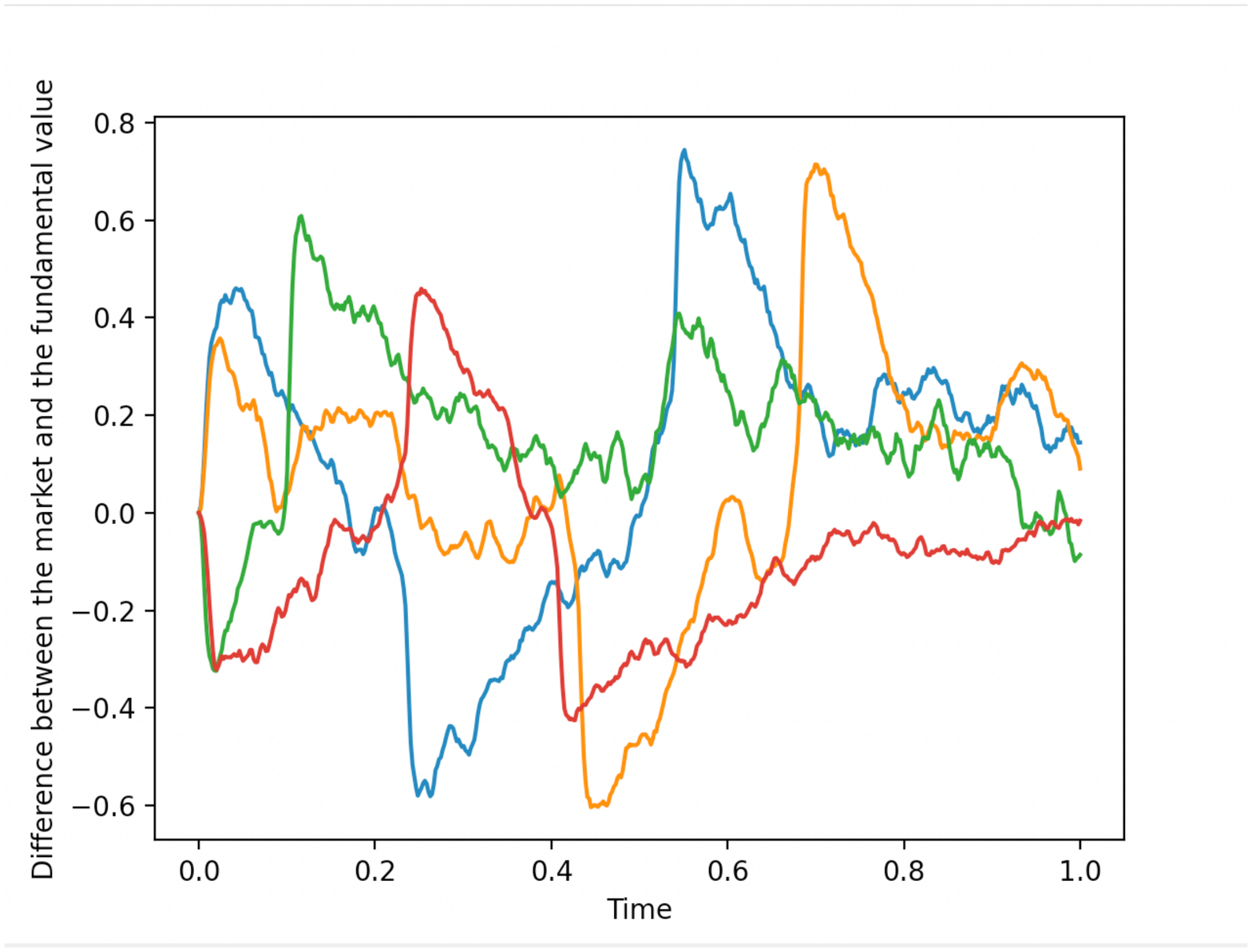}
\end{minipage}%
\begin{minipage}{.5\textwidth}
  \centering
  \includegraphics[scale=0.3]{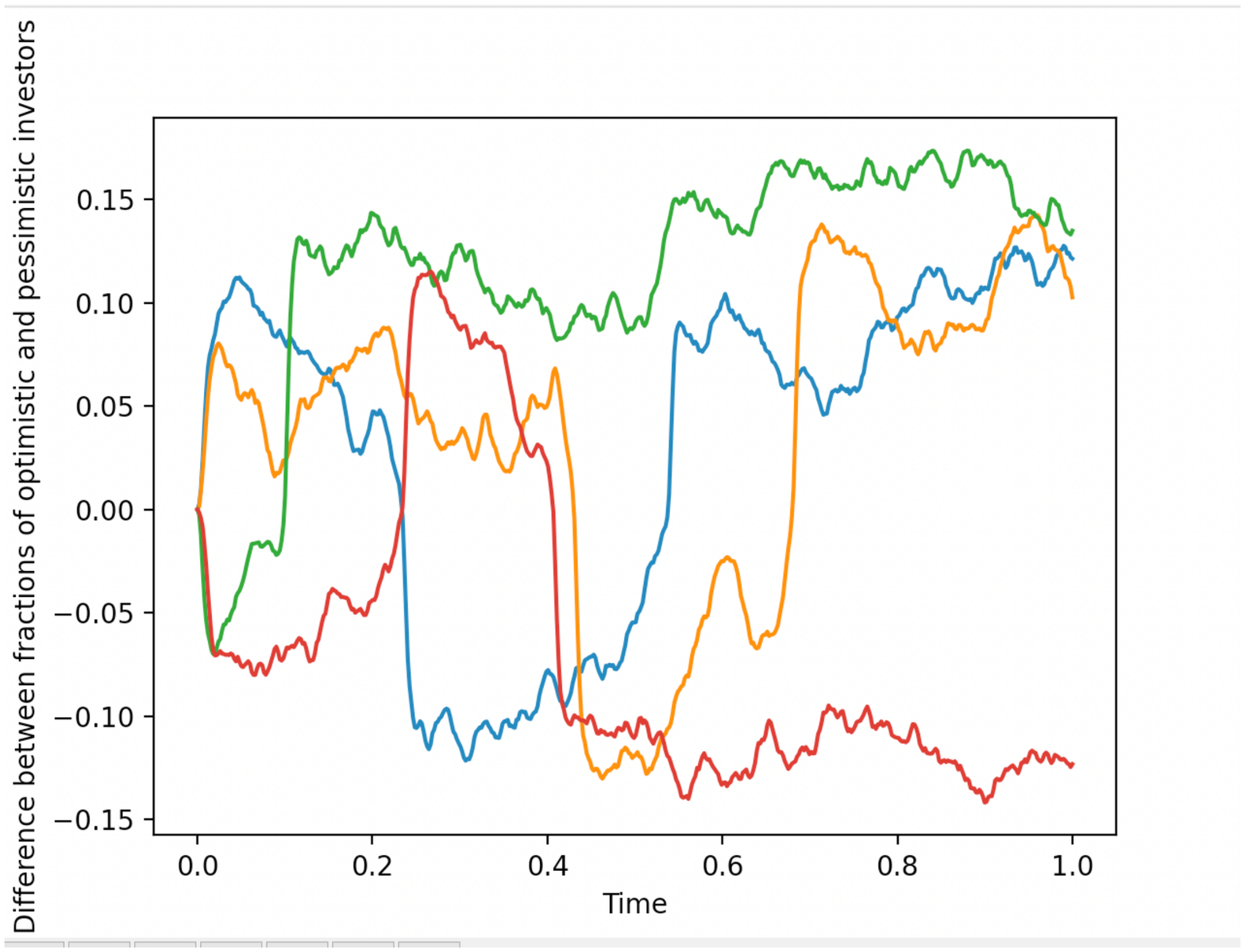}
\end{minipage}
\centering
\begin{minipage}{.5\textwidth}
  \centering
  \includegraphics[scale=0.3]{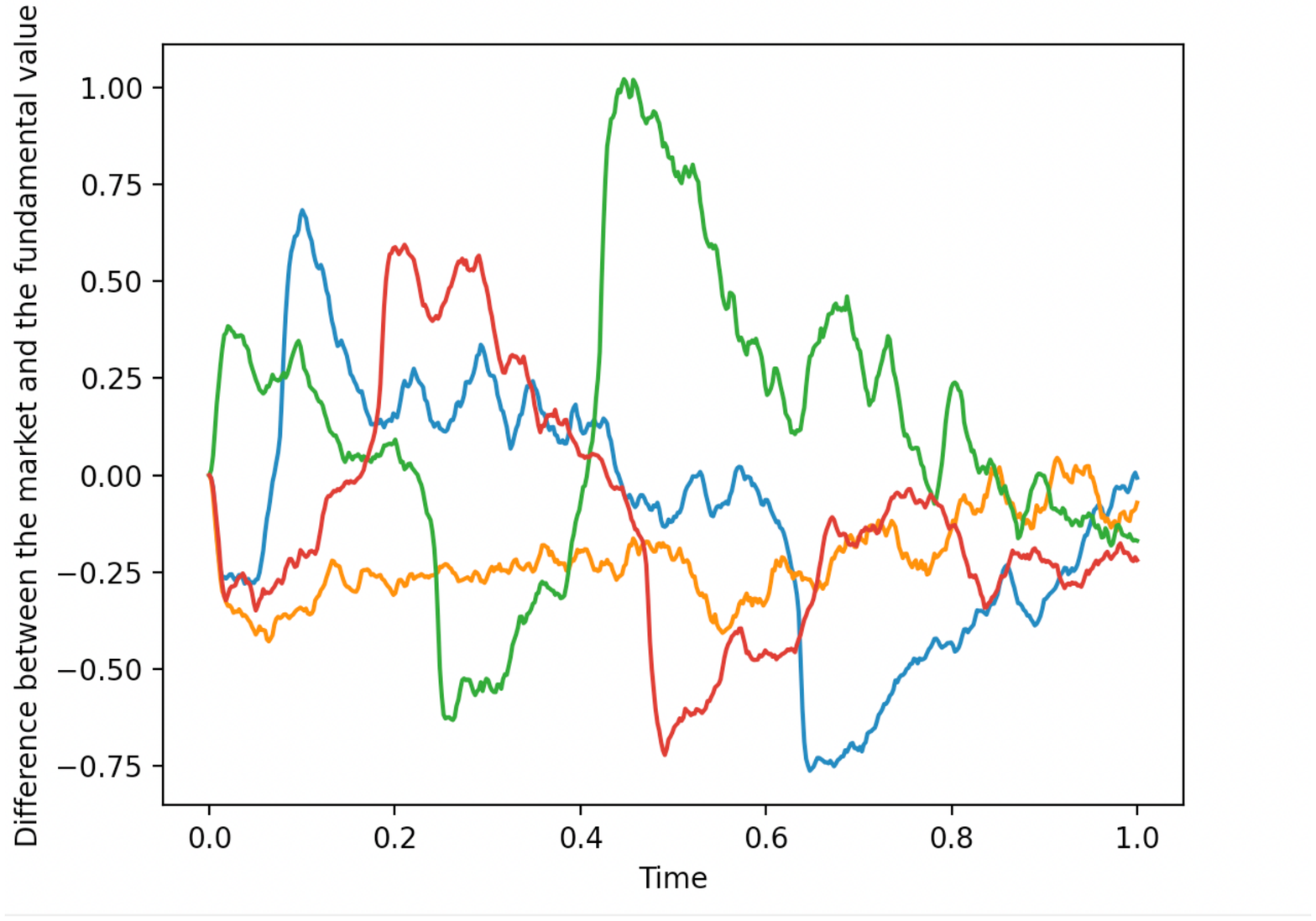}
\end{minipage}%
\begin{minipage}{.5\textwidth}
  \centering
  \includegraphics[scale=0.3]{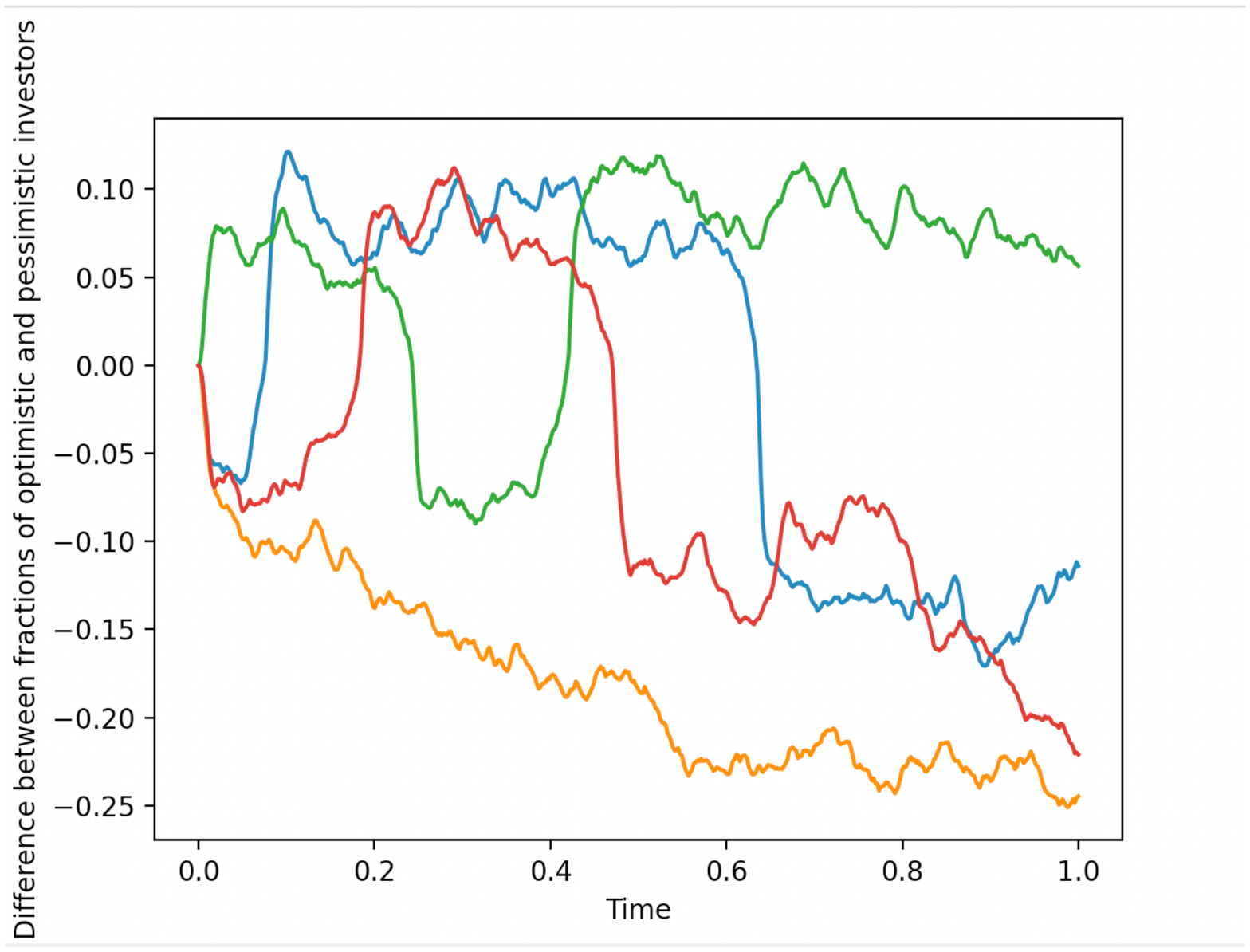}
\end{minipage}
\caption{Some trajectories of the process $(p^n_1-p_3^n)_{n = 0, \dots,N}$, in the right panels, and of the bubble in \eqref{eq:BubbleEvolution} in the left panels. For any couple of panels, a trajectory of a given color on the right drives the bubble trajectory of the same color on the left, together with the realizations of the processes $\Lambda$, $M$ and $\theta$. Parameters are:  $N=100$, $T=1$, $p^0_1=p_2^0=p_3^0=1/3$, $\Lambda^0=M^0=1$, $\sigma_{\Lambda}=\sigma_M=0.3$, $\Theta^0=5$, $\sigma_\Theta=0.2$, $\eta^0_{ij}=\varsigma^0_{ijl}=0.2$, $i,j,l=1,2,3$, $\sigma_{\eta}=\sigma_{\varsigma}=0.4$, $\theta=0.5$, $\sigma_{\theta}=0.2$, $\kappa=0.01$.}
\end{figure}

In Figure 2 we plot the function $t \to \frac{1}{n} \sum_{i=1}^n \beta_t^i$, where $\beta_t^i$ is the value for the $i$-th simulated trajectory of the bubble, $i=1,\dots,n$.  The left and right panels show the sample average for $n=100\,000$  and $n=1\,000\,000$, respectively.  

In Figure 3 we plot instead the function above for a sample of $1\,000\,000$ trajectories in the case when the values of the fractions of investors are $p^0_1=4/9, p_2^0=2/9, p_3^0=1/3$ \footnote{We let anyway the bubble start from zero here: one can assume a sudden jump of $p^1$ at initial time, or that $\theta=0$ before time $0$.}. In this case, we see that at the beginning, the bubble blows up on average, because of the actions of the functions $f_{ij}$ and $g_{ijj}$,  $i,j, l=1,2,3$. However, as the number of pessimistic investors changing their views starts to decrease, the bubble slows down, and then bursts on average because of the action of the mean reverting term $-\kappa \beta$ in  \eqref{eq:BubbleEvolution}.

Figure 3 clearly shows that, when the fraction of pessimistic and optimistic investors are different, the bubble is not a martingale under the measure for which we simulate the processes. However, we find a measure under which the expectation of the bubble at time $t_1$ is very close to its value at $t_0$ by tuning the parameters. In particular, this measure is identified by letting the binomial processes driving $\tilde\eta_{13}$ and $\tilde\varsigma_{133}$ increasing with probability $0.95$ and the ones driving $\tilde\eta_{31}$ and $\tilde\varsigma_{311}$ increasing with probability $0.1$, at the first time step. We see that the average value of $\beta^1$ is $0.1$ under the first measure and close to $10^{-5}$ under the new measure.
\begin{figure}\label{fig:averages}
\centering
\begin{minipage}{.5\textwidth}
  \centering
  \includegraphics[scale=0.3]{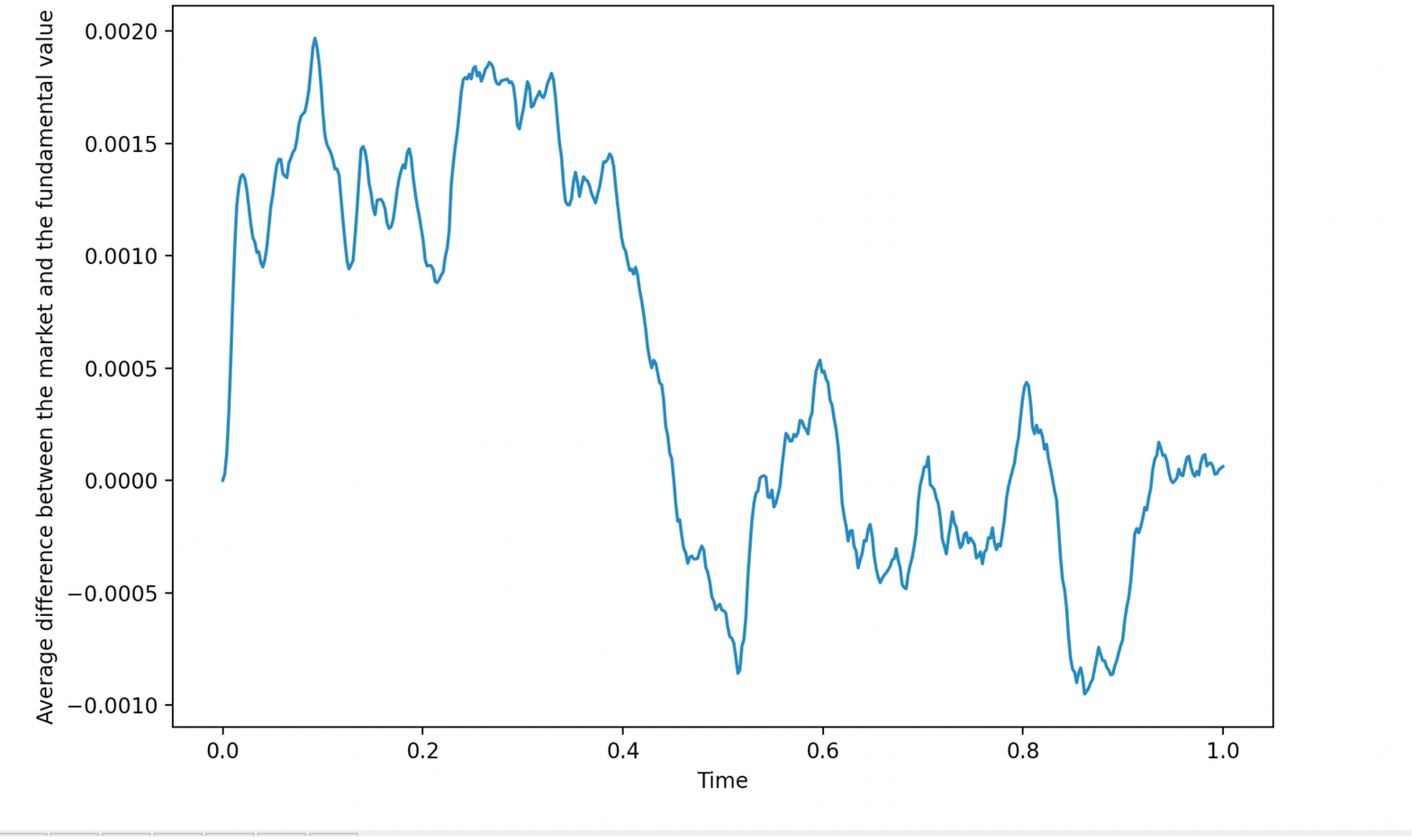}
\end{minipage}%
\begin{minipage}{.5\textwidth}
  \centering
  \includegraphics[scale=0.35]{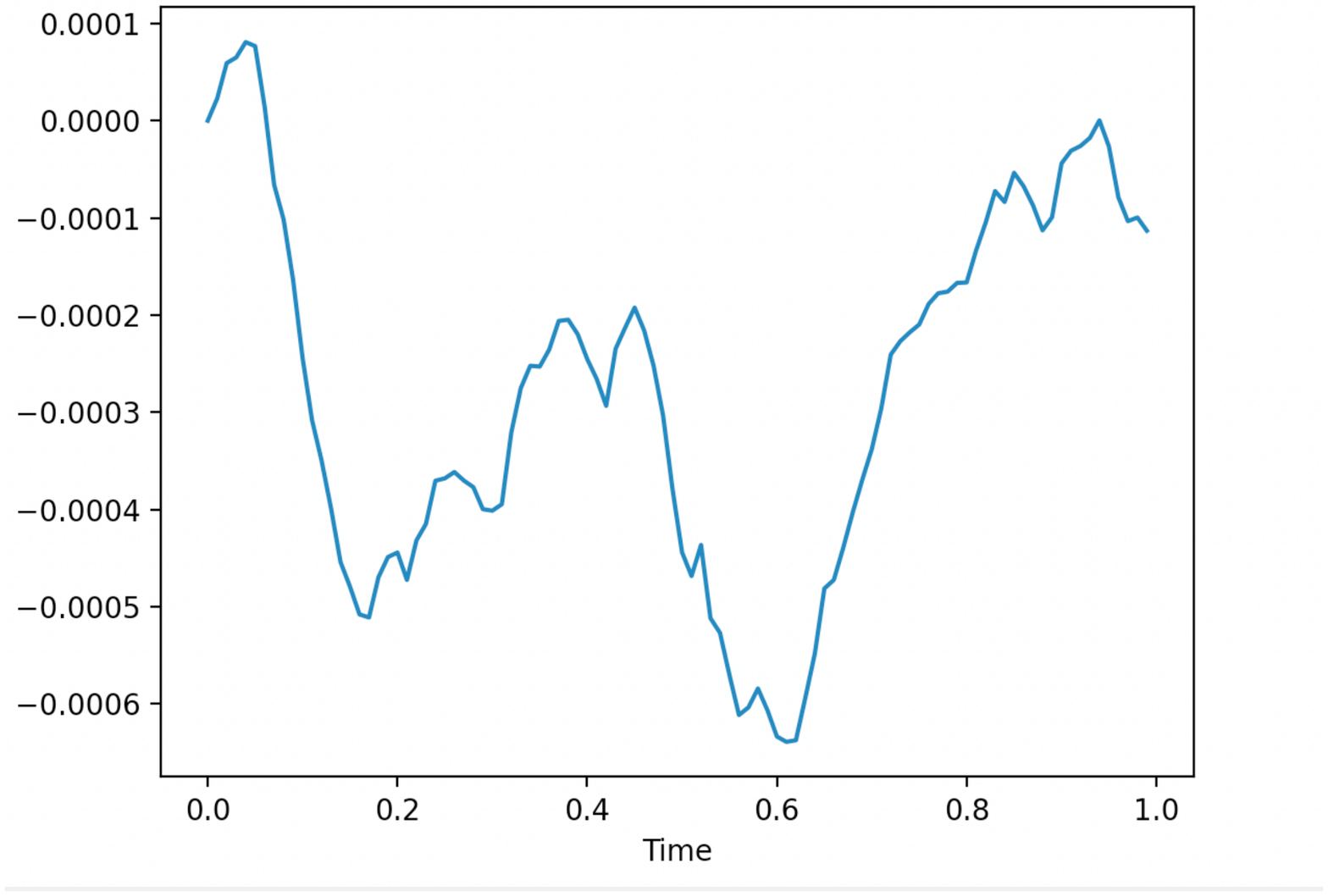}
\end{minipage}
\caption{Average of a sample of trajectories of the bubble process in \eqref{eq:BubbleEvolution}. The left and right panels show the sample average for $100\,000$  and $1\,000\,000$ simulations, respectively.   Parameters are:  $N=100$, $T=1$, $p^0_1=p_2^0=p_3^0=1/3$, $\Lambda^0=M^0=1$, $\sigma_{\Lambda}=\sigma_M=0.3$, $\Theta^0=5$, $\sigma_\Theta=0.2$, $\eta^0_{ij}=\varsigma^0_{ijl}=0.2$, $i,j,l=1,2,3$, $\sigma_{\eta}=\sigma_{\varsigma}=0.4$, $\theta=0.5$, $\sigma_{\theta}=0.2$, $\kappa=0.01$.}
\end{figure}

\begin{figure}\label{fig:averagesmoreoptimistic}
  \includegraphics[scale=0.5]{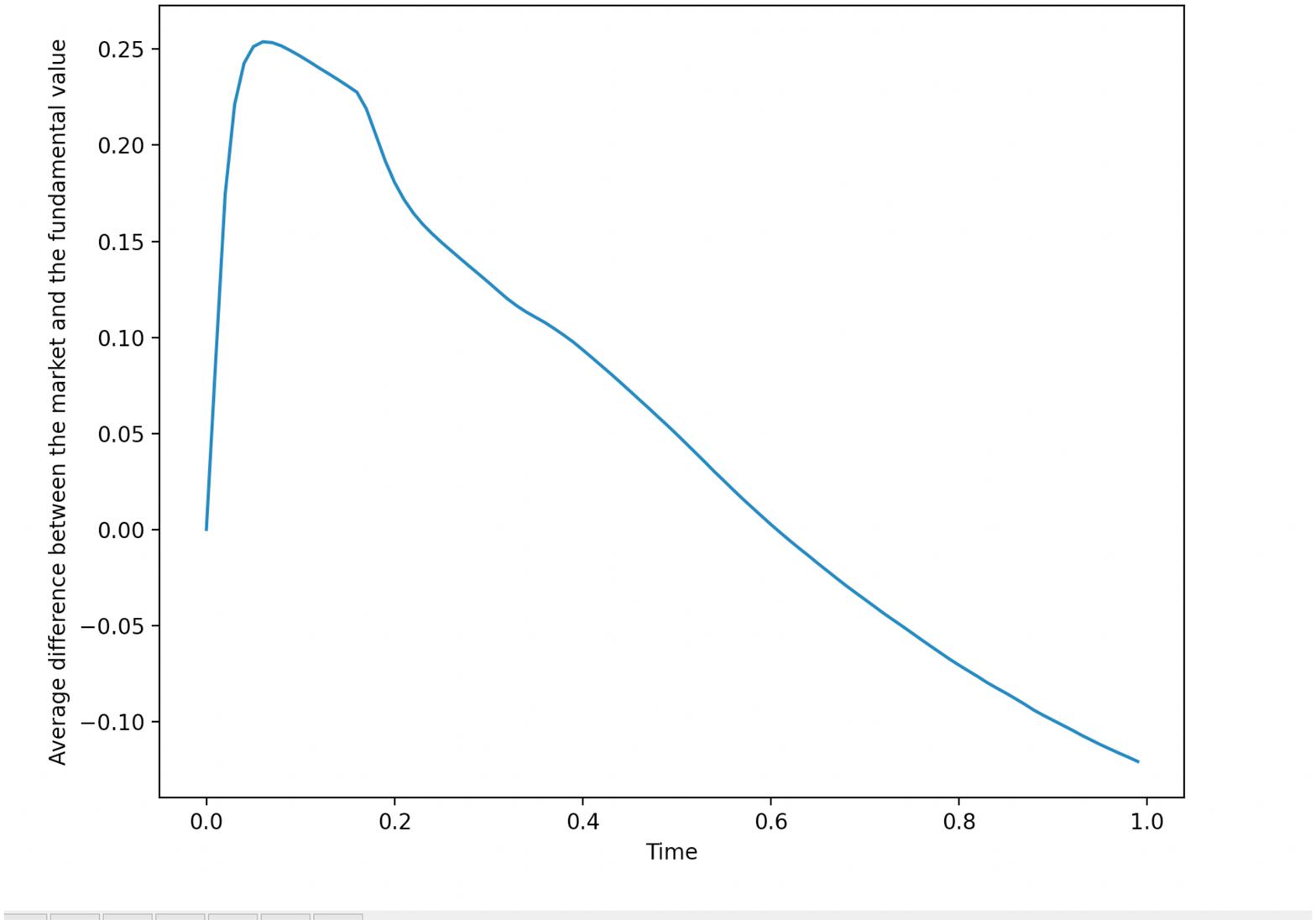}
\caption{Average of a sample of trajectories of the bubble process in \eqref{eq:BubbleEvolution} of size $1\,000\,000$.   Parameters are:  $N=100$, $T=1$, $p^0_1=4/9, p_2^0=2/9, p_3^0=1/3$, $\Lambda^0=M^0=1$, $\sigma_{\Lambda}=\sigma_M=0.3$, $\Theta^0=5$, $\sigma_\Theta=0.2$, $\eta^0_{ij}=\varsigma^0_{ijl}=0.2$, $i,j,l=1,2,3$, $\sigma_{\eta}=\sigma_{\varsigma}=0.4$, $\theta=0.5$, $\sigma_{\theta}=0.2$, $\kappa=0.01$.}
\end{figure}

\section{Conclusions}
We have modeled the formation of asset price bubbles by introducing a random matching mechanism among agents in a discrete time version of the model in \cite{JarrowProtter2012}. In order to do it, we extend results of \cite{RandomMatchingDiscrete} to a stochastic setting. In particular, via the introduction of a Markov kernel, we are able to construct the probability space where the asset price process is defined as the product of the space $\hat \Omega$ of the random matching and the space $\tilde \Omega$ of the factors which may influence the transition probabilities. This approach allows to isolate and model the self exciting mechanism governing the blow up of the bubble and the exogenous factors impacting the bursting phase of the bubble. In Section \ref{sec:simulations} we present numerical experiment showing how this approach is able to capture important behavioral features of asset price bubbles. 

\begin{appendix}
\section{Proof of Proposition \ref{OnlyMatchingDiscrete}}	\label{sec:AppendixProof}
For every fixed $\tilde{\omega} \in \tilde{\Omega}$ we construct the measure $\hat{P}^{\tilde{\omega}}$ as in the proof of Lemma 7 in \cite{RandomMatchingDiscrete}. Furthermore, as the construction of the space $\hat{\Omega}$ in Lemma 7 in \cite{RandomMatchingDiscrete} is independent of the input functions, we can also follow their approach. Then, the definitions of $\Omega, \mathcal{F}$ and $P_0$ in Points 1. and 2. in the Proposition allow us to finish the proof.
For the readers's convenience we show the proof in details in the following.
\begin{proof}
Let $(I, \cal{I}_0, \lambda_0)$ be the hyperfinite counting probability space with its Loeb space $(I, \cal{I}, \lambda)$. The proof consists of four steps.\\
\textbf{Step 1}: For each $k \in S$, $\tilde{\omega} \in \tilde \Omega$ and $\hat{p} \in \leftidx{^*}{\hat{\Delta}}$, let 
$$
b_k^{0}(\tilde{\omega})=b_k({\tilde{\omega},0,\hat{p}}):=1-\sum_{r \in S} \blue{\theta}_{kr}({\tilde{\omega},0,\hat{p}})
$$
and $I_k=\lbrace i \in I: \alpha^0(i)=k, \pi^0(i)=i \rbrace$. 

For each $i \in I_k$, $\tilde{\omega} \in \tilde{\Omega}$ and $\hat{p} \in \leftidx{^*}{\hat{\Delta}}$ define a probability $\zeta_i^{\hat{p},\tilde{\omega}}$ on $S \cup \lbrace J \rbrace $ such that 
$$
\zeta_i^{\hat{p},\tilde{\omega}}(l):=\theta_{kl}({\tilde{\omega},0,\hat{p}}) \text{ for } l \in S\text{ and }\zeta_i^{\hat{p},\tilde{\omega}}(J):=\delta_J(l) \text{ for }l \in S \cup \lbrace J \rbrace.
$$
 
 Let ${\hat{\Omega}}_0=(S \cup \lbrace J \rbrace)^{I}$ be the internal set of all the internal functions from $I$ to $S \cup \lbrace J \rbrace$. For any $\tilde{\omega} \in \tilde{\Omega}$ and $\hat{p} \in \leftidx{^*}{\hat{\Delta}}$, also let $\mu_0^{\hat{p},\tilde{\omega}}$ be the internal product probability measure $\prod_{i \in I}\zeta_i^{\hat{p},\tilde{\omega}}$ on $(\blue{\hat{\Omega}}_0, \mathcal{A}_0)$, where $\mathcal{A}_0$ is the internal power set of $\blue{\hat{\Omega}}_0$. For each fixed $\blue{\hat{\omega}_0} \in \blue{\hat{\Omega}}_0$ and $k,l \in S$, the agents in the set $\bar{A}_{kl}^{\blue{\hat{\omega}_0}}=\lbrace i \in I_k: \blue{\hat{\omega}_0}(i)=l \rbrace$ are now supposed to be matched with agents in $\bar{A}_{lk}^{\blue{\hat{\omega}_0}}$. 
 
\textbf{Step 2}: The issue now is that $\bar{A}_{lk}^{\blue{\hat{\omega}_0}}$ and $\bar{A}_{kl}^{\blue{\hat{\omega}_0}}$ might fail to have the same internal cardinality, for $k \neq l$, and $\bar{A}_{kk}^{\blue{\hat{\omega}_0}}$ may fail to have an even internal cardinality, which would allow an internal full matching on $\bar{A}_{kk}^{\blue{\hat{\omega}_0}}$. Scope of the second step of the proof is to fix such a problem. For $k,l \in S$ with $k \neq l$, let
\begin{equation*}
	C_{kl}^{\blue{\hat{\omega}_0}}=\lbrace A_{kl}: A_{kl} \subseteq \bar{A}_{kl}^{\blue{\hat{\omega}_0}}, A_{kl} \text{ is internal and }\vert A_{kl} \vert=\min \lbrace \vert \bar A_{kl}^{\blue{\hat{\omega}_0}}\vert, \vert \bar A_{lk}^{\blue{\hat{\omega}_0}}\vert \rbrace  \rbrace.
\end{equation*}
For any $k \in S$, let $C_{kk}^{\blue{\hat{\omega}_0}}$ be the family of the sets of the form $\bar{A}^{\blue{\hat{\omega}_0}}_{kk} \setminus \lbrace i \rbrace$ for $i \in \bar{A}^{\blue{\hat{\omega}_0}}_{kk}$ if $ \vert \bar{A}^{\blue{\hat{\omega}_0}}_{kk} \vert$ is odd, and $C_{kk}^{\blue{\hat{\omega}_0}}$ the set with one element $\bar{A}^{\blue{\hat{\omega}_0}}_{kk}$ if $ \vert \bar{A}^{\blue{\hat{\omega}_0}}_{kk} \vert$ is even. Set  $C^{\blue{\hat{\omega}_0}}:=\prod_{k,l \in S} C_{kl}^{\blue{\hat{\omega}_0}}$. Define an internal probability measure $\mu^{  \blue{\hat{\omega}_0}}$ on $C^{\blue{\hat{\omega}_0}}$ with internal power set $\cal{C}^{\omega_0}$ by letting $\mu^{ \blue{\hat{\omega}_0}}(\textbf{A})= \frac{1}{\vert C^{\blue{\hat{\omega}_0}} \vert}$ for $\textbf{A} \in C^{\blue{\hat{\omega}_0}}$. 

Let 
\begin{equation*}
	\blue{\hat{\Omega}}_1:=\lbrace (A_{kl})_{k,l \in S}: A_{kl} \subseteq I \text{ and } A_{kl} \text{ is internal, where } k,l \in S \rbrace.
\end{equation*}
The probability measure $\mu^{\blue{\hat{\omega}_0}}$ can be trivially extended to the common sample space $\blue{\hat{\Omega}}_1$ with its internal power set by letting $\mu^{\blue{\hat{\omega}_0}}(\textbf{A})=0$ for $\textbf{A} \in \blue{\hat{\Omega}}_1 \setminus C^{\blue{\hat{\omega}_0}}$. \\
Given the hyperfinite internal probability space $(\blue{\hat{\Omega}_0}, \mathcal{A}_0, \mu_0^{\hat{p},\tilde{\omega}})$ and the internal transition probability $\mu^{\blue{\hat{\omega}_0}}$, $\blue{\hat{\omega}_0} \in \blue{\hat{\Omega}}_0$, we can define the internal probability measure $\mu_1^{\hat{p},\tilde{\omega}}$ on $\blue{\hat{\Omega}_0 \times \hat{\Omega}_1}$ with its internal power set by letting $\mu_1^{\hat{p},\tilde{\omega}}(\blue{\hat{\omega}_0}, \textbf{A})=\mu_0^{\hat{p},\tilde{\omega}}(\blue{\hat{\omega}_0}) \times \mu^{\blue{\hat{\omega}_0}}(\textbf{A})$ for any $\blue{\hat{\omega}_0} \in \blue{\hat{\Omega}}_0$ and $\textbf{A} \in \blue{\hat{\Omega}_1}$. \\
\textbf{Step 3}: For any fixed $\blue{\hat{\omega}_0} \in \blue{\hat{\Omega}}_0$ and $\textbf{A}^{ \blue{\hat{\omega}_0}}=(A_{kl})_{k,l \in S} \in C^{\blue{\hat{\omega}_0} }$, we consider the internal partial matchings on $I$ that match agents from $A_{kl}$ to $A_{lk}$. Let $B_k^{ \blue{\hat{\omega}_0}}=I_k \setminus \left( \bigcup_{l \in S} A_{kl}^{ \blue{\hat{\omega}_0}} \right) $, which is the set of initially unmatched agents who remain unmatched. Let $\bar{B}_k^{ \blue{\hat{\omega}_0}}$ denote the set $\lbrace i \in I_k: \blue{\hat{\omega}_0}(i)=J \rbrace$; then it is clear that $B_k^{ \blue{\hat{\omega}_0}}= \bar{B}_k^{ \blue{\hat{\omega}_0}} \cup \bigcup_{l \in S }\left( \bar{A}_{kl}^{ \blue{\hat{\omega}_0}} \setminus \bar{A}_{kl}^{ \blue{\hat{\omega}_0}}\right)$. Let $B^{ \blue{\hat{\omega}_0}}= \bigcup_{k =1}^K B_k^{ \blue{\hat{\omega}_0}}$. For each $k \in S$, let $\blue{\hat{\Omega}}^{ \blue{\hat{\omega}_0}, \textbf{A}^{ \blue{\hat{\omega}_0}}}_{kk}$ be the internal set of all the internal full matchings on $A_{kk}^{ \blue{\hat{\omega}_0}}$. Let $\mu_{kk}^{ \blue{\hat{\omega}_0}, \textbf{A}^{ \blue{\hat{\omega}_0}}}$ be the internal counting probability measure on $\blue{\hat{\Omega}}^{ \blue{\hat{\omega}_0}, \textbf{A}^{ \blue{\hat{\omega}_0}}}_{kk}$. For $k,l \in S$ with $k<l$, let $\mu^{ \blue{\hat{\omega}_0}, {A}^{ \blue{\hat{\omega}_0}}}_{kl}$ be the internal set of all the internal bijections from $A_{kl}^{ \blue{\hat{\omega}_0}}$ to $A_{lk}^{ \blue{\hat{\omega}_0}}$. Let $\mu^{ \blue{\hat{\omega}_0}, \textbf{A}^{ \blue{\hat{\omega}_0}}}_{kl}$ be the internal counting probability on $A_{kl}^{ \blue{\hat{\omega}_0}}$. Let $\blue{\hat{\Omega}}_2$ be the internal set of all the internal partial matchings from $I$ to $I$. Define $\blue{\hat{\Omega}}_{2}^{ \blue{\hat{\omega}_0}, \textbf{A}^{ \blue{\hat{\omega}_0}}}$ to be the set of $\phi \in \blue{\hat{\Omega}}_2$, such that
\begin{enumerate}
	\item the restriction $\phi \vert_H= \pi^0\vert_H$, where $H$ is the set $\lbrace i: \pi^0(i) \neq i \rbrace$ of initially matched agents. 
	\item $\lbrace i \in I_k: \phi(i)=i \rbrace = B_{k}^{ \blue{\hat{\omega}}_0}$ for each $k \in S$. 
	\item The restriction $\phi \vert_{A_{kk}^{ \blue{\hat{\omega}}_0}} \in \blue{\hat{\Omega}}_{kk}^{ \blue{\hat{\omega}}_0, \textbf{A}^{ \blue{\hat{\omega}}_0}}$ for $k \in S$.
	\item For $k,l \in S$ with $k<l$, $\phi \vert_{A_{kl}^{ \blue{\hat{\omega}_0}}} \in \blue{\hat{\Omega}}_{kk}^{ \blue{\hat{\omega}_0}, \textbf{A}^{ \blue{\hat{\omega}_0}}}$.
	\end{enumerate}
	We now define an internal probability measure $\mu_2^{ \blue{\hat{\omega}_0}, \textbf{A}^{ \blue{\hat{\omega}_0}}}$ on $\blue{\hat{\Omega}}_2$ such that
	\begin{enumerate}
		\item for $\phi \in \blue{\hat{\Omega}}_2^{ \blue{\hat{\omega}_0}, \textbf{A}^{ \blue{\hat{\omega}_0}}}$,
		\begin{equation*}
			\mu_2^{ \blue{\hat{\omega}_0}, \textbf{A}^{ \blue{\hat{\omega}_0}}}(\phi)=\prod_{1 \leq k \leq l \leq K, A_{kl}^{ \blue{\hat{\omega}_0}} \neq \emptyset} \mu_{kl}^{ \blue{\hat{\omega}_0}, \textbf{A}^{ \blue{\hat{\omega}_0}}}(\phi \vert_{A_{kl}^{ \blue{\hat{\omega}_0}}}).
		\end{equation*}
		\item For $\phi \notin \blue{\hat{\Omega}}_2^{ \blue{\hat{\omega}_0}, \textbf{A}^{ \blue{\hat{\omega}_0}}}$, $\mu_2^{ \blue{\hat{\omega}_0}, \textbf{A}^{ \blue{\hat{\omega}_0}}}(\phi)=0$.
	\end{enumerate}
	The probability measure $\mu_2^{ \blue{\hat{\omega}_0}, \textbf{A}^{ \blue{\hat{\omega}_0}}}$ can be trivially extended to the sample space $\blue{\hat{\Omega}}_2$. \\
  For any $\tilde \omega \in \tilde \Omega$, define an internal probability measure $\blue{\hat{P}}_0^{\hat{p}}(\tilde{\omega})$ on $\blue{\hat{\Omega}=\hat{\Omega}_0 \times \hat{\Omega}_1 \times \hat{\Omega}_2}$ with the internal power set $\mathcal{\hat{F}}_0$ by letting 
  \begin{align} \label{eq:MarkovKernelProposition}
	\hat{P}_0^{\hat{p}}(\tilde{\omega})((\blue{\hat{\omega}}_0,\textbf{A}, \blue{\hat{\omega}}_2))=\begin{cases}
\mu_1^{\hat{p},\tilde{\omega}}(\blue{\hat{\omega}}_0, \textbf{A}) \times \mu_2^{\blue{\hat{\omega}}_0, \textbf{A}}(\blue{\hat{\omega}}_2) & \text{ if } \textbf{A} \in C^{\blue{\hat{\omega}}_0}\\
0 & \text{ otherwise. }
\end{cases}
\end{align}
The construction in \eqref{eq:MarkovKernelProposition} provides the Markov kernel from $\tilde \Omega$ to $\hat \Omega$ as in Point 2 of the Proposition. From now on, denote $\hat P_0^{\hat{p},\tilde{\omega}}:=\hat P_0^{\hat{p}}(\tilde \omega)$ for any $\tilde \omega \in \tilde \Omega$.

For $(i, \blue{\hat{\omega}}) \in I \times \blue{\hat{\Omega}}$, let $\hat{\pi}(i,(\blue{\hat{\omega}}_0, \textbf{A}, \blue{\hat{\omega}}_2))=\blue{\hat{\omega}}_2(i)$ and 
\[\hat{g}(i, \blue{\hat{\omega}})=\begin{cases}
\alpha^0(\hat{\pi}(i,\blue{\hat{\omega}})) & \text{ if } \hat{\pi}(i,\blue{\hat{\omega}}) \neq i\\
J& \text{ if } \hat{\pi}(i, \blue{\hat{\omega}}) = i.
\end{cases}\]
Denote the corresponding Loeb probability spaces of the internal probability spaces $(\blue{\hat{\Omega}}_0, \mathcal{\hat{F}}_0, \blue{\hat{P}}_0^{\hat{p}, \tilde{\omega}})$ and $(I \times\blue{\hat{\Omega}}_0, \cal{I}_0 \otimes \mathcal{\hat{F}}_0, \lambda_0 \otimes \blue{\hat{P}}_0^{\hat{p}, \tilde{\omega}})$ by $(\blue{\hat{\Omega}}, \mathcal{\hat{F}}, \blue{\hat{P}}^{\hat{p}, \tilde{\omega}})$ and $(I \times \blue{\hat{\Omega}}, \cal{I} \boxtimes \mathcal{\hat{F}}, \lambda \boxtimes \blue{\hat{P}}^{\hat{p}, \tilde{\omega}}),$ respectively. Set
$$
\bar{\Omega}=\lbrace (\blue{\hat{\omega}}_0, \textbf{A}, \blue{\hat{\omega}}_2) \in \blue{\hat{\Omega}}: \blue{\hat{\omega}}_0 \in \blue{\hat{\Omega}}_0, \textbf{A} \in C^{  \blue{\hat{\omega}}_0}, \blue{\hat{\omega}}_2 \in \Omega_2^{  \blue{\hat{\omega}}_0, \textbf{A}} \rbrace.
$$
Then by construction of $\blue{\hat{P}}^{\hat{p}, \tilde{\omega}}_0$, it is clear that $\blue{\hat{P}}_0^{\hat{p}, \tilde{\omega}}(\bar{\Omega})=1$. Moreover, $\hat{\pi}$ is by construction an internal matching and satisfies Point 4. of the proposition.
It is then possible to define $(\Omega, \mathcal{F}_0,P_0)$ as stated in Point 1. and 2. of the proposition and consider the corresponding Loeb probability space, see Point 6. Furthermore, we can extend $\hat{\pi}$ and $\hat{g}$ to $\Omega$ as stated in Point 3. and 5. \\
\textbf{Step 4}: We now prove {Points 5. and 6.} of the proposition. Define an internal process $\hat{f}$ from $I \times \blue{\hat{\Omega}}$ to $S \cup \lbrace J \rbrace$ such that for any $(i,\blue{\hat{\omega}}) \in I \times \blue{\hat{\Omega}}$ we have
\[\hat{f}(i, \blue{\hat{\omega}})=\begin{cases}
\blue{\hat{\omega}}^0(i) & \text{ if } \pi^0(i) = i\\
\alpha^0(\pi^0(i))& \text{ if } \pi^0(i) \neq i.
\end{cases}\]
Fix from now on $\hat{p} \in \leftidx{^*}{\hat{\Delta}}$ and $\tilde{\omega} \in \tilde{\Omega}$. It is clear that if $\alpha^0(i)=k$ and $\pi^0(i)=i$, then 
\begin{equation*}
	\blue{\hat{P}}^{\hat{p}, \tilde{\omega}}(\hat{f}_i=l) \simeq \blue{\hat{P}}^{\hat{p}, \tilde{\omega}}_0(\hat{f}_i=l)=\mu_0^{\hat{p}, \tilde{\omega}}(\blue{\hat{\omega}}_0(i)=l)= \zeta_i^{\hat{p}, \tilde{\omega}}(l)={\theta}_{kl}({\tilde{\omega},0,\hat{p}}),
\end{equation*}
which means that 
$$
\blue{\hat{P}}^{\hat{p}, \tilde{\omega}}(\hat{f}_i=l) = \leftidx{^{\circ}}{{\theta}_{kl}({\tilde{\omega},0,\hat{p}})}.
$$
With similar arguments it follows that 
$$
\blue{\hat{P}}^{\hat{p}, \tilde{\omega}}(\hat{f}_i=J)=\leftidx{^{\circ}}{{b}_{k}({\tilde{\omega},0,\hat{p}})}.
$$
 Moreover,  $\hat{f}_i$ and $\hat{f}_j$ are independent random variables on the sample space $(\blue{\hat{\Omega}}, \mathcal{\hat{F}}, \blue{\hat{P}}^{\hat{p}, \tilde{\omega}})$ for any $i \neq j$ in $I$. 
The exact law of large numbers as in Lemma 1 in \cite{RandomMatchingDiscrete} implies that, under the scenario of a current distribution $\hat p$ and of a realization $\tilde \omega \in \tilde \Omega$, it holds
$$
\lambda(\lbrace \alpha^0(i)=k, \pi^0(i)=i, \blue{\hat{\omega}}_0(i)=l \rbrace)=\leftidx{^{\circ}}{\hat{\rho}_{kJ}} \cdot \leftidx{^{\circ}}{{\theta}_{kl}}({\tilde{\omega},0,\hat{p}})
$$ and 
$$
\lambda(\lbrace \alpha^0(i)=k, \pi^0(i)=i, \blue{\hat{\omega}}_0(i)=J \rbrace)=\leftidx{^{\circ}}{\hat{\rho}_{kJ}} \cdot \leftidx{^{\circ}}{{b}_{k}}({\tilde{\omega},0,\hat{p}}),
$$
for $\blue{\hat{P}}^{\hat{p}, \tilde{\omega}}$-almost all $\omega=(\blue{\hat{\omega}}_0, \textbf{A}, \blue{\hat{\omega}}_2) \in \blue{\hat{\Omega}}$ and for any $k,l \in S$, 
which means that
\begin{equation} \label{eq:ExtensionStaticDiscreteLemma1}
	\frac{\vert \bar{A}_{kl}^{  \blue{\hat{\omega}}_0} \vert }{\hat{M}} \simeq \hat{\rho}_{kJ} {\theta}_{kl}({\tilde{\omega},0,\hat{p}}) \simeq \hat{\rho}_{lJ} {\theta}_{lk}({\tilde{\omega},0,\hat{p}}) \simeq  \frac{\vert \bar{A}_{lk}^{  \blue{\hat{\omega}}_0} \vert }{\hat{M}} \text{ and } \frac{\vert \bar{B}_{k}^{  \blue{\hat{\omega}}_0} \vert }{\hat{M}} \simeq \hat{\rho}_{kJ} {b}_{k}({\tilde{\omega},0,\hat{p}}).
\end{equation}

Let $\tilde{\Omega}^{\hat{p},\tilde{\omega}}$ be the set of $\blue{\hat{\omega}}=(\blue{\hat{\omega}}_0, \textbf{A}, \blue{\hat{\omega}}_2) \in \blue{\hat{\Omega}}$ such that \eqref{eq:ExtensionStaticDiscreteLemma1} holds. Then $\blue{\hat{P}}^{\hat{p},\tilde{\omega}}(\tilde{\Omega}^{\hat{p},\tilde{\omega}})=1$, and hence $\blue{\hat{P}}^{\hat{p},\tilde{\omega}}(\tilde{\Omega}^{\hat{p},\tilde{\omega}} \cap \blue{\Omega} )=1$. \\
Fix any $\blue{\hat{\omega}}=(\blue{\hat{\omega}}_0, \textbf{A}, \blue{\hat{\omega}}_2) \in \tilde{\Omega}^{\hat{p},\tilde{\omega}} \cap \blue{\Omega}$; then $\textbf{A}= \textbf{A}^{\blue{\hat{\omega}}_0}$ for some $ \textbf{A}^{\blue{\hat{\omega}}_0} \in C^{\blue{\hat{\omega}}_0}$, so $\blue{\hat{\omega}}_2 \in \blue{\hat{\Omega}}_2^{\blue{\hat{\omega}}_0, \textbf{A}^{\blue{\hat{\omega}}_0}}$. \\
For any $k \neq l \in S$ we have 
\small{
\begin{equation} \label{eq:ExtensionStaticDiscreteLemma2}
	\frac{\vert {A}_{kl}^{ \blue{\hat{\omega}}_0} \vert }{\hat{M}} = \min \left( \frac{\vert \bar{A}_{kl}^{\blue{\hat{\omega}}_0}\vert}{\hat{M}}, \frac{\vert \bar{A}_{lk}^{ \blue{\hat{\omega}}_0}\vert}{\hat{M}} \right) \simeq \hat{\rho}_{lJ} {\theta}_{lk}({\tilde{\omega},0,\hat{p}})=\frac{\vert \bar{A}_{kl}^{ \blue{\hat{\omega}}_0} \vert }{\hat{M}} \text{ and } \frac{\vert {A}_{kk}^{ \blue{\hat{\omega}}_0} \vert }{\hat{M}} \simeq \frac{\vert \bar{A}_{kk}^{ \blue{\hat{\omega}}_0} \vert }{\hat{M}} \simeq  \hat{\rho}_{lJ} {\theta}_{kk}({\tilde{\omega},0,\hat{p}})
\end{equation}}
which also implies that
\begin{equation*}
	\frac{\vert {B}_{k}^{ \blue{\hat{\omega}}_0} \vert }{\hat{M}} \simeq \hat{\rho}_{kJ} {b}_{k}({\tilde{\omega},0,\hat{p}}) \simeq \frac{\vert \bar{B}_{k}^{ \blue{\hat{\omega}}_0} \vert }{\hat{M}}. 
\end{equation*}
For any $i \in I_k, i \in A_{kl}^{ \blue{\hat{\omega}}_0}$ if and only if $\pi(\blue{\hat{\omega}}_0, \textbf{A}^{ \blue{\hat{\omega}}_0}, \blue{\hat{\omega}}_2)= \blue{\hat{\omega}}_2(i) \in A_{lk}^{ \blue{\hat{\omega}}_0}$; and $i \in B_k^{ \blue{\hat{\omega}}_0}$ if and only if $\pi(\blue{\hat{\omega}}_0, \textbf{A}^{ \blue{\hat{\omega}}_0}, \blue{\hat{\omega}}_2)=\blue{\hat{\omega}}_2(i)=J$. Hence, for fixed $\blue{\hat{\omega}} = (\blue{\hat{\omega}}_0, \textbf{A}^{ \blue{\hat{\omega}}_0}, \blue{\hat{\omega}}_2)$, and for any $k,l \in S$, we can see that if $i \in A_{kl}^{ \blue{\hat{\omega}}_0} \subseteq \bar{A}_{kl}^{\blue{\hat{\omega}}_0},$ then
 $$
 \hat{f}(i,\blue{\hat{\omega}})=\blue{\hat{\omega}}_0(i)=l=\alpha^0(\blue{\hat{\omega}}_2(i))=\hat{g}(i,\blue{\hat{\omega}}),
 $$
 and that if $i \in B_{k}^{ \blue{\hat{\omega}}_0} \subseteq \bar{B}_{k}^{\blue{\hat{\omega}}_0},$ then
 $$
   \hat{f}(i,\blue{\hat{\omega}})=\blue{\hat{\omega}}_0(i)=J=\alpha^0(\blue{\hat{\omega}}_2(i))=\hat{g}(i,\blue{\hat{\omega}}).
  $$
   For any $i \in I \setminus(\cup_{k \in S}I_k)$, that is, for any $i \in I$ such that $\pi^0 \neq i$, we have that 
   $$
    \hat{f}(i,\blue{\hat{\omega}})=\alpha^0(\pi^0(i))=\alpha^0(\pi(i,\omega))=\hat{g}(i,\blue{\hat{\omega}}).
   $$
    It is clear that 
    $$
    \lbrace i \in I:  \hat{f}(i,\blue{\hat{\omega}}) \neq \hat{g}(i,\blue{\hat{\omega}}) \rbrace \subseteq \bigcup_{l \in S} \left( \bar{A}_{kl}^{\blue{\hat{\omega}}_0} \setminus {A}_{kl}^{ \blue{\hat{\omega}}_0} \right),
    $$
    which has $\lambda$-measure zero by \eqref{eq:ExtensionStaticDiscreteLemma2}. By the fact that $\blue{\hat{P}}^{\hat{p}, \tilde{\omega}} \left( \tilde{\Omega}^{\hat{p}, \tilde{\omega}} \cap  \blue{\Omega}\right)=1$, we know that 
\begin{equation*}
	\lambda(i \in I: \hat{f}(i,\blue{\hat{\omega}})=\hat{g}(i,\blue{\hat{\omega}}))=1
\end{equation*}
for $\blue{\hat{P}}^{\hat{p}, \tilde{\omega}}$-almost all $\blue{\hat{\omega}} \in \blue{\hat{\Omega}}$.

Since the Loeb product space $(I \times \blue{\hat{\Omega}}, \cal{I} \boxtimes \mathcal{\hat{F}}, \lambda \boxtimes \blue{\hat{P}}^{\hat{p}, \tilde{\omega}})$ is a Fubini extension, the Fubini property implies that for $\lambda$-almost all $i \in I, \hat{g}(i, \blue{\hat{\omega}})$ is equal to $\hat{f}(i, \blue{\hat{\omega}})$ for $\blue{\hat{P}}^{\hat{p}, \tilde{\omega}}$-almost all $\blue{\hat{\omega}} \in \blue{\hat{\Omega}}$. Hence $g$ satisfies the second part of the Lemma. Let $\tilde{I}$ be an $\cal{I}$-measurable set with $\lambda(\tilde{I})=1$ such that for any $i \in \tilde{I}, \hat{g}_i(\blue{\hat{\omega}})=\hat{f}_i(\blue{\hat{\omega}})$ for $\hat{\blue{P}}^{\hat{p}, \tilde{\omega}}$-almost all $\blue{\hat{\omega}} \in \blue{\hat{\Omega}}$. Therefore, by the construction of $f$ we know that the collection of random variables $\lbrace \hat{f}_i \rbrace_{i \in \tilde{I} }$ is mutually independent in the sense that any finitely many random variables from that collection are mutually independent. This also implies Point 6. of the proposition.   
\end{proof}

\end{appendix}

\bibliography{2022_11_02_Liquidity_based_modeling_of_asset_price_bubles_via_random_matching.bib}
\bibliographystyle{plainnat}
\end{document}